\documentclass[floatfix,superscriptaddress,reprint,aps,longbibliography,nofootinbib,prx]{revtex4-2}
\usepackage[utf8]{inputenc}
\usepackage{amsmath}
\usepackage{amsfonts,dsfont}
\usepackage{amssymb,amsthm}
\usepackage{graphicx}
\usepackage[]{hyperref}

\hypersetup{
  colorlinks = true,
  urlcolor = Blue,
  linkcolor = red,
  citecolor = cyan
}

\usepackage[all]{hypcap}
\usepackage{url}
\newcommand{\ket}[1]{|#1\rangle}
\newcommand{\bra}[1]{\langle#1|}
\newcommand{\braket}[2]{\langle#1|#2\rangle}
\usepackage{textcomp, complexity}
\usepackage{caption,subcaption}
\usepackage[dvipsnames]{xcolor}
\usepackage{tikz}
\usepackage{physics}
\usepackage{tabularx,environ}
\usepackage{booktabs,multirow}
\usepackage{mathpazo}
\usepackage{qcircuit}
\usepackage{mathrsfs}

\let\oldnl\nl% Store \nl in \oldnl
\newcommand{\nonl}{\renewcommand{\nl}{\let\nl\oldnl}}% Remove line number for one line
\newcommand{\ccol}[2]{\color{#1}#2\color{black}}
\usepackage[capitalise,compress]{cleveref}
\crefname{section}{Section}{Section} % abbreviate Section
\crefname{subsection}{Section}{Section}

\newtheoremstyle{mystyle}%
{3pt}% Space above
{3pt}% Space below 
{\upshape}% Body font
{}% Indent amount
{\bfseries}% Theorem head font
{.}% Punctuation after theorem head
{.5em}% Space after theorem head
{}% Theorem head spec (can be left empty, meaning â€˜normalâ€™)

\theoremstyle{theorem}
\crefname{thm}{Theorem}{Theorems}
\Crefname{thm}{Theorem}{Theorems}
\newtheorem{thm}{Theorem}
\newtheorem{lemma}[thm]{Lemma}

\newtheorem{corll}[thm]{Corollary}
\newtheorem{conj}[thm]{Conjecture}
\crefname{corll}{Corollary}{Corollaries}
\theoremstyle{remark}

\theoremstyle{definition}
\newtheorem{defn}{Definition}
\theoremstyle{mystyle}
\newtheorem{todoaux}{To-Do}

\NewDocumentEnvironment{todo}{o}
 {\IfNoValueTF{#1}
   {\todoaux\addcontentsline{toc}{subsection}{\protect\numberline{\thesubsection}\ccol{red}{To-Do item}}}
   {\todoaux[#1]\addcontentsline{toc}{subsection}{\protect\numberline{\thesubsection}{\ccol{red}{To-Do item}} (#1)}}%
   \ignorespaces}
 {\endtodoaux}

\makeatletter

\newcolumntype{\expand}{}
\long\@namedef{NC@rewrite@\string\expand}{\expandafter\NC@find}

\NewEnviron{problem}[2][]{%
  \def\problem@arg{#1}%
  \def\problem@framed{framed}%
  \def\problem@lined{lined}%
  \def\problem@doublelined{doublelined}%
  \ifx\problem@arg\@empty%
    \def\problem@hline{}%
  \else%
    \ifx\problem@arg\problem@doublelined%
      \def\problem@hline{\hline\hline}%
    \else%
      \def\problem@hline{\hline}%
    \fi%
  \fi%
  \ifx\problem@arg\problem@framed%
    \def\problem@tablelayout{|>{\bfseries}lX|c}%
    \def\problem@title{\multicolumn{2}{|l|}{%
        \raisebox{-\fboxsep}{\textsc{#2}}%
      }}%
  \else
    \def\problem@tablelayout{>{\bfseries}lXc}%
    \def\problem@title{\multicolumn{2}{l}{%
        \raisebox{-\fboxsep}{\textsc{#2}}%
      }}%
  \fi%
\par
\noindent%
  \begin{tabularx}{\columnwidth}{\expand\problem@tablelayout}%
    \problem@hline%
    \problem@title\\[\fboxsep]%
    \BODY %\\\problem@hline%
  \end{tabularx}%
    \\
}
\makeatother
\newcounter{parentnumber}
\renewcommand{\exp}{\mathsf{exp}}
\renewcommand{\log}{{\llog}}
\newclass{\PQP}{PQP}
\newclass{\UQMA}{UQMA}
\newclass{\UMA}{UMA}
\newclass{\UQCMA}{UQCMA}
\newclass{\UNP}{UNP}
\newclass{\GQMA}{GQMA}
\newclass{\GQCMA}{GQCMA}
\newclass{\PGQMA}{PGQMA}
\newclass{\EGQMA}{EGQMA}
\newclass{\PGQCMA}{PGQCMA}
\newclass{\EGQCMA}{EGQCMA}
\newclass{\pQMA}{PreciseQMA}
\newclass{\pBQP}{PreciseBQP}
\newclass{\pQCMA}{PreciseQCMA}
\newclass{\postBQP}{postBQP}
\newclass{\postQMA}{postQMA}
\newclass{\postQCMA}{postQCMA}
\newclass{\pUQMA}{PreciseUQMA}
\newclass{\pGQMA}{PreciseGQMA}
\newclass{\pPGQMA}{PrecisePGQMA}
\newclass{\pPGQCMA}{PrecisePGQCMA}
\newclass{\pEGQMA}{PreciseEGQMA}
\newclass{\pEGQCMA}{PreciseEGQCMA}
\newclass{\pAPGQMA}{PreciseAPGQMA}
\newclass{\pAEGQMA}{PreciseAEGQMA}
\newclass{\sP}{\#P}

\usepackage[normalem]{ulem}

\begin{document}
\title{The importance of the spectral gap in estimating ground-state energies}
\author{Abhinav Deshpande}\email{abhinavd@caltech.edu} \thanks{Current affiliation: Institute for Quantum Information and Matter, California Institute of Technology, Pasadena, California 91125, USA}
\affiliation{Joint Center for Quantum Information and Computer Science and Joint Quantum Institute, NIST/University of Maryland, College Park, MD 20742, USA}
\author{Alexey V.\ Gorshkov}\email{gorshkov@umd.edu}
\affiliation{Joint Center for Quantum Information and Computer Science and Joint Quantum Institute, NIST/University of Maryland, College Park, MD 20742, USA}
\author{Bill Fefferman} \email{wjf@uchicago.edu}
\affiliation{Department of Computer Science, University of Chicago, Chicago, Illinois 60637, USA}

\begin{abstract}
The field of quantum Hamiltonian complexity lies at the intersection of quantum many-body physics and computational complexity theory, with deep implications to both fields.
The main object of study is the \textsc{LocalHamiltonian} problem, which is concerned with estimating the ground-state energy of a local Hamiltonian and is complete for the class $\QMA$, a quantum generalization of the class $\NP$.
A major challenge in the field is to understand the complexity of the \textsc{LocalHamiltonian} problem in more physically natural parameter regimes.
One crucial parameter in understanding the ground space of any Hamiltonian in many-body physics is the spectral gap, which is the difference between the smallest two eigenvalues.
Despite its importance in quantum many-body physics, the role played by the spectral gap in the complexity of the \textsc{LocalHamiltonian} is less well-understood.
In this work, we make progress on this question by considering the precise regime, in which one estimates the ground-state energy to within inverse exponential precision.
Computing ground-state energies precisely is a task that is important for quantum chemistry and quantum many-body physics.

In the setting of inverse-exponential precision, there is a surprising result \cite{Fefferman2018} that the complexity of \textsc{LocalHamiltonian} is magnified from $\QMA$ to $\PSPACE$, the class of problems solvable in polynomial space (but possibly exponential time).
We clarify the reason behind this boost in complexity.
Specifically, we show that the full complexity of the high precision case only comes about when the spectral gap is exponentially small.
As a consequence of the proof techniques developed to show our results, we uncover important implications for the representability and circuit complexity of ground states of local Hamiltonians, the theory of uniqueness of quantum witnesses, and techniques for the amplification of quantum witnesses in the presence of postselection.
\end{abstract}
\maketitle

\section{Introduction} \label{sec_conceptual}

Several exotic phenomena in our world occur only at very low temperatures, most notably, those occurring in condensed matter such as superconductivity, superfluidity and the fractional and integer quantum Hall effects.
Beyond these examples in condensed-matter physics, the low-energy physics of systems of several interacting particles is of interest in several fields such as particle physics, atomic, molecular, and optical physics, chemistry, and quantum computing.
Accordingly, finding effective descriptions of ground states of many-body Hamiltonians is a very natural and important task in physics.

Given the prevalence and importance of this task, a natural question is that of the computational difficulty of solving this task in naturally occurring situations.
Questions such as the hardness of solving a computational task belong to the domain of computational complexity theory.
A good proxy for the difficulty of obtaining ground-state descriptions is the difficulty of solving a weaker problem, namely that of computing ground-state energies of many-body Hamiltonians.
This question is studied in the domain known as ``Hamiltonian complexity'' (see e.g.\ Ref.~\cite{Gharibian2015}), an area of research at the intersection of quantum many-body physics and computational complexity theory. 

This area of research originated from Kitaev's result that the \textsc{LocalHamiltonian} problem, which is the problem of computing the ground-state energy of a local Hamiltonian, is \QMA-complete \cite{Kitaev2002} (we refer a reader unfamiliar with complexity-theoretic language to \cref{sec_prelims}).
The complexity class $\QMA$ is the quantum generalization of $\NP$.
Kitaev's result may be viewed as an analogue of the seminal Cook-Levin theorem \cite{Cook1971,Trakhtenbrot1984} in computer science, generalized to the setting of quantum constraint satisfaction problems.
Despite the tremendous amount of progress in understanding the power of local Hamiltonians, many important questions remain, such as whether the task remains hard under less-demanding notions of approximation error \cite{Aharonov2002,Aharonov2009a,Aharonov2013} and whether there exist short classical descriptions of ground states of local Hamiltonians (see e.g., Refs.~\cite{Aharonov2002,aaronsonkuperberg,feffermankimmel}), among others.

One important question about \textsc{LocalHamiltonian} is the role played by the spectral gap.
The spectral gap is a traditionally important quantity in the context of ground-state properties of any physical system and is defined as the difference between the smallest two eigenvalues of the Hamiltonian.
Many important families of Hamiltonians in physics have the ``gap property'', meaning that the spectral gap in the limit of large system size $n \to \infty$ is lower-bounded by a constant.
Important conjectures in physics are concerned with the existence of the gap property for certain Hamiltonians \cite{Haldane1983,Witten2002}, a problem that is known to be undecidable in general \cite{Cubitt2015}.
Furthermore, the existence of a spectral gap implies various tractability results for the ground states of Hamiltonians.
For instance, in one dimension, the gap property significantly restricts the entanglement structure of ground states through the area law of entanglement, implying efficient classical representations of the same \cite{Hastings2007}, and further, classically efficient algorithms to compute the ground-state energy \cite{Landau2015,Arad2017}.
It is not known whether these properties hold for higher dimensions.

Despite the physical importance of the spectral gap, its role in the 
context of the \textsc{LocalHamiltonian} problem itself is much less clear.
In particular, it is not known whether \textsc{LocalHamiltonian} is $\QMA$-complete in the presence of nontrivial lower bounds on the spectral gaps, even when the lower bound is $\Omega(1/\poly(n))$ \cite{Aharonov2008,Jain2012}.
Meanwhile, if the spectral gap is promised to be lower bounded by a constant, there are no-go results \cite{Gonzalez-Guillen2018,Crosson2018a} that rule out any \QMA-hardness proof that proceeds by generalizing the clock construction technique.
This technique underlies all known \QMA-completeness results, in analogy with the theory of \NP-completeness, where the Cook-Levin theorem plays a foundational role.
Therefore, Hamiltonians with any nontrivial lower bounds on the spectral gap can be less complex than the general case.

In this work, we take an initial step towards answering the question of the role played by the spectral gap in the \textsc{LocalHamiltonian} problem.
To do so, we study $\QMA$ in the precise setting, i.e.\ the class $\pQMA$.
In the precise setting, the completeness (the minimum probability of accepting a correct statement) and soundness (the maximum probability of accepting an incorrect statement) of the protocol are separated by a quantity called the promise gap that scales inverse-exponentially in the size of the input.
For the \textsc{LocalHamiltonian} problem, this translates to computing the ground-state energy to within inverse-exponential precision in the system size.

Computing ground-state energies to inverse-exponential precision is not an artificial task.
This task corresponds to computing polynomially many digits of the answer, which is very desirable in some cases \cite{Dalzell2019}.
Algorithms whose runtimes scale as $\polylog(1/\epsilon)$ for additive error $\epsilon$ can compute quantities to inverse-exponential precision in polynomial time, and such algorithms have been found for Hamiltonian simulation and linear systems \cite{Berry2014,Babbush2016,Childs2017}.
There are also situations where precise knowledge of the ground-state energy of a Hamiltonian is essential.
For example, in quantum chemistry, chemical reactivity rates depend on the Born-Oppenheimer potential-energy surface for the nuclei.
Each point on this surface is an electronic ground-state energy for a particular arrangement of the nuclei.
Small uncertainties in the ground-state energy can exponentially influence the calculated rate $k$ via Arrhenius's law $k \propto \exp[-\beta \Delta E]$, where $\Delta E$ is an energy barrier and $\beta$ the inverse temperature (see, e.g., Ref.~\cite{AguileraSammaritano2020}).
Another example is in condensed-matter physics, where algorithms such as the density matrix renormalization group (DMRG) routinely compute several digits of the ground-state energy (see, e.g., Ref.~\cite{Hubig2018}).
Precise knowledge of the ground-state energy can enable one to identify the locations of quantum phase transitions by identifying non-analyticities \cite{Sachdev2011}.
Interestingly, the class of Hamiltonians for which the energy can be precisely measured correspond to Hamiltonians that can be fast-forwarded \cite{Atia2017}.

Fefferman and Lin \cite{Fefferman2018} studied the complexity of the class $\pQMA$, and showed the mysterious result that it equals \PSPACE.
This is surprising since $\QMA \subseteq \PP$ \cite{Kitaev2000,Vyalyi2003,Marriott2005} (also see \cref{fig_classes} for reference), and an alternative characterization of the class $\PP$ is $\pBQP$, the precise analogue of $\BQP$ (the class of problems efficiently solvable on quantum computers).
Since $\pBQP$ can handle inverse-exponentially small promise gaps and contains $\QMA$, one might have expected that adding the modifier $\mathsf{Precise-}$ to $\QMA$ would not have changed the power of the class by much.

We provide an explanation for this seemingly unexpected boost in complexity from $\QMA$, which is a subset of $\PP$, to $\pQMA$, which equals $\PSPACE$\footnote{The class $\PSPACE$ contains $\PP$, and is believed to be unequal to and much larger than $\PP$: see \cref{fig_classes}.}.
Specifically, we find that in order for the precise version of \textsc{LocalHamiltonian}, i.e.\  \textsc{PreciseLocalHamiltonian}, to be $\PSPACE$-hard, the spectral gap of the Hamiltonian must necessarily shrink superpolynomially with the size of the system $n$ (measured by the number of qudits in the system).
We give strong evidence that if the spectral gap shrinks no faster than a polynomial in the system size, i.e.\ if the spectral gap is bounded by $\Omega(1/\poly)$, the complexity of the problem is strictly less powerful.
In particular, we show that this problem characterizes the complexity class $\PP$, which is a subset of $\PSPACE$ and is widely believed to be distinct from $\PSPACE$.
If the problem were \PSPACE-hard, the so-called counting hierarchy, defined as $\CH = \PP \cup \PP^\PP \cup \ldots$ \cite{Allender1993}, would collapse, which is considered an unlikely possibility.
Our results therefore bring out the importance of the spectral gap, a quantity not well understood so far in Hamiltonian complexity.

Another main result of ours concerns the existence of polynomial-size quantum circuits to prepare ground states of local Hamiltonians.
This is an important question that has implications in circuit-complexity of ground states of natural Hamiltonians and is directly related to whether natural Hamiltonians can be efficiently cooled down to zero temperature.
In complexity-theoretic language, the question may be phrased in terms of the power of classical versus quantum witnesses in Merlin-Arthur proof systems, or more formally, the so-called $\QMA$ vs.\ $\QCMA$ question.
The (in)equivalence of these classes is an important open question in quantum complexity theory and many-body physics, which has remained unsettled despite recent progress in the oracle setting (see e.g., Refs.~\cite{aaronsonkuperberg,feffermankimmel}).
The precise version of $\QCMA$, or $\pQCMA$, is known to be equal to $\NP^\PP$ (see e.g., Refs.~\cite{Gharibian2018,Morimae2017a}), indicating a separation between $\pQCMA$ and $\pQMA \ (=\PSPACE)$ unless the counting hierarchy collapses.
Interestingly, we show strong \emph{equivalence} results for the $\pQMA$ vs.\ $\pQCMA$ question in the presence of spectral gaps.

Our results and proof techniques we develop here also have consequences for other areas of quantum computation, complexity theory and many-body physics.
Our second main result mentioned earlier roughly says that in the precise regime, the promise of an inverse-polynomial lower bound on the spectral gap is equivalent to the promise that there exists a polynomial-size circuit to prepare the ground state.
This leads to an interesting conjecture we make in \cref{sec_discussion}, which could have a bearing on the performance of near-term quantum algorithms for quantum chemistry and on the circuit complexity of various low-energy states, which is an important question in many-body physics and gravitational and high-energy physics \cite{Jefferson2017,White2020}.
We obtain some additional evidence for the conjecture in \cref{sec_implications} by showing that some implications of the conjecture are correct.
Furthermore, our results can shed light on an attempt to give a quantum-inspired reproof \cite{Aharonov2017a,Green2019} of the celebrated $\IP=\PSPACE$ result \cite{Shamir1992} via interactive protocols for the class $\pQMA$.
Our results also allow us to rule out sufficiently strong error-reduction techniques for the class $\postQMA$.

This paper is structured as follows.
In the rest of \cref{sec_conceptual} we give an introduction to the basic notions of complexity theory used in this work (which an experienced reader may skip), state and refer to the main results, give a high-level overview of the proof techniques and their implications, and discuss the relation of our results to other work in the literature.
In \cref{sec_defs}, we give the definitions of some other complexity classes and define some new classes that appear in this work.
We also define natural problems complete for these classes.
We then formally state the results pertaining to the class $\PP$ in \cref{sec_pp} and $\PSPACE$ in \cref{sec_pspace}.
We also consider the complexity of related classes in \cref{sec_implications}, after which the Appendices have detailed proofs of our claims.

\subsection{Preliminaries} \label{sec_prelims}
Here, we give a very brief introduction to the complexity-theoretic definitions and terminology in this work.
The reader is referred to a textbook (e.g.\ Refs.~\cite{Arora2009,Sipser2012}) for a more pedagogical exposition.
We are generally concerned with decision problems, where the answer is either ``YES'' or ``NO''.
These problems can be cast as follows: given an instance $x$, the task is to decide if it belongs to the class of YES instances ($x \in A_\mathrm{yes}$), or to the class of NO instances ($x \in A_\mathrm{no}$).
In principle, there can be problems where certain instances (for example, ill-defined ones) belong neither to $A_\mathrm{yes}$ or $A_\mathrm{no}$.
In such cases, we either allow an algorithm to answer arbitrarily, or we supplant the problem with a promise that such instances never occur.
These are called \emph{promise problems}.

In complexity theory, one is typically interested in the resources taken to solve various classes of decision problems.
Further, one is interested in how the resource cost scales with the size of the problem to be solved, which is quantified in terms of the length of the input, often denoted $n$.
In this work, we use the notation $\poly(n)$ to denote any function that can be upper bounded by $O(n^c)$ for some constant $c=\Theta(1)$.
We also denote $\exp(n)$ to be any function $2^{\poly(n)}$.
We will omit the dependence on $n$, which in our work is taken to be the number of qudits.
The results in this work are applicable generally to qudits of any dimension $d \geq 2$, but we will often work with qubits in our proofs.

We first define the class $\BQP$ (Bounded-error Quantum Polynomial time), which is the class of problems solvable in polynomial time (in $n$) on a quantum computer with bounded error.
The error here is measured via the parameters $c$ (minimum probability of saying ``YES'' if the answer is YES) and $s$ (maximum probability of saying ``YES'' if the answer is NO).
More formally,
\begin{defn}[{$\BQP[c,s]$}]
$\BQP[c,s]$ is the class of promise problems $A=(A_\mathrm{yes},A_\mathrm{no})$ such that for every instance $x$, there is a uniformly generated circuit $U_x$ of size $\poly(n)$ acting on the state $\ket{0^{\otimes m}}$ for $m=\poly(n)$, with the property that upon measuring the first bit at the output, $o$, also called the decision qubit, we have\\
\begin{tabularx}{\linewidth}{l X}
If $x \in A_\mathrm{yes}$: & $ \Pr(o=1) \geq c$\\
If $x \in A_\mathrm{no}$: & $ \Pr(o=1) \leq s$.
\end{tabularx}
\end{defn}
In the above, we imagine that a quantum computer applies a circuit $U_x$ that acts on a standard initial state, measures the first bit at the output, and says YES (``accepts'') or NO (``rejects''), depending on whether the bit is measured to be in state $\ket{1}$ or $\ket{0}$.
The choice of the bit to measure at the output is arbitrary.
The term \emph{uniformly generated circuit} means that given an instance $x$ there is a polynomial-time classical algorithm to generate a description of the circuit $U_x$ to be applied.
\begin{defn}
 $\BQP = \cup_{c-s \geq 1/\poly} \BQP[c,s]$.
\end{defn}
The class $\BQP$ is the quantum generalization of class $\BPP$ (Bounded-Error Probabilistic Polynomial time), the class of problems solvable in polynomial time by a randomized classical computer.

We now come to the class $\QMA$ (Quantum Merlin Arthur), which is a quantum generalization of $\NP$.
We imagine two parties, Merlin (the prover) and Arthur (the verifier).
The prover would like to convince the verifier that a certain problem instance $x$ is a YES instance.
The prover, who is computationally unbounded, can supply any state $\ket{\psi}$ on $w=\poly(n)$ qubits to the verifier as a ``proof'' or ``witness''.
The verifier can apply any circuit of their choice acting on some $m$ qubits they possess and the witness state, and accept/reject based on the outcome of a decision bit.
$\QMA$ is the class of problems such that a YES answer can be reliably verified in this way and in case the answer is NO, no matter what state is sent by the (possibly cheating) prover, the verifier rejects with high probability.
Just like with $\BQP$, $\QMA$ is defined with respect to parameters $c$ and $s$, which are called \emph{completeness} and \emph{soundness}, respectively.

\begin{defn}[{$\QMA[c,s]$}]
$\QMA[c,s]$ is the class of problems $A=(A_\mathrm{yes},A_\mathrm{no})$ with the property that, for every instance $x$, there exists a uniformly generated circuit $U_x$ with the following properties:
$U_x$ is of size $\poly(n)$ and acts on an input state $\ket{0}^{\otimes m}$, together with a proof (or witness) state $\ket{\Psi}$ of size $w$ supplied by an arbitrarily powerful prover.
Both $m$ and $w$ are bounded by polynomials in $n$.
Upon measuring the decision qubit $o$ of the output register, the verifier accepts if $o=1$, and rejects otherwise.
We say $A=(A_\mathrm{yes},A_\mathrm{no})$ is a $\QMA[c,s]$ problem iff \\
\begin{tabularx}{\linewidth}{l X}
If $x \in A_\mathrm{yes}$: & $\exists \ \ket{\Psi}$ such that $\Pr(o = 1) \geq c$ \\
If $x \in A_\mathrm{no}$: & $\forall \ \ket{\Psi}$, $\Pr(o = 1) \leq s$.
\end{tabularx}
\end{defn}
\noindent $\QMA$ is defined as $\cup_{c-s \geq 1/\poly} \QMA[c,s]$.

To characterize the complexity of a problem, we give ``upper'' and ``lower'' bounds on the complexity of the problem.
Upper bounds are statements of the form ``\textsc{X} $\in \mathsf{Y}$'', which means that the problem \textsc{X} can be solved with access to a solver for the complexity class $\mathsf{Y}$.
For example, Shor \cite{Shor1997} proved that \textsc{Factoring} $\in \BQP$, which means that quantum computers can factor integers in polynomial time (since quantum computers may be viewed as ``solvers for the class $\BQP$'').
Lower bounds are statements of the form ``\textsc{X} is $\mathsf{Y}$-hard''.
This means that problem \textsc{X} is as hard as any problem in $\mathsf{Y}$.
Such statements are often shown via \emph{reductions}.
One assumes the existence of an \emph{oracle}, a black box that can solve any instance of the problem \textsc{X} in one timestep.
A reduction is a mapping from the complexity class $\mathsf{Y}$ to the problem \textsc{X} with the property that any problem in $\mathsf{Y}$ can be solved by querying the oracle for \textsc{X}.
If such a reduction exists, it implies that the problem \textsc{X} is at least as hard as any problem in the class $\mathsf{Y}$.
If a problem \textsc{X} is both in the class $\mathsf{Y}$ and is $\mathsf{Y}$-hard, then it means that the upper and lower bounds to the problem match.
This means that the problem \textsc{X} is the hardest in the class it belongs to, namely $\mathsf{Y}$.
In this case, we say ``\textsc{X} is $\mathsf{Y}$-complete.'' or ``\textsc{X} is complete for $\mathsf{Y}$.''.
We also denote by $\mathsf{Y}^{\mathsf{Z}}$ the class of problems solvable by a $\mathsf{Y}$ machine with access to an oracle for any problem in $\mathsf{Z}$.

\begin{figure}
\includegraphics[width=\linewidth]{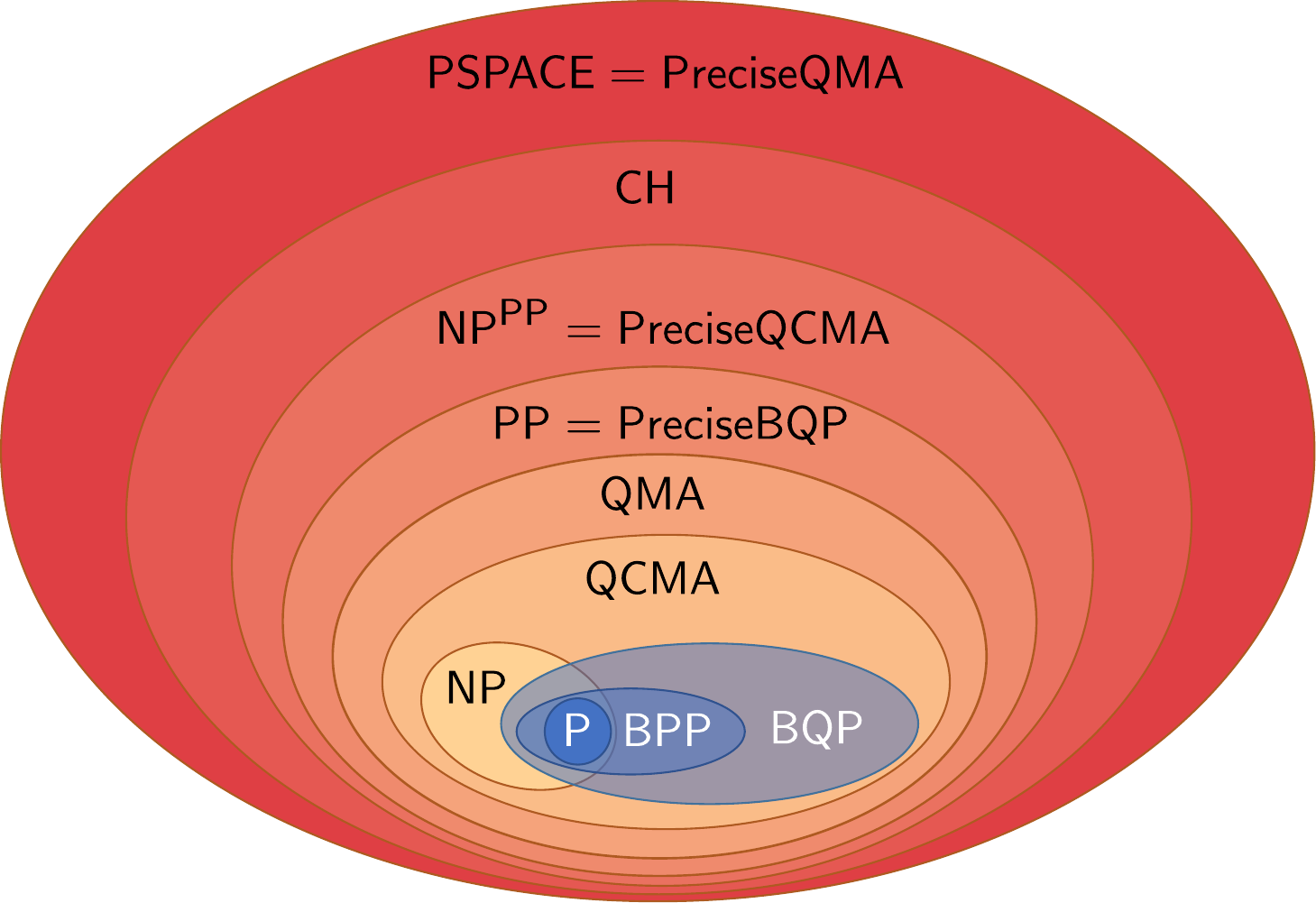}
\caption{\raggedright Major complexity classes featuring in this work.
The classes $\PSPACE$, $\NP^\PP$, and $\PP$ can be defined purely in terms of quantum computation, and are equal to $\pQMA$, $\pQCMA$, and $\pBQP$, respectively.
All inclusions except $\P \subseteq \BPP$ are believed to be strict.} \label{fig_classes}
\end{figure}

Lastly, we depict the known inclusions between complexity classes in \cref{fig_classes}.
We also describe here the classes not mentioned so far.
$\P$ is the class of problems efficiently solvable on classical computers, while $\NP$ is the class of problems for which a YES answer may be verified efficiently, via a protocol involving a classical prover and classical verifier.
The class $\QCMA$ is analogous to $\QMA$, except that the prover sends a classical witness instead of a quantum one.
As for $\PP$, it suffices to know that it equals $\pBQP$, a precise version of $\BQP$.
The class $\NP^\PP$ is a subset of $\PP^\PP$, since $\NP\subseteq \PP$.
These classes belong to the counting hierarchy ($\CH$), which is defined as $\CH = \PP \cup \PP^\PP \cup \ldots$ \cite{Allender1993}.
All of these classes are in $\PSPACE$, the class of problems solvable on classical computers that use polynomial space (but which are free to use exponential time).

\subsection{Results}\label{sec_results}

We describe a general problem we study here, called \textsc{($\delta,\Delta$)-LocalHamiltonian}.
Informally, it is the problem of estimating the ground-state energy of a given $k$-local Hamiltonian acting on $n$ qudits to additive error at most $\delta$, when promised that the spectral gap is at least $\Delta$ (see precise definitions in \cref{sec_defs}).
In the absence of any bound on the spectral gap (i.e.\ $\Delta=0$), the problem \textsc{$(1/\poly(n),0)$-LocalHamiltonian} is, by definition, the same as \textsc{$k$-LocalHamiltonian}, which is complete for $\QMA$ for $k\geq 2$ \cite{Kitaev2002,Kempe2003,Kempe2006}.
Meanwhile, \textsc{$(1/\exp(n),0)$-LocalHamiltonian} is, by definition, \textsc{Precise-$k$-LocalHamiltonian} \cite{Fefferman2018}, which is complete for $\pQMA$.
We henceforth suppress the dependence on the number of qudits $n$ in the notation $\exp$ and $\poly$ for the rest of the paper.

To our knowledge, Aharonov et al.\ \cite{Aharonov2008} were the first to study the \textsc{$k$-LocalHamiltonian} problem in the presence of a spectral gap.
Specifically, they considered \textsc{$(1/\poly,1/\poly)$-LocalHamiltonian} and showed it to be complete for the class $\PGQMA$ (Polynomially Gapped \QMA).
The definition of $\PGQMA$, which is given in \cref{sec_defs}, depends on a notion of a spectral gap for \emph{proof systems}, distinct from that for Hamiltonians.
For complexity classes associated with proof systems such as $\QMA$, $\QCMA$ and the variants we study in this work, the spectral gap corresponds to the gap in the highest and second-highest accept probabilities of the optimal witness and the next-optimal orthogonal witness.
A priori, the two notions of a spectral gap have no relation with each other.
We show that the two notions are equivalent for various cases ($\delta$ and $\Delta$ each behaving as $1/\poly$ or $1/\exp$), by showing that \textsc{$(\delta,\Delta)$-LocalHamiltonian} is complete for the appropriate spectral-gapped $\QMA$ class.

To understand the relation between the gapped $\QMA$ classes and the regular versions without a spectral gap, we focus on the precise regime, so that $\delta = 1/\exp$ henceforth for the rest of this section.
By specifying the spectral gap to be $\Omega(1/\poly)$, we get the problem \textsc{$(1/\exp,1/\poly)$-LocalHamiltonian}.
We show in \cref{thm_gappedham_in_ppgqma} that this problem is in a class we call $\pPGQMA$ (Precise Polynomially Gapped \QMA), which is the precise analogue of $\PGQMA$.
We also show (\cref{thm_ppgqma_inpp}) that $\pPGQMA \subseteq \PP$, implying that $\pPGQMA$ is likely different from $\pQMA$, which equals $\PSPACE$.
Specifically, assuming that $\PP\neq \PSPACE$, there is a separation between $\pPGQMA$ and $\pQMA$.
The $\PP$ upper bound on $\pPGQMA$ is optimal: we show that \textsc{$(1/\exp,1/\poly)$-LocalHamiltonian} is $\PP$-hard (\cref{thm_gappedham_pphard}).
Thus, we tightly characterize the complexity of the class by showing $\pPGQMA=\PP$ and prove that \textsc{$(1/\exp,1/\poly)$-LocalHamiltonian} is its associated complete problem.

The results in the previous paragraph show that the $\PSPACE$-hardness result of Ref.~\cite{Fefferman2018} relies on the fact that the spectral gaps of the associated Hamiltonians can decay rapidly with the system size.
This raises the question of the maximum scaling of the spectral gap required in order to retain $\PSPACE$-hardness.
This is an important question since if the $\PSPACE$-hardness results only apply when there is no promise whatsoever on the spectral gap, it would indicate that $\PSPACE$-hardness of \textsc{Precise-$k$-LocalHamiltonian} is artificial.
We rule out this possibility by showing that if the spectral gap is bounded below by $1/\exp$, i.e.\ if we consider the problem \textsc{$(1/\exp,1/\exp)$-LocalHamiltonian}, the problem remains $\PSPACE$-hard.
Specifically, we show in \cref{thm_pegqma_complete} that this problem is complete for a class called $\pEGQMA$ (Precise Exponentially Gapped \QMA).
Next, we show that $\pEGQMA$ equals $\PSPACE$ (\cref{thm_pegqma_pspace}), implying that instances with $\Omega(1/\exp)$ spectral gaps are no less complex than the general case.

Lastly, we consider the analogues of these classes when the witness is classical, which gives us the classes $\QCMA$ (Quantum Classical Merlin Arthur), $\pQCMA$, $\pPGQCMA$ (Precise Polynomially Gapped \QCMA) and $\pEGQCMA$ (Precise Exponentially Gapped \QCMA).
The complete problems for these classes are the appropriate versions of the \textsc{LocalHamiltonian} problem under the additional promise that there is an efficient classical description of a circuit to prepare a low-energy state, as we show in \cref{thm_qcma_complete,thm_pqcma_complete,thm_pgqcma_complete,thm_pegqcma_complete}.
We define this problem in \cref{sec_completeprobs} and denote it \textsc{$(\delta,\Delta)$-GS-Description-LocalHamiltonian}, which is the problem of computing the ground-state energy to additive error $\delta$, given the promise that there exists a polynomial-size circuit to prepare a low-energy state and promised that the spectral gap of the Hamiltonian is at least $\Delta$.
As stated in \cref{corll_ppgqcma_ppgqma}, we show that $\pPGQCMA$ has the same complexity as $\pPGQMA$, implying that in the precise setting, once there is a $\Omega(1/\poly)$ promise on the spectral gap, a further promise that there exists an efficient circuit to prepare a low-energy state is redundant.
We comment more on this result in \cref{sec_discussion}.

\begin{table*}[t]
\begin{tabular}{@{}l@{\hspace{10pt}}l@{\hspace{8pt}}l@{\hspace{10pt}}l@{\hspace{8pt}}l@{}} \toprule
\multirow{2}{*}[0pt]{\parbox{4em}{Spectral gap ($\Delta$)}} & \multicolumn{2}{c}{\textsc{$(\delta,\Delta)$-GS-Description-LocalHamiltonian}} & \multicolumn{2}{c}{\textsc{$(\delta,\Delta)$-LocalHamiltonian}} \\ \cmidrule(lr){2-3} \cmidrule(lr){4-5}
& $\delta=1/\poly$ & $\delta=1/\exp$ & $\delta=1/\poly$ & $\delta=1/\exp$ \\ \midrule
$1/\poly$ & $\PGQCMA$ \{\ref{thm_pgqcma_complete}\} ($=_R\QCMA$ \{\ref{lem_egqcma}\}) & $\pPGQCMA $ \{\ref{thm_ppgqcma_complete}\} ($=\PP$ \{\ref{thm_ppgqcma_pp}\}) & $\PGQMA$ & $\pPGQMA$ \{\ref{thm_ppgqma_complete}\} ($=\PP$ \{\ref{thm_ppgqma_pp}\})
\\ $1/\exp$ & $\EGQCMA$ ($=_R\QCMA$ \{\ref{lem_egqcma}\})  & $\pEGQCMA$ \{\ref{thm_pegqcma_complete}\} ($=\NP^\PP$ \{\ref{lem_pegqcma}\}) & $\EGQMA (?)$ & $\pEGQMA$ \{\ref{thm_pegqma_complete}\} ($=\PSPACE$ \{\ref{thm_pegqma_pspace}\})
\\ 0 & $\QCMA$ \{\ref{thm_qcma_complete}\} & $\pQCMA$ \{\ref{thm_pqcma_complete}\} ($=\NP^\PP$)  & \QMA & $\pQMA$ ($=\PSPACE$)
\\ \bottomrule
\end{tabular}
\caption{\raggedright Complexity of variants of the \textsc{LocalHamiltonian} problem as a function of the parameters $\delta$, the promise gap, and $\Delta$, the spectral gap.
The problem is complete for the class mentioned in each cell.
For reference, we mention in curly brackets the theorem number corresponding to the results proved in this work.
The question mark corresponding to the entry $\EGQMA$ indicates that the result is a conjecture and the notation $=_R$ denotes equivalence under randomized reductions (defined in \cref{sec_gqcma}).
}
\label{tab_complexitylh}
\end{table*}

In \cref{tab_complexitylh}, we give an overview of the parameter dependence of the complexity of two main problems studied in this work, namely \textsc{$(\delta,\Delta)$-LocalHamiltonian} and \textsc{$(\delta,\Delta)$-GS-Description-LocalHamiltonian}.
The problems are completely characterized by the appropriately gapped versions of $\QMA$ or $\QCMA$, or their precise variants.
The complexity class in any cell in the table is a subset of all the classes below it in the same column, since these classes correspond to weaker promises on the spectral gap.
Similarly, the complexity class associated with \textsc{$(\delta,\Delta)$-GS-Description-LocalHamiltonian} is a subset of that associated with \textsc{$(\delta,\Delta)$-LocalHamiltonian}, because the former problem is associated with an extra promise.
While we have given evidence that $\pPGQMA \neq \pQMA$, it is unknown whether the same holds for the question $\PGQMA \stackrel{?}{=} \QMA$.
Similarly, while we have proved $\pEGQMA = \pQMA$, it would be interesting to see if a similar result holds for $\EGQMA$.

\subsection{Techniques} \label{sec_techs}
Here, we give an overview of the primary techniques used in proving our results.

\textbf{Imaginary-time evolution and the power method.}---
To show the containment $\pPGQMA \subseteq \PP$, we use a technique called the ``power method'' \cite{Mises1929}.
The broad idea behind the algorithm is that if a matrix $A$ is promised to have a spectral gap between the largest two eigenvalues, the behavior of $A^d$ for large $d$ is dominated by the largest eigenvalue.
We give a $\PP$ algorithm to compute $\Tr(A^d)$ for an exponentially large matrix $A$ and $d = \poly(n)$ for a wide class of matrices $A$.
This wide class includes sparse matrices and matrices representing local observables as special cases.
The $\PP$ algorithm uses the Feynman sum-over-paths idea \cite{Adleman1997} to express the trace as a sum over $2^\poly$ many terms, each of which is a product over quantities of the form $\bra{x}R\ket{y}$ for some matrix $R$ whose entries are efficiently computable.
A $\PP$ algorithm can decide whether the sum over $2^\poly$ many terms, each term computable in polynomial time, is above or below a threshold.

The power method is closely related to another technique called the ``cooling algorithm'', inspired by a brief discussion by Schuch et al.\ \cite{Schuch2007a}.
The idea is that letting a system evolve in imaginary time can produce an unnormalized state close to the ground state.
Imaginary-time evolution is a linear, albeit nonunitary, operation and produces an unnormalized state $\rho'$ in general.
Schuch et al.~relied on a quantum characterization of $\PP$, namely $\postBQP$.
The class $\postBQP$ \cite{Aaronson2005} is the class of problems solvable in polynomial time on a quantum computer with access to the resource of postselection, which is the ability to condition on exponentially unlikely events.
Aaronson \cite{Aaronson2005} showed that any linear operation, even nonunitary ones, may be simulated in $\postBQP$.
Schuch et al.'s algorithm \cite{Schuch2007a} proposes to decompose the imaginary-time evolution operation $\exp[-\beta H]$ into a series of local operations $\exp[-\beta H_i]$ using Trotterization, and implementing each local operation using the resource of postselection.
Unfortunately, the state-of-the-art error bounds for Trotterization of imaginary-time evolution \cite{Childs2019a} give, at best, a multiplicative error that is exponential in $n$ (see also Refs.~\cite{Haah2018,Tran2019}), and hence this technique does not work in the precise regime.
We prove a more general statement about precise computation of ground-state local observables for Hamiltonians with a spectral gap using exact imaginary time evolution as opposed to a Trotterized version.
Specifically, we give a $\P^\PP$ algorithm that provably works not just for $1/\poly$ precision, but also $1/\exp$ precision in computing local observables in addition to the Hamiltonian.
Our technique is closely related to the power method, since the core of the algorithm is to compute expectation values of powers of the Hamiltonian.

\textbf{Small-penalty clock construction.}--- Our second major technical contribution is a modification of the clock construction that we call the small-penalty clock construction.
One of the ways this technique is useful is as follows.
As mentioned earlier and as will be described in detail in \cref{sec_defs}, it is possible to consider spectral-gapped versions of both the \textsc{LocalHamiltonian} problem and the class $\QMA$ and their variants.
We have already discussed the (natural) notion of a spectral gap for Hamiltonians.
For $\QMA$ and related classes, the spectral gap is related to the difference in accept probabilities between the optimal and next-optimal witness.
Our technique allows us to bridge the notion of spectral gap in both cases by constructing spectral-gap-preserving reductions.
In other words, the small-penalty clock construction allows us to prove that the Hamiltonians resulting from the construction inherit a spectral gap related to the gap in accept probabilities in the circuit, for several variants of $\QMA$.
This ability is used in the proofs of \cref{thm_pegqma_complete,thm_qcma_complete,thm_pqcma_complete,thm_pgqcma_complete,thm_pegqcma_complete}.
An interesting feature of the modified clock construction is that it also allows us to show that, when there is a classical witness (i.e.\ a $\QCMA$ computation), the resulting Hamiltonian has a classical description for a state with energy close to the ground-state energy.
Another related application of the small-penalty clock construction is that it also allows us to show complexity lower bounds like in \cref{thm_gappedham_pphard,thm_pqcma_gap_pphard}.
In these cases, we directly reduce from $\PP$ to the appropriate gapped version of the \textsc{LocalHamiltonian} problem instead of a reduction from the corresponding $-\QMA$ class.

We now spell out what enables the small-penalty clock construction to show the above results.
As mentioned before, the clock construction and its variants encompass all current proofs of hardness for $\QMA$ and related classes.
Typically, this consists of mapping a circuit to a Hamiltonian $H = H_\mathrm{input} + H_\mathrm{prop} + H_\mathrm{clock} + H_\mathrm{output}$.
Roughly speaking, each term locally enforces that the computation is a valid step of a $\QMA$ protocol by adding energy penalties to undesirable states.
The ``witness register'', where a quantum prover may input any quantum state, is left unpenalized and the Hamiltonian therefore has no terms acting on the witness register.
The role of $H_\mathrm{output}$ is to ensure that witnesses and computations that lead to a \emph{low} accept probability at the output get a \emph{high} energy penalty.
In the absence of the penalty term at the output, the ground-state space of the Hamiltonian is well-known and is given by the subspace of the so-called ``history states'', each with the same energy.
The output penalty term $H_\mathrm{output}$ is what breaks the degeneracy and helps create a promise gap, and we will henceforth refer to this as simply the penalty term without qualification.

However, the addition of the penalty term makes the eigenstates of the Hamiltonian difficult to analyze, since the magnitude of the penalty can be large, i.e. $\Omega(1)$ in strength.
In this work, we often choose the output penalty terms to have small strength.
This might seem like a strange choice to make since one is typically interested in making the promise gap as large as possible.
However, since we are dealing with instances where the promise gap is already exponentially small, our choice is not too costly.
The advantage this gives us is that the ground-state energy tracks the effect of the output penalty more faithfully.
More concretely, the smallness of the penalty term allows us to use tools like the Schrieffer-Wolff transformation \cite{Schrieffer1966,Bravyi2011}, which can be viewed as a rigorous formulation of degenerate perturbation theory.
We review the Schrieffer-Wolff transformation in \cref{sec_schwolff}.

\textbf{Spectral gap in adjacency matrix.}---
For the proof of \cref{thm_pegqma_pspace}, we show a reduction\footnote{This reduction is inspired by unpublished work by one of us and Cedric Lin \cite{Fefferman2016c}, and we supplement it with a technique to create spectral gaps.} from a natural \PSPACE-complete graph problem to an instance of a problem known as \textsc{$(1/\exp,1/\exp)$-SparseHamiltonian}\footnote{We actually show a reduction to the complement of the problem (where YES and NO instances are reversed), but this turns out not to matter because $\PSPACE$ is closed under complement.}.
This problem is a generalization of \textsc{$(1/\exp,1/\exp)$-$k$-LocalHamiltonian}, allowing for the Hamiltonian to be any sparse Hamiltonian with a spectral gap $\geq 1/\exp$.
Sparse Hamiltonians are Hermitian matrices that can be exponentially large, with at most $\poly(n)$ nonzero entries per row in some basis and an efficient algorithm for computing any entry of the matrix.
They are a generalization of local Hamiltonians.

The \PSPACE-complete graph problem may be described as \textsc{SuccinctGraphReachability}, which is to decide if there is a path from one vertex to another in a succinctly described graph of exponential size (also see Ref.~\cite{Atia2017}).
We show that one can always construct a \PSPACE-bounded Turing machine such that the resulting Hamiltonian after the reduction always has a spectral gap that is at least $1/\exp(n)$.
We do this through an explicit analysis of the eigenvalues of the Hamiltonian, which are related to the lengths of cycles and paths of the graph constructed from the Turing machine.
Next, we give a $\pEGQMA$ upper bound to \textsc{$(1/\exp,1/\exp)$-SparseHamiltonian}, i.e.\ the problem in the presence of a spectral gap, establishing that $\PSPACE \subseteq \pEGQMA$.
\subsection{Discussion} \label{sec_discussion}

Our first main result was that the addition of even an inverse-polynomially small spectral gap takes the complexity of precisely estimating the ground-state energy of a local Hamiltonian from $\pQMA=\PSPACE$ to $\pPGQMA=\PP$.
Note that this result also implies a difference between the case of no spectral gap and a constant spectral gap.
Therefore, we have given a \emph{provable setting} where the difference in complexity between two problems is attributable entirely to the spectral gap.

Our second main result concerned a modification of the same problem of
precisely estimating the ground-state energy of a local Hamiltonian promised to have an inverse-polynomial spectral gap.
When additionally promised that there exists a classical description of a circuit to prepare a state whose energy is exponentially close to the ground-state energy, our results show that the complexity of the problem does not get weaker.
Specifically, we show that the class $\pPGQCMA$ is equivalent to $\pPGQMA$.

The above equivalence result is in sharp contrast with the belief $\pQCMA\neq \pQMA$ in the non-spectral-gapped case.
This inequality follows from the conjecture that $\NP^\PP \neq\PSPACE$, which, if false, would lead to a collapse of the counting hierarchy.
The inequality $\pQCMA\neq \pQMA$ rules out the possibility of there being polynomial-size circuits to prepare ground states of local Hamiltonians to exponential precision, since otherwise the prover could simply supply a description of such a circuit.
Our equivalence result that $\pPGQMA = \pPGQCMA$ is consistent with the following intriguing conjecture about the circuit-complexity of ground states of low-energy Hamiltonians, although it does not imply the conjecture.

\begin{conj}
Consider any Hamiltonian $H$ on $n$ qubits with ground-state energy $E_1$ and a $1/\poly$ spectral gap.
Then there exists a low-energy state $\ket{\psi}$ satisfying $\bra{\psi}H\ket{\psi} \leq E_1 + 2^{-\poly(n)}$ that can be prepared by an efficient quantum circuit, namely a state of the form $\ket{\psi} = U \ket{0}^{m}$, where $m$ and the size of $U$ are both polynomials in $n$.\label{prop_gqcma}
\end{conj}

Note that \cref{prop_gqcma} implies the following results: i) $\pPGQMA=\pPGQCMA$, ii) $\PGQMA=\PGQCMA$, and iii) $\PGQCMA = \QCMA$.
We have proved i) and given strong evidence for iii) in \cref{lem_egqcma} by showing that  $\PGQCMA =_R \QCMA$.
These results do not imply \cref{prop_gqcma} because the reductions do not imply anything about the classical witnesses.
We also note that the quantum circuits referred to in \cref{prop_gqcma} may be hard to find---the conjecture is only concerned with the existence of such circuits, and not with whether these circuits can be obtained by an efficient algorithm.
In complexity-theoretic language, these circuits may be nonuniform.
This is why \cref{prop_gqcma} is not in contradiction with Ref.~\cite{Schuch2008b}, which argues that finding efficient matrix-product-state representations of Hamiltonians with a $\Omega(1/\poly)$ spectral gap can be hard.

If \cref{prop_gqcma} were true, it would also explain the observed success of quantum algorithms such as the variational quantum eigensolver (VQE) \cite{Peruzzo2014,Jones2019}, which seek to solve a much simpler problem of preparing low-energy states of translation-invariant many-body Hamiltonians with energy $1/\poly$-close to the ground-state energy.
A large class of translation-invariant Hamiltonians have a spectral gap that is either a constant, $\Theta(1)$ (gapped phases), or vanishing in the system size as $\Theta(1/n^{1/D})$ (gapless phases described by conformal field theories in $D$-dimensions).
Therefore, \cref{prop_gqcma} applies to both these cases and would imply the existence of polynomial-size circuits to prepare states with high overlap with the ground state.
Such circuits are generally found in the VQE algorithm if one optimizes over sufficiently many parameters.
This behavior is in line with other instances where a lower bound on the spectral gap implies tractability of the ground state in various senses \cite{Hastings2007,Aharonov2010a,Molnar2015,Arad2017}.

Coming to the case of exponentially small spectral gaps, we have shown that $\pEGQMA=\pQMA$.
This implies that $\pEGQMA\neq \pEGQCMA$ unless the counting hierarchy collapses.
Therefore, we have given a class of local Hamiltonians (in the proof of \cref{lem_ham_preciseegqma}) with exponentially small spectral gaps, whose ground states have exponentially large circuit complexity.
This is a result of independent interest, and it might be interesting to study the whether these Hamiltonians can be classified as quantum spin glasses, which are believed to be hard to cool down to zero temperature \cite{Knysh2016}.

In another intriguing line of work, Aharonov and Green \cite{Aharonov2017a} and Green, Kindler, and Liu \cite{Green2019} have given interactive protocols for precise quantum complexity classes with a computationally bounded prover $\mathcal{P}$ and a computationally bounded verifier $\mathcal{V}$, denoted $\IP[\mathcal{P},\mathcal{V}]$.
A goal of this line of work is to give a quantum-inspired proof of the result $\IP=\PSPACE$ \cite{Shamir1992} by giving an interactive protocol for $\pQMA$ \cite{Green2019} (which equals $\PSPACE$) with a $\BPP$ verifier.
This has been successful so far with $\pBQP$ and $\pQCMA$ (which equals $\NP^\PP$) but not yet with $\pQMA$.
From the result of Ref.~\cite{Aharonov2017a} and our result that $\pPGQMA = \PP$, there is an $\IP[\pBQP,\BPP]$ protocol for $\pPGQMA$.
Our results indicate that the spectral gap might play an important role in extending such an interactive protocol to $\PSPACE$.
Namely, such an extension would need to be able to work with inverse-exponentially small spectral gaps.

In addition, the class $\postQMA$ \cite{Usher2017,Morimae2017a} is the class where there is a quantum prover and a $\postBQP$ verifier, where one may condition (postselect) on exponentially unlikely outcomes.
This class has been shown to be equal to $\pQMA$ \cite{Morimae2017a}, so an alternative approach mentioned by Green et al.~\cite{Green2019} to reprove the result $\IP=\PSPACE$ is to exhibit an $\IP[\postQMA,\BPP]$ protocol for $\postQMA$.
To complete such a proof, it would suffice to prove a witness-preserving amplification technique like in $\QMA$ \cite{Marriott2005,Nagaj2009} that additionally handles postselection.
Witness-preserving amplification is a technique for improving the promise gap of an interactive protocol by modifying the verifier's strategy while keeping the witness fixed.
We show in \cref{lem_postqma_amp} that, assuming $\PP \neq \PSPACE$, the soundness of a $\postQMA$ protocol cannot be reduced beyond a particular point without requiring the witness to grow larger or requiring the postselection success probability to shrink.
Therefore, we obtain evidence that a witness-preserving amplification technique for $\postQMA$ should differ significantly from the technique of Marriott and Watrous \cite{Marriott2005}, since in the latter, repeating the verifier's circuit suffices to get any soundness parameter $s \leq 2^{-\poly}$.

So far, we have considered the spectral-gap promise to be applicable to both YES and NO instances of the problems defined.
We can also define asymmetric problems where only the YES instances are promised to have a spectral gap.
The motivation for considering such asymmetric promises is that they are related to complexity classes where the accepting witness is promised to be unique, such as the class $\UQMA$ \cite{Aharonov2008}.
The problems with asymmetric promises can only be harder than their symmetric analogues, since the promise is weaker.
We show that for both $\Omega(1/\poly)$ and $\Omega(1/\exp)$ spectral gaps in the precise setting, there is no difference between symmetric and asymmetric promises on the spectral gaps.
Specifically, we show in \cref{thm_asym_symgaps} that the classes with asymmetric promises are of the same complexity as those with symmetric promises.

We remark here that the promise of a spectral gap above a unique ground state is distinct from assuming that we have a $\UQMA$ instance.
The reason is that for \textsc{LocalHamiltonian}, the presence of  a spectral gap does not imply that there is a unique accepting witness, it only implies a unique ground state.
In case the ground-state subspace is polynomially degenerate, the $\PP$ algorithm continues to work to produce estimates of the ground-state energy.

Lastly, we add that results shown in the precise regime do not always imply analogous results in the non-precise regime.
For example, our work gives evidence that $\pPGQCMA \neq \pQCMA$, but in the non-precise regime we can show $\PGQCMA=_R \QCMA$.
In this respect, inequivalence results in the high-precision regime resemble oracle separation results in complexity theory, which is a mature area of research with several important results \cite{Aaronson2002,Aaronson2010,Raz2018}.
While oracle separations do not constitute strong evidence for the inequivalence of two complexity classes, they are useful in ruling out proof techniques that work relative to oracles, or ``relativize''.
Similarly, inequivalence results in the precise regime can rule out proof techniques from extending to the precise regime.
For example, a purported proof that $\QCMA =\QMA$ must not work in the precise regime, otherwise we would obtain $\pQCMA=\pQMA$, or $\PSPACE=\PP$, which is believed to be unlikely.

\subsection{Related work} \label{sec_others}

The study of Hamiltonian complexity \cite{Kempe2003,Kempe2006,Bravyi2008,Oliveira2008,Bravyi2008a,Gottesman2009,Aharonov2009,Cubitt2013,Gharibian2015} has given rise to many techniques and important results applicable in quantum many-body physics, such as \cite{Aharonov2007,Schuch2007,Schuch2008b,Cubitt2015,Cubitt2018,Bausch2018b,Piddock2018,Kohler2019,Kohler2020}.
The clock construction has also been analyzed in detail recently \cite{Bausch2018a,Caha2018,Watson2019}.

The study of exponentially small promise gaps in the context of quantum classes can be traced to Watrous \cite{Watrous2009}, who defined $\mathsf{PQP}$ and showed its equivalence with $\postBQP$, which equals $\PP$ \cite{Aaronson2005}.
In the precise setting, one can sometimes give far stronger evidence for the (in)equivalence of complexity classes than in the analogous bounded error setting, as is the case for precise versions of the questions of $\QCMA$ vs.\ $\QMA$ \cite{Fefferman2018} and $\QMA(2)$ vs.\ $\QMA$ \cite{Brandao2008,Blier2009,Pereszlenyi2012,Fefferman2018}.
There has been work on quantum interactive proof systems with exponentially small promise gaps, such as in the context of $\QMA(2)$ \cite{Pereszlenyi2012}, or with even smaller gaps, such as in Refs.~\cite{Ito2012,Ji2017,Fitzsimons2019}.
Fefferman and Lin \cite{Fefferman2016c,Fefferman2018} studied the precise regime of $\QMA$, showing it to equal $\PSPACE$, leading to other works concerning precise classes \cite{Fefferman2016b,Morimae2017a}.
Gharibian et al.\ \cite{Gharibian2018} considered quantum generalizations of the polynomial hierarchy, where precise classes and spectral gaps are relevant to the definitions and proof techniques.

Aharonov et al.\ \cite{Aharonov2008} were the first to consider the complexity of the \textsc{LocalHamiltonian} problem in the presence of spectral gaps, motivated by the question of uniqueness \cite{Valiant1986} for randomized and quantum classes.
They showed the equivalence of $\UQCMA$ and $\QCMA$, and that of $\UQMA$ and $\PGQMA$, using similar techniques as Valiant and Vazirani \cite{Valiant1986} in their proof of equivalence of $\UNP$ and $\NP$.
Jain et al.\ \cite{Jain2012} defined the class $\mathsf{FewQMA}$ and showed that it is contained in $\P^{\UQMA}$, giving a technique to reduce the dimension of accepting witnesses.

More recently, Gonz\'alez-Guill\'en and Cubitt \cite{Gonzalez-Guillen2018} studied the spectral gap of a large class of Hamiltonians that encode history states in their ground state and showed that the spectral gap is upper bounded by $O(1/\poly)$.
A similar result was obtained by Crosson and Bowen \cite{Crosson2018a} using different techniques.
These works are mainly concerned with the existence of a $\Theta(1)$ spectral gap, whereas our results distinguish between $1/\poly$ and $1/\exp$ spectral gaps.

Finally, Ambainis \cite{Ambainis2014} studied the problem of estimating spectral gaps and local observables and gave a $\P^{\QMA{[\log]}}$ upper bound for these problems, while also giving $\P^{\QMA{[\log]}}$-hardness results (also see Ref.\ \cite{Gharibian2019b}).
The class $\P^{\QMA{[\log]}}$ is the class of problems solvable in polynomial time by making logarithmically many (adaptive) queries to a $\QMA$ oracle. 
Gharibian and Yirka \cite{Gharibian2019b} showed that $\P^{\QMA{[\log]}} \subseteq \PP$ and extended previous hardness results to more natural Hamiltonians.
Gharibian, Piddock, and Yirka \cite{Gharibian2019} also gave a very natural complete problem for the class $\P^{\QMA{[\log]}}$ in the context of computing local observables in ground states.
Novo et al.~\cite{Novo2019} have recently studied the closely-related problem of sampling from the distribution obtained by making energy measurements and obtain various interesting hardness results, under different notions of error.

\section{Definitions and complete problems} \label{sec_defs}
We have seen the definition of $\BQP$ in terms of the class $\BQP[c,s]$ with general parameters $c$ and $s$.
The \textsf{Precise-} version of $\BQP$ can be defined similarly.
\begin{defn}
 $\pBQP = \cup_{c-s \geq 1/\exp} \BQP[c,s]$.
\end{defn}
\noindent This class is known to be equal to $\PP$ (see, e.g., Ref.~\cite{Gharibian2018}).

We now give an equivalent definition of $\QMA$ in terms of the eigenvalues of an operator called the \emph{accept operator}.
We will then define a very general class called Gapped \QMA, $\GQMA[c,s,g_1,g_2]$, which has several parameters.
By specifying these parameters, we can define the major complexity classes in this work.
The complexity classes corresponding to classical witnesses ($\QCMA$ and its derivatives) are defined analogously.

The alternative definition of $\QMA$ is in terms of the ``accept operator'' $Q(U_x) = \bra{0}^{\otimes m} U_x^\dag\Pi_\mathrm{out}U_x \ket{0}^{\otimes m}$ on the witness register, where 
$\Pi_\mathrm{out}$ is the projector on to the accept state ($\ket{1}_o$).
For any state $\ket{\Psi}$ provided as a witness, the quantity $\bra{\Psi} Q_x \ket{\Psi}$ is the accept probability of the circuit.
We will henceforth suppress the dependence of $Q$ on the unitary $U_x$ and the instance $x$.
The eigenvalues of $Q$, $\lambda_1(Q) \geq \lambda_2(Q) \geq \ldots$ are important quantities to consider since the accept probability of any input proof state is a convex combination of these eigenvalues.
The alternative definition of $\QMA$ in terms of the operator $Q$ is as follows:
\begin{defn}[Alternative definition of $\QMA{[c,s]}$]
$A=(A_\mathrm{yes},A_\mathrm{no})$ is a $\QMA[c,s]$ problem iff for every instance $x$ there exists a uniformly generated circuit $U_x$ of size $\poly(n)$ acting on $m+w = \poly(n)$ qubits, with the property that \\
\begin{tabularx}{\linewidth}{l X}
If $x \in A_\mathrm{yes}$: & $\lambda_1(Q) \geq c$ \\
If $x \in A_\mathrm{no}$: & $\lambda_1(Q) \leq s$,
\end{tabularx}
where $Q = Q(U_x)$ is as above.
\end{defn}

Note that we are typically interested in the behavior of the maximum accept probability, which equals the largest eigenvalue of $Q$.
We are also interested in the lowest eigenvalue of a Hamiltonian $H$ for the \textsc{LocalHamiltonian} problem and its variants.
Therefore, we order eigenvalues in nonincreasing order for accept operators and in nondecreasing order for Hamiltonians.
For the same reason, we define the spectral gap differently for accept operators and Hamiltonians.
For a Hamiltonian, we define the spectral gap to be the difference in the smallest two eigenvalues $E_2-E_1$.
For accept operators, the spectral gap is the difference between the \emph{highest two eigenvalues} $\lambda_1(Q)-\lambda_2(Q)$.
This is equal to the difference in the accept probabilities of the optimal witness and the next-optimal witness orthogonal to it.
It will usually be clear from context which spectral gap we are referring to.

Now let us define the class $\GQMA[c,s,g_1,g_2]$.
It corresponds to a promise on the operator $Q$ having a spectral gap of at least $g_1$ in the YES case, and at least $g_2$ in the NO case:

\begin{defn}[Gapped \QMA]
$\GQMA[c,s,g_1,g_2]$ is the class of promise problems $A=(A_\mathrm{yes},A_\mathrm{no})$ such that for every instance $x$, there exists a polynomial size verifier circuit $U_x$ acting on $\poly(n)$ qubits and its associated accept operator $Q$ such that\\
\begin{tabularx}{\linewidth}{l X}
If $x \in A_\mathrm{yes}$: & $\lambda_1(Q) \geq c$ and $\lambda_1(Q) - \lambda_2(Q) \geq g_1$ \\
If $x \in A_\mathrm{no}$: & $\lambda_1(Q) \leq s$ and $\lambda_1(Q) - \lambda_2(Q) \geq g_2$.
\end{tabularx}
\end{defn}
\noindent This definition is a generalization of the class $\PGQMA$ (Polynomially Gapped \QMA) defined by Aharonov et al.\ in Ref.~\cite{Aharonov2008}:
\begin{defn}
$\PGQMA $$=$$ \cup_{c-s, g_1, g_2 \geq 1/\poly} \GQMA[c,s,g_1,g_2]$.
\end{defn}
To see the relation of this class with $\QMA$, notice that by setting $g_1 = g_2 =0$, the promise on spectral gaps becomes vacuous, since $\lambda_1(Q) \geq \lambda_2(Q)$ by definition.
Therefore, we get the equality $\GQMA[c,s,0,0] = \QMA[c,s]$.
We also define
\begin{defn}[Exponentially Gapped \QMA]
 $\EGQMA = \cup_{\substack{c-s \geq 1/\poly\\ g_1, g_2 \geq 1/\exp}} \GQMA[c,s,g_1,g_2]$.
\end{defn}

We now come to precise versions of these classes, where the completeness--soundness gap $c-s$ can be exponentially small, giving us more powerful classes.
The first of these is $\pQMA$, which was defined in Ref.~\cite{Fefferman2018} and shown to be equal to $\PSPACE$.
\begin{defn}
$\pQMA =$ $ \cup_{c-s \geq 1/\exp} \QMA[c,s]$.
\end{defn}

This definition should be compared to the precise version of $\GQMA$, which comes in two varieties: the spectral gaps can either be polynomially small ($\pPGQMA$) or exponentially small ($\pEGQMA$).
\begin{defn}[$\pPGQMA$]
$\pPGQMA$, short for Precise Polynomially Gapped \QMA, is the class with exponentially small promise gaps and polynomially small spectral gaps: \\$\pPGQMA = \cup_{\substack{c-s \geq 1/\exp \\ g_1, g_2  \geq 1/\poly}} \GQMA[c,s,g_1,g_2]$.
\end{defn}
\begin{defn}[$\pEGQMA$]
$\pEGQMA$, short for Precise Exponentially Gapped \QMA, has both the promise gap and spectral gap exponentially small:\\
$\pEGQMA = \cup_{\substack{c-s \geq 1/\exp \\ g_1, g_2  \geq 1/\exp}} \GQMA[c,s,g_1,g_2]$.
\end{defn}

We now come to complexity classes in which the prover sends a classical witness but the verifier remains quantum.
The classicality of the witness can be enforced by measuring the qubits sent by the prover in the computational basis and interpreting qubits in the computational basis as classical bits.
If the verifier is only allowed to make measurements at the end, we use the standard protocol for deferring measurements: we apply a ``copy operation'' $U_c$ that has CNOTs from the qubits in the witness register to an ancilla register in the state $\ket{0}^w$.
We leave the qubits in the witness state unmeasured.
This modified circuit has the property that it preserves the accept probabilities of input witness states that are in the computational basis.
Further, the eigenstates of the modified accept operator acting on the register can be taken to be computational basis states.
This allows us to define $\QCMA$ and its derivatives in terms of the accept operator and also allows us to consider a gapped version of $\QCMA$:

\begin{defn}[{$\GQCMA[c,s,g_1,g_2]$}]
$A=(A_\mathrm{yes},A_\mathrm{no})$ is a $\GQCMA[c,s]$ problem iff for every instance $x$ there exists a uniformly generated circuit $U_x$ of size $\poly(n)$ acting on $m+w = \poly(n)$ qubits, with the property that \\
\begin{tabularx}{\linewidth}{l X}
If $x \in A_\mathrm{yes}$: & $\lambda_1(Q) \geq c$ and $\lambda_1(Q)-\lambda_2(Q) \geq g_1$ \\
If $x \in A_\mathrm{no}$: & $\lambda_1(Q) \leq s$, and $\lambda_1(Q)-\lambda_2(Q) \geq g_2$,
\end{tabularx}
where $Q = Q(U_xU_c)$ is the accept operator of the modified circuit with the copy operation $U_c$ described above.
\end{defn}
\begin{defn}
The derived classes of $\GQCMA$ are given by
\begin{itemize} \raggedright
 \item $\QCMA[c,s] = \GQCMA[c,s,0,0]$. \\
 \item $\QCMA = \cup_{c-s > 1/\poly}\QCMA[c,s]$. \\
\item $\pQCMA = \cup_{c-s > 1/\exp}\QCMA[c,s]$. \\
\item Polynomially Gapped \QCMA: $\PGQCMA = \cup_{\substack{c-s> 1/\poly \\ g_1,g_2 > 1/\poly}} \GQCMA[c,s,g_1,g_2]$. \\
\item Precise Polynomially Gapped \QCMA: $\pPGQCMA = \cup_{\substack{c-s > 1/\exp \\ g_1,g_2 > 1/\poly}}\GQCMA[c,s,g_1,g_2]$. \\
\item Exponentially Gapped \QCMA: $\EGQCMA = \cup_{\substack{c-s> 1/\poly \\ g_1,g_2 > 1/\exp}} \GQCMA[c,s,g_1,g_2]$. \\
\item Precise Exponentially Gapped \QCMA: $\pEGQCMA = \cup_{\substack{c-s > 1/\exp \\ g_1,g_2 > 1/\exp}}\GQCMA[c,s,g_1,g_2]$.
\end{itemize}
\end{defn}

\subsection{Complete problems} \label{sec_completeprobs}
We now come to the definitions of problems that are complete for these classes.
The classic problem complete for the class $\QMA$ is the \textsc{LocalHamiltonian} problem \cite{Kitaev2002,Kempe2003,Kempe2006}.
We define a $k$-local observable to be a Hermitian operator $A$ that can be written as a sum over operators $A_i$ supported on $k$ qudits at most: $A = \sum_i^{\poly(n)} A_i$.
We assume that each term has bounded operator norm $\norm{A_i} \leq \poly(n)$.
The task in the \textsc{LocalHamiltonian} problem is to estimate the ground-state energy of a local Hamiltonian.
The decision version of the problem is as follows:
\begin{problem}{$k$-LocalHamiltonian[$a,b$]}
Input & A description of a $k$-local Hamiltonian $H = \sum_i h_i$ on $n$ qubits with $h_i \succeq 0$, two numbers $a$ and $b$ with $b>a$.\\
Output & YES if the ground-state energy $E_1 \leq a$, \\
& NO if $E_1 \geq b$, promised that one of them is the case.
\end{problem}
Henceforth we omit the phrase ``promised that one of them is the case'' because we will be exclusively considering promise problems unless otherwise specified.
Kitaev \cite{Kitaev2002} showed that \textsc{5-LocalHamiltonian[$a,b$]} with $b-a = \Omega(1/\poly)$ is $\QMA$-complete, which was improved to $k=3$ and then $k=2$ in Refs.~\cite{Kempe2003,Kempe2006}.
The parameter $\delta := b-a$, the promise gap, is a measure of the accuracy to which the solution is desired.
We define the problem in terms of $\delta$ only, as follows:
\begin{defn}
\textsc{$\delta$-$k$-LocalHamiltonian} $:= \cup_{b-a \geq \delta} $ \textsc{$k$-LocalHamiltonian[$a,b$]}.
\end{defn}

We now come to the gapped and precise versions of the problem, which turn out to be complete for their respective $-\QMA$ variants.
We also suppress the notation $k$ in the name of the problem, though there is formally a dependence on $k$.
In this work, our hardness results hold for $k\geq 3$ and it may be possible to improve our results to hold for $k=2$.
\begin{problem}{LocalHamiltonian[$a,b,g_1,g_2$]}
Input & Description of a $k$-local Hamiltonian $H = \sum_i h_i$ with $h_i \succeq 0$, numbers $a$, $b$, $g_1$, and $g_2$ with $b>a$. \\
Output & YES if the ground-state energy $E_1 \leq a$ and any state orthogonal to the ground state has energy $\geq E_1 + g_1$, \\
& NO if $E_1 \geq b$ and any state orthogonal to the ground state has energy $\geq E_1 + g_2$.
\end{problem}
In both the YES and NO cases above, we see that the Hamiltonian has a unique ground state and a spectral gap of at least $g_1$ in the YES case and $g_2$ in the NO case.
The above problem with promise gap $\delta = b-a$ and spectral gap $\Delta = \min[g_1, g_2]$ is defined to be:
\begin{defn}
\textsc{$(\delta,\Delta)$-LocalHamiltonian}  $:=$ $ \cup_{\substack{b-a \geq \delta \\ g_1,g_2 > \Delta}}$ \textsc{LocalHamiltonian[$a,b,g_1,g_2$]}.
\end{defn}
\noindent In the non-precise regime, the problem \textsc{$(1/\poly,1/\poly)$-LocalHamiltonian} was shown to be complete for $\PGQMA$ for $k\geq 2$ \cite{Aharonov2008}.

We now focus on the precise regime, i.e.\ $\delta = \Omega(1/\exp)$.
From the results of Ref.~\cite{Fefferman2018}, we know that \textsc{$(1/\exp,0)$-LocalHamiltonian} is $\pQMA$-complete for $k\geq 3$.
We show that:
\begin{thm}\label{thm_ppgqma_complete}
\textsc{$(1/\exp,1/\poly)$-LocalHamiltonian} is $\pPGQMA$-complete.
\end{thm}
\begin{thm}\label{thm_pegqma_complete}
\textsc{$(1/\exp,1/\exp)$-LocalHamiltonian} is $\pEGQMA$-complete.
\end{thm}
\noindent By virtue of these theorems, we can talk about the complexity of the classes $\pPGQMA$ and $\pEGQMA$ interchangeably with their complete problems.
The proofs of these theorems are given in \cref{sec_gapreds,sec_ingqma}.
The hardness results rely on the small-penalty clock construction, where the size of the penalty term is either $\Theta(1/\poly)$ or $\Theta(1/\exp)$.
The upper bounds are shown in \cref{lem_gappedham_in_precisepgqma,lem_gappedham_in_preciseegqma} and rely on a modification of the standard phase-estimation protocol used to show \textsc{$k$-LocalHamiltonian} is in \QMA.
Specifically, we consider the modified protocol of Ref.~\cite{Fefferman2018} used for \textsc{Precise-$k$-LocalHamiltonian} and observe that the spectral gaps in the energies translate to separations in the accept probabilities.

Finally, we turn to complete problems for $\QCMA$ and its derivatives.
The first problem, \textsc{GS-Description-LocalHamiltonian}, concerns finding the ground-state energy of a $k$-local Hamiltonian when there is a polynomial-size circuit to prepare a state close to the ground state (which constitutes a classical description of the ground state).

\begin{problem}{GS-Description-LocalHamiltonian{$[a,b,g_1,g_2]$}}
 Input & Description of a $k$-local Hamiltonian $H = \sum_i h_i$, numbers $a$, $b \geq a + \delta$, polynomials $T(n), m(n)$, together with the promise that there exists a circuit $V$ of size $T$ such that $V \ket{0^m} = \ket{\psi}$ satisfies $\bra{\psi}H\ket{\psi} \leq E_1 + {\delta^3}/{f(n)^2}$ for some polynomial $f(n) \geq \norm{H}$.\\
 Output & YES if the ground-state energy of $H$ satisfies $E_1 \leq a$ and the spectral gap of $H$ is at least $g_1$, \\
 & NO if $E_1 \geq b$ and the spectral gap of $H$ is at least $g_2$.
\end{problem}
\begin{defn}
\textsc{$(\delta,\Delta)$-GS-Description-LocalHamiltonian} $:= \cup_{\substack{b-a \geq \delta, g_1,g_2 \geq \Delta}}$ \textsc{GS-Description-LocalHamiltonian[$a,b,g_1,g_2$]}
\end{defn}

As in the case of \textsc{$(\delta,\Delta)$-LocalHamiltonian}, if we take $\Delta=0$, we get a version without any promise on the spectral gap.
This is a close relative of the following problem proved to be \QCMA-complete for $\delta=\Omega(1/\poly)$ \cite{Wocjan2003}.
\begin{problem}{$\delta$-LowComplexity-LowEnergyStates}
Input & Description of a $k$-local Hamiltonian $H = \sum_i h_i$, numbers $a$, $b$ and polynomials $T(n)$, $m(n)$, with $b \geq a + \delta$. \\
Output & YES if there exists a circuit of size $ \leq T(n)$ that acts on $\ket{0^m}$ to prepare a state $\ket{\psi}$ with energy $\bra{\psi}H\ket{\psi} \leq a$, \\
& NO if any state $\ket{\psi}$ obtained by applying a circuit of size $T(n)$ on $\ket{0^m}$ has energy $\bra{\psi}H\ket{\psi} \geq b$.
\end{problem}
This latter problem has a weaker promise than \textsc{$(\delta,0)$-GS-Description-LocalHamiltonian}. 
This is because a NO instance of \textsc{$\delta$-LowComplexity-LowEnergyStates} is automatically a NO instance of \textsc{$(\delta,0)$-GS-Description-LocalHamiltonian}, since any state necessarily has energy $\geq b$.
Meanwhile, a NO instance of \textsc{$(\delta,0)$-GS-Description-LocalHamiltonian} need not be a NO instance of \textsc{$\delta$-LowComplexity-LowEnergyStates}, since for the latter there is no guarantee of a circuit to prepare a state with energy close to the ground-state energy.

Despite having a stronger promise on \textsc{$(\delta,0)$-GS-Description-LocalHamiltonian} (which only makes the problem less complex), our small-penalty clock construction allows us to prove the same hardness result for both $\delta=1/\poly$ and $\delta=1/\exp$:

\begin{thm}\label{thm_qcma_complete}
 \textsc{$(1/\poly,0)$-GS-Description-LocalHamiltonian} is $\QCMA$-complete.
\end{thm}
\begin{thm}\label{thm_pqcma_complete}
 \textsc{$(1/\exp,0)$-GS-Description-LocalHamiltonian} is $\pQCMA$-complete.
\end{thm}

For the latter theorem in the precise regime, we use the small-penalty clock construction with an exponentially small energy penalty.
Lastly, when we add the promise of spectral gaps, we have the following results:

\begin{thm}\label{thm_pgqcma_complete} 
 \textsc{$(1/\poly,1/\poly)$-GS-Description-LocalHamiltonian} is $\PGQCMA$-complete.
\end{thm}
\begin{thm}\label{thm_pegqcma_complete}
\textsc{$(1/\exp,1/\exp)$-GS-Description-LocalHamiltonian} is $\pEGQCMA$-complete.
\end{thm}
\begin{thm}\label{thm_ppgqcma_complete}
 \textsc{$(1/\exp,1/\poly)$-GS-Description-LocalHamiltonian} is $\pPGQCMA$-complete.
\end{thm}
The upper bounds in \cref{thm_qcma_complete,thm_pqcma_complete,thm_pgqcma_complete,thm_pegqcma_complete,thm_ppgqcma_complete} follow from a precise version of phase estimation, together with the promise that there is a classical description of a circuit to prepare a low-energy state.
The lower bounds either follow directly through a small-penalty clock construction or through a reduction from a class that contains the relevant class.
\section{Problems characterized by $\PP$} \label{sec_pp}
In this section, we discuss the complexity of the classes $\pPGQMA$ and $\pPGQCMA$, both of which turn out to equal $\PP$.

\begin{thm}\label{thm_ppgqma_pp}
$\pPGQMA = \PP$.
\end{thm}
\begin{thm}\label{thm_ppgqcma_pp}
$\pPGQCMA =\PP$.
\end{thm}

We describe here the overall strategy for proving these results.
First, we adapt the one-bit phase estimation circuit in Ref.~\cite{Fefferman2018} to show that it is possible to compute ground-state energies of sparse Hamiltonians with a spectral gap in the corresponding $\GQMA$ class.
In particular, we have
\begin{lemma} \label{thm_gappedham_in_ppgqma}
\textsc{$(1/\exp,1/\poly)$-LocalHamiltonian} $\in \pPGQMA$.
\end{lemma}

Next, we use the ``power method'' \cite{Mises1929} to give a $\PP$ algorithm for any problem in $\pPGQMA$.
\begin{lemma}[One half of \cref{thm_ppgqma_pp}] \label{thm_ppgqma_inpp}
 $\pPGQMA \subseteq \PP$.
\end{lemma}
\begin{proof}
Suppose we have a $\GQMA[c,s,g_1,g_2]$ instance.
Then we should give a $\PP$ algorithm to precisely compute the maximum eigenvalue $\lambda_1$ of the accept operator $Q$ associated with the instance, under the promise that the spectral gap of $Q$ is bounded below by an inverse polynomial.
In particular, the spectral gap of the accept operator, given by $\lambda_1 - \lambda_2$, is at least $\min[g_1,g_2] =: \Delta$.
Consider the power method to compute the maximum eigenvalue and eigenvector of a positive semidefinite operator $Q$.
This method relies on the observation that upon taking positive powers of the operator $Q$ and estimating its trace, the quantity is dominated by the maximum eigenvalue of $Q$.
In the following, we suppress the dependence of $\lambda_i$ on $Q$:
\begin{align}
\Tr(Q^q) &= \sum_i \lambda_i^q
\\ &= \lambda_1^q\left( 1 + \left(\frac{\lambda_2}{\lambda_1}\right)^q + \ldots \right)
\\ &\leq \lambda_1^q + \lambda_1^q (2^w -1)\left(1 - \frac{\Delta}{\lambda_1}\right)^q,
\end{align}
where $w$ is the size of the witness register.
On the other hand, we have $\Tr(Q^q) \geq \lambda_1^q$.
Therefore, in the YES case, we have
\begin{align}
\Tr(Q^q) \geq c^q,
\end{align}
while in the NO case,
\begin{align}
\Tr(Q^q) &\leq s^q + s^q(2^w -1)\left(1 - \frac{\Delta}{\lambda_1}\right)^q. \label{eq_upperboundeigenvalue}
\end{align}
By the promise of the spectral gap, we must have $\lambda_1 \geq \Delta$, since otherwise the second largest eigenvalue of $Q$ would be $\lambda_2 < 0$.
The difference in the two cases is
\begin{align}
&c^q - s^q - s^q(2^w -1)\left(1 - \frac{\Delta}{\lambda_1}\right)^q
\\& = c^q - s^q - s^q(2^w -1)\exp\left[q\log \left(1 - \frac{\Delta}{\lambda_1}\right)\right]
\\& \geq c^q - s^q - s^q 2^w \exp\left[- \frac{q\Delta}{\lambda_1}\right]
\\ & = s^q\left(\left( 1 + \frac{c-s}{s} \right)^q -1\right)- s^q2^w\exp\left[- \frac{q\Delta}{\lambda_1}\right]
\\& \geq s^q \left(\frac{q(c-s)}{s} - 2^w \exp\left[- \frac{q\Delta}{\lambda_1}\right]\right)
\\& \geq s^q \left({c-s} - 2^w \exp\left[- \frac{q\Delta}{\lambda_1}\right]\right),
\end{align}
since $q \geq 1$ and $s\leq 1$.
If we pick $q = \left\lceil\frac{\lambda_1}{\Delta} \log\left( \frac{2^{w+1}}{c-s} \right)\right\rceil = O(\poly)$, we can ensure that the term $2^w \exp\left[- \frac{q\Delta}{\lambda_1}\right]$ is at most $(c-s)/2$.
Thus the difference in $\Tr(Q^q)$ between the YES and NO cases is at least
\begin{align}
 s^q \frac{c-s}{2} = \Omega(2^{-\poly}).
\end{align}
In the last line above we assumed that $s^q \geq \frac{c^q}{2^{w+1}} = \Omega(2^{-\poly})$ for some polynomial.
In case this assumption is not true, we would nevertheless still have a difference of at least $c^q - 2^w s^q > c^q/2 \geq \Omega(2^{-\poly})$ in between the YES and NO cases when measuring $\Tr(Q^q)$.

This observation suggests that a $\PP$ algorithm can decide between the YES and NO cases by computing $\Tr(Q^q)$ for some large enough polynomial $q$.
This is possible because a $\PP$ algorithm can compute a sum of $2^\poly$ terms, where every term is efficiently computable in polynomial time.
We prove this in \cref{sec_ppdetails} (\cref{lem_cooling_othermatrices}).
\end{proof}

The above result implies that, since \textsc{$(1/\exp,1/\poly)$-LocalHamiltonian} is in $\pPGQMA$, a $\PP$ algorithm can precisely compute ground-state energies of local Hamiltonians with a $\Omega(1/\poly)$ spectral gap.
A similar technique can also be used to show a slightly more general result:
\begin{lemma} \label{lem_gslo_ppp}
Given a local Hamiltonian $H$ and a local observable $A$, along with a promise that $\norm{A} = O(\poly)$ and the spectral gap of $H$ is lower-bounded by $\Omega(1/\poly)$, a $\P^\PP$ algorithm can decide if the ground-state local observable $\expval{A}{E_1}$ is either $\leq a$ or $\geq b$, for $b-a = \Omega(2^{-\poly})$, where $\ket{E_1}$ is the ground state of $H$.
\end{lemma}
\noindent This lemma is proved in \cref{sec_ppdetails}.
Note that both of these results include the case $\Delta = \Theta(1)$, the important case of constant spectral gaps.

We complete the characterization of the power of $\pPGQMA$ with the following result.

\begin{lemma}\label{thm_gappedham_pphard} \raggedright
\textsc{$(1/\exp,1/\poly)$-LocalHamiltonian} is $\PP$-hard.
\end{lemma}
For this proof, we use the small-penalty clock construction, albeit one for the class $\pBQP$ as opposed to the class $\pQMA$.
In this aspect, it resembles the clock construction of Aharonov et al.\ \cite{Aharonov2007}, where it was used to show $\BQP$-universality of the model of adiabatic quantum computing.
We use the technique of applying $\Theta(1/\poly)$ small penalties at the output so as to preserve the lower bound of $\Omega(1/\poly)$ on the spectral gap shown in Ref.~\cite{Aharonov2007}.
In sum, \cref{thm_gappedham_in_ppgqma,thm_ppgqma_inpp,thm_gappedham_pphard} together imply \cref{thm_ppgqma_complete,thm_ppgqma_pp}.

We now come to the class $\pQCMA$ and its complete problem, \textsc{$(1/\exp,0)$-GS-Description-LocalHamiltonian}, where we are promised that there is an efficient circuit to prepare a low-energy state.
We know that $\pQCMA = \NP^\PP$ \cite{Morimae2017a}, which lies in the second level of the counting hierarchy.
Since $\pPGQMA$ is characterized by $\PP$, the promise of having a spectral gap is only slightly stronger than the promise of an efficient circuit to prepare the ground state.

Consider now the gapped version of the problem, \textsc{$(1/\exp,1/\poly)$-GS-Description-LocalHamiltonian}, where there is a $1/\poly$ spectral gap in addition to the promise of an efficient circuit to prepare the ground state.
This characterizes the class $\pPGQCMA$, for which the proof technique is similar to $\pPGQMA$.

We first show that the gapped version of \textsc{GS-Description-LocalHamiltonian} is in the corresponding $\GQCMA$ class, and in particular,
\begin{lemma}\label{lem_in_ppgqcma}
\textsc{$(1/\exp,1/\poly)$-GS-Description-LocalHamiltonian} $\in \pPGQCMA$.
\end{lemma}
$\PP$-hardness of the problem follows by the same argument as the proof of \cref{thm_gappedham_pphard}:
\begin{lemma}\label{thm_pqcma_gap_pphard}
 \textsc{$(1/\exp,1/\poly)$-GS-Description-LocalHamiltonian} is $\PP$-hard.
\end{lemma}
\noindent We give a unified proof of \cref{thm_gappedham_pphard,thm_pqcma_gap_pphard} in \cref{sec_gapreds}.
Since $\pPGQCMA \subseteq \pPGQMA = \PP$, this implies:
\begin{corll}\label{corll_ppgqcma_ppgqma}
 $\pPGQMA = \pPGQCMA = \PP$.
\end{corll}

\section{Problems characterized by \PSPACE} \label{sec_pspace}
In this section, we discuss the complexity of the class $\pEGQMA$, which turns out to equal $\PSPACE$.
This result indicates that the complexity of the Local Hamiltonian problem does not jump immediately in the presence of a tiny, nonzero spectral gap.
This means that there is a notion of robustness of the complexity of the problem with respect to the spectral gap.

\begin{thm}\raggedright
$\pEGQMA = \pQMA$\ $(=\PSPACE)$. \label{thm_pegqma_pspace}
\end{thm}
\begin{proof}
The containment $\pEGQMA \subseteq \pQMA$ follows trivially since any $\pEGQMA$ instance is automatically a $\pQMA$ instance.
We show the other direction, $\pEGQMA \supseteq \pQMA$, in two steps.
Our proof relies on the complexity of the following problem:
\begin{problem}{SparseHamiltonian[$a,b,g_1,g_2$]}
Input & A succinct description of a Hermitian matrix of size $2^{\poly(n)} \times 2^{\poly(n)}$, with at most $d=\poly(n)$ many entries in each row and two numbers $a$ and $b$, with $b>a$.
The magnitude of each entry is bounded by $k= \poly(n)$.\\
Output & YES if the smallest eigenvalue $E_1 \leq a$ and the spectral gap of the matrix is at least $g_1$,\\
& NO if $E_1 \geq b$, and the spectral gap of the matrix is at least $g_2$.
\end{problem}
We define \textsc{$(\delta,\Delta)$-SparseHamiltonian} to be $\cup_{\substack{b-a \geq \delta\\ g_1,g_2 \geq \Delta}}$\textsc{SparseHamiltonian[$a,b,g_1,g_2$]} and consider the problem with parameters $\delta,\Delta =\Omega(1/\exp)$.
First, in \cref{lem_sparseHamPspaceHard}, we prove that \textsc{$(1/\exp,1/\exp)$-SparseHamiltonian} is \PSPACE-hard, or equivalently $\pQMA$-hard.
Next, we show in \cref{lem_sparseHam_in_pEGQMA} that \textsc{$(1/\exp,1/\exp)$-SparseHamiltonian} may be solved in $\pEGQMA$.
The theorem then follows.
\end{proof}

\begin{lemma}\label{lem_sparseHamPspaceHard}
 \textsc{$(1/\exp,1/\exp)$-SparseHamiltonian} is \PSPACE-hard.
\end{lemma}
The reduction is from any problem in $\PSPACE$ to an instance of co-\textsc{$(1/\exp,1/\exp)$-Gapped-\\SparseHamiltonian}, which is the complement of the problem, in the sense that the YES and NO instances are reversed.
Since $\PSPACE$ is closed under complement, this still gives the desired hardness result.
The broad idea is to represent a $\PSPACE$ computation as an exponentially large, but sparse, graph.
The smallest eigenvalue of the adjacency matrix of this graph encodes information about whether the computation accepts or rejects.
\begin{proof}
We use a proof technique adapted from an unpublished manuscript by Fefferman and Lin \cite{Fefferman2016c}.
First, we use the fact that $\PSPACE$ with reversible operations in every step still equals $\PSPACE$: $\mathsf{revPSPACE} = \PSPACE$ \cite{Bennett1989}.
Indeed, it is known that $\SPACE[\mathfrak{s}(n)] = \mathsf{revSPACE}[\mathfrak{s}(n)]$ \cite{Lange2000} with an overhead in time that is exponential in the space, $\mathfrak{s}(n)$.
Let $t(n)$ be this upper bound on the running time of the Turing machine, so that we can restrict our attention to the class $\mathsf{revSPACE}[\mathfrak{s}(n)] \cap \TIME[t(n)] = \SPACE[\mathfrak{s}(n)]$.
Any computation on a reversible Turing machine may be viewed as traversing a directed \emph{configuration graph}, where each vertex of the graph is determined by the state of the head and the list of symbols on the input and work tapes (\cref{fig_pspacegraph_oyes,fig_pspacegraph_ono}).
When such a Turing machine is restricted to use space polynomial in the input length $n$, the number of vertices in the graph is upper bounded by an exponential, $2^{\poly(n)}$.
Consider the adjacency matrix of the graph, $A_x$.
The description of this exponentially large matrix is succinct because it only requires specifying the input $x$ and the rules of the Turing machine.

We modify the configuration graph $G_x \to G_x'$ so that the smallest eigenvalue of the matrix ${A_x^\dag} ' A_x'$ is 0 in the NO case and bounded away by an exponentially small amount in the YES case.
We do this modification in a way that ensures the matrix has a spectral gap lower bound of at least $\Omega(1/\exp)$.
This is done as follows.
First, we modify the configuration graph of the Turing machine by adding self-loops to all vertices except for the start and accept configurations $s_x$ and $t_x$.
We then add a sequence of vertices $\{1, 2, \ldots t(n)\}$ from the accept configuration $t_x$, with the directed edges $t_x \to 1 \to 2 \to \ldots \to t(n) \to s_x$, as shown in \cref{fig_pspacegraph_myes,fig_pspacegraph_mno}.
The adjacency matrix of this modified directed graph $G_x'$ is $A_x'$, and we are interested in the eigenvalues and spectral gap of ${A_x^\dag}' A_x'$, which is Hermitian and sparse, and also has a succinct representation.

\begin{figure*} \centering
\begin{subfigure}{0.45\textwidth}
 \includegraphics[width=\textwidth]{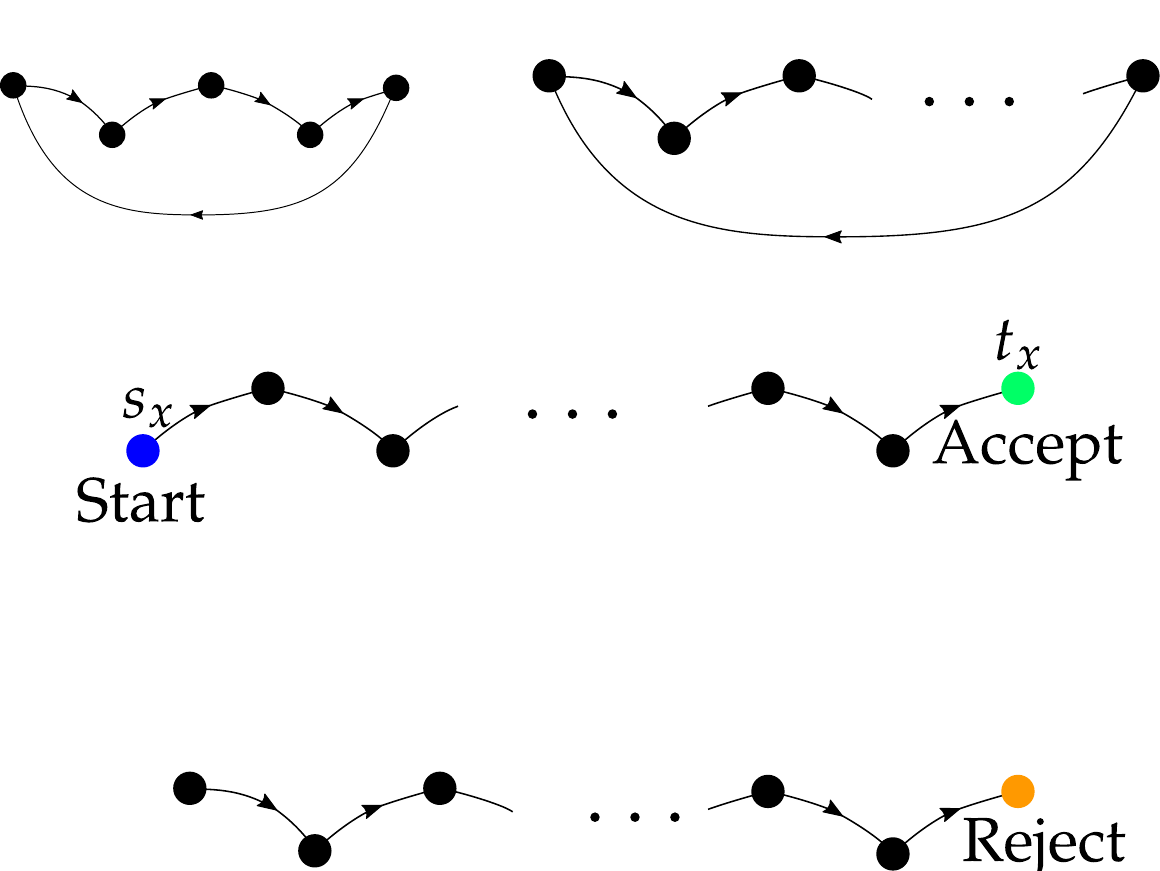}
 \caption{YES case, original graph}\label{fig_pspacegraph_oyes}
\end{subfigure}
\hfill
\begin{subfigure}{0.45\textwidth}
 \includegraphics[width=\textwidth]{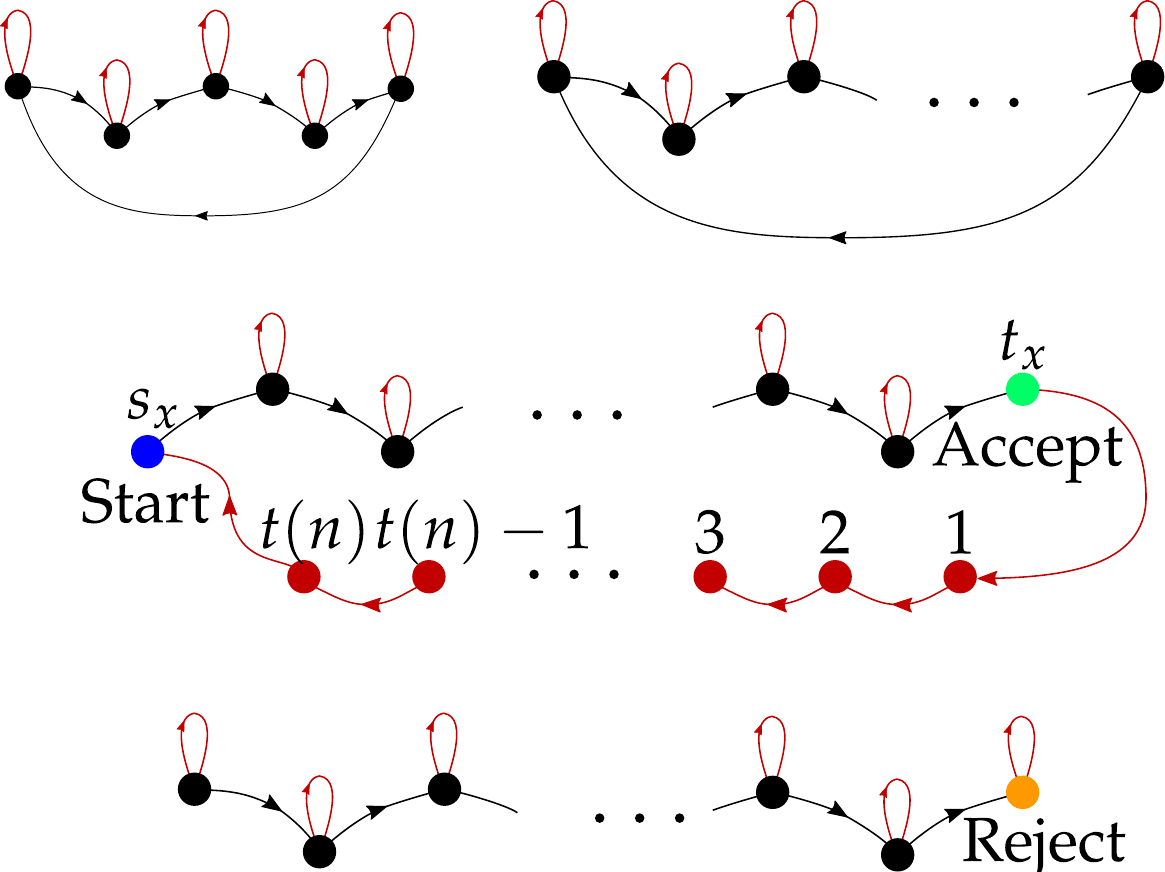}
 \caption{YES case, modified graph}\label{fig_pspacegraph_myes}
\end{subfigure}
\begin{subfigure}{0.45\textwidth}
 \includegraphics[width=\textwidth]{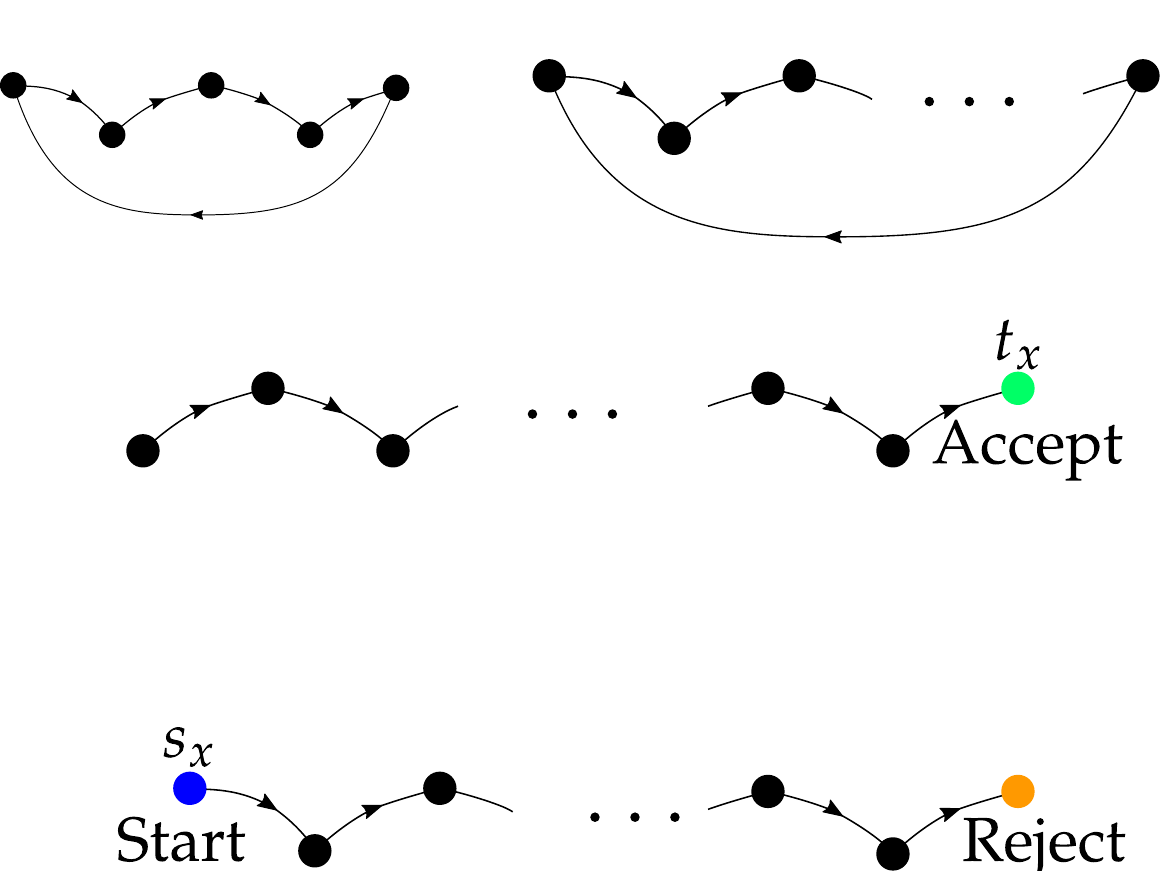}
 \caption{NO case, original graph}\label{fig_pspacegraph_ono}
\end{subfigure}
\hfill
\begin{subfigure}{0.45\textwidth}
 \includegraphics[width=\textwidth]{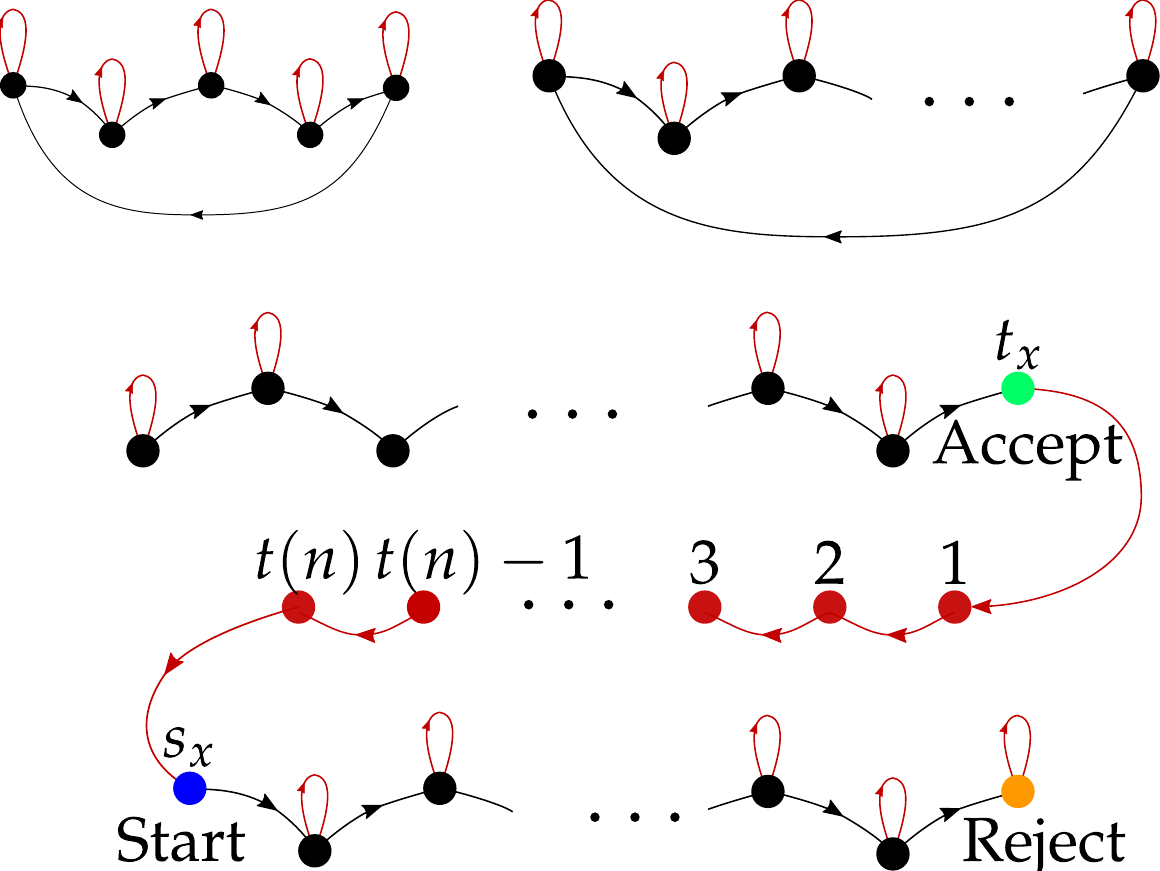}
 \caption{NO case, modified graph}\label{fig_pspacegraph_mno}
\end{subfigure}
 \caption{\raggedright Schematic of the original and modified graphs for both YES and NO cases.
 The original graph in both YES and NO cases consists of vertices with in-degree and out-degree at most 1, due to the fact that the Turing machine is reversible.
 The start vertex $s_x$ is marked in blue, the accept vertex $t_x$ in green, and the reject vertex in orange.
 The modified graphs have self-loops on all vertices except the start and the accept vertices.
 They have additional vertices $1,2,\ldots t(n)$ without self-loops.
 All modifications are in maroon.} \label{fig_pspacegraph}
\end{figure*}

We now analyze this construction.
The proof relies on an explicit computation of the eigenvalues for the various subgraphs of the modified configuration graph.
In the NO case, the graph $G_x'$ has a path of vertices ending in the reject state (\cref{fig_pspacegraph_mno}).
This path contains the starting configuration $s_x$.
Let $\ell$ be the graph distance between $s_x$ and the reject state.
Since we have added the edges $t_x \to 1 \to \ldots t(n) \to s_x$, these vertices and the vertices leading to the accept state are also part of the path (the Turing machine does not explore these vertices in practice).
All vertices in this path except for $t_x$, $s_x$, and $i: i\in [t(n)]$ have self-loops on them.
As we show in \cref{lem_expsmallgapNOcase}, there is a zero eigenvalue in the NO case, with a spectral gap above the zero eigenvalue.
The spectral gap is lower bounded by $\Omega(1/\ell_{\max}^2) = \Omega(2^{-\poly})$, where $\ell_{\max}$ is the number of vertices in the longest subgraph.

In the YES case, the subgraph containing the starting vertex is a cycle, with self-loops on all vertices except for $t_x$, $s_x$, and the intermediate vertices $i$. 
In each case, the eigenvalues for any subgraph are given by $2-2\cos \frac{(2k-1)\pi}{2 \ell + 1} = 4 \sin^2 \left( \frac{(2k-1)\pi}{4 \ell + 2}\right)$, $k\in [\ell]$ \cite{Fefferman2016c}, where $\ell$ is the number of vertices in the subgraph.
The smallest eigenvalue is therefore given by the longest subgraph and this eigenvalue is nondegenerate if no two subgraphs have the same number of vertices.
This is why we have added the sequence of edges $t_x \to 1 \to \ldots t(n)$.
The role played by these vertices is to elongate the length of the subgraph containing the start and accept configurations by $t(n)$.
This ensures that no other subgraph has a length equal to the longest subgraph (since $t(n)$ is the upper bound on the total number of vertices in the graph before elongation).
Therefore, the smallest two eigenvalues are given by $4 \sin^2 \left( \frac{(2k-1)\pi}{4 \ell + 2}\right)$, which are separated by $\Theta(t(n)^{-2}) = \Theta(2^{-\poly})$.

To summarize, in the YES case we have $E_1 \geq 2^{-\poly}$ and $E_2 - E_1 \geq 2^{-\poly}$.
In the NO case, we have $E_1 = 0$ and $E_2 \geq 2^{-\poly}$.
Therefore, we have a promise gap of $2^{-\poly}$ and spectral gap $2^{-\poly}$ in both the YES and NO instances.
Furthermore, the matrix ${A_x^\dag}'A_x'$ has entries of magnitude at most 2, and is 3-sparse because of the bounded degree of the configuration graph.
Since the minimum eigenvalue is small in the NO case and large in the YES case, we have a reduction to co-\textsc{$(1/\exp,1/\exp)$-SparseHamiltonian}.
Due to the fact that $\PSPACE$ is closed under complement, we get $\PSPACE$-hardness of \textsc{$(1/\exp,1/\exp)$-SparseHamiltonian}.
\end{proof}

\begin{lemma}\label{lem_sparseHam_in_pEGQMA}
\textsc{$(1/\exp,1/\exp)$-SparseHamiltonian} $\in \pEGQMA$.
\end{lemma}
The proof of this is mostly the same as the proof of containment of \textsc{($1/\exp,1/\exp$)-LocalHamiltonian} in $\pEGQMA$ and is also given in \cref{sec_ingqma}.
The only difference is that we have a sparse Hamiltonian instead of a local Hamiltonian.
This distinction turns out not to matter, however, because of quantum algorithms for Hamiltonian evolution that work well with sparse Hamiltonians \cite{Berry2014}.

\section{Other related classes} \label{sec_implications}
In this section, we discuss implications of our proof techniques for other complexity classes.
The first concerns a technique for amplifying the promise gap in $\QMA$ and related classes, called in-place amplification, due to Marriott and Watrous \cite{Marriott2005}.
The second is about the complexity of related classes when the spectral gap promise only applies to one kind of instance (YES instances, for example).
We also complete a discussion of the results in \cref{tab_complexitylh} by characterizing the complexity classes $\PGQCMA$, $\EGQCMA$, and $\pEGQCMA$.

\subsection{Amplification for \postQMA}
We first define the class $\postQMA$:
\begin{defn}[$\postQMA$]
$\postQMA[c,s]$ is the class of promise problems $A=(A_\mathrm{yes},A_\mathrm{no})$ that can be decided in the following way:
Apply a uniformly generated quantum circuit $U$ of size $\poly(n)$ on a state $\ket{x}$ encoding the input, together with a proof state of size $w(n)$ supplied by an arbitrarily powerful prover.
Postselect the first $l=\poly(n)$ qubits at the output onto the $\ket{0}^l$ state, and measure the first qubit of the remaining register at the output, called the decision qubit ($o$).
The postselection probability is $\Omega(2^{-f(n)})$ for a polynomial $f(n)$. \\
\begin{tabular}{l l}
 If $x \in A_\mathrm{yes}$: & $\exists \ \ket{\psi}$ such that $\Pr(o = 1) \geq c$ \\
 If $x \in A_\mathrm{no}$: & $\forall \ \ket{\psi}$, $\Pr(o = 1) \leq s$.
\end{tabular}
\end{defn}
Morimae and Nishimura \cite{Morimae2017a} defined this class and showed that $\postQMA := \postQMA[\frac{1}{3},\frac{2}{3}] = \pQMA = \PSPACE$.
This result is similar to the result $\postBQP=\pBQP (=\PP)$.
They raised the question of whether one can do a Marriott-Watrous type in-place amplification for this class, which, for instance, means boosting the parameters $c$ and $s$ to be $c = 1-2^{-\poly}$, $s = 2^{-\poly}$ without changing the size of the witness.
If one is allowed to change the witness size, one can simply ask for polynomially many copies of the witness and run the verification in parallel to get the required parameters.
The benefit of in-place amplification is that it allows for good completeness and soundness parameters without blowing up the witness size, which turns out to be useful in the proof of $\QMA \subseteq \PP$.
In-place amplification for $\postQMA$ would also be useful to show $\IP=\PSPACE$ \cite{Aharonov2017a,Green2019}.
Here we give a negative result for a sufficiently strong in-place amplification for $\postQMA$.

\begin{lemma}[Upper bound for in-place amplified $\postQMA$] \label{lem_postqma_amp}
If $f(n) = O(w(n))$, then $\postQMA[1 - 2^{-t(n)},2^{-u(n)}] \subseteq \PP$ for $u(n) > w(n) + 1$ and for any polynomial $t(n)>1$.
\end{lemma}
\begin{proof}
Consider a $\postQMA[1 - 2^{-t(n)},2^{-u(n)}]$ language.
Replace the witness state in the amplified protocol by a maximally mixed state $\frac{\mathds{1}}{2^w}$.
Now, since the overlap of any witness state with the maximally mixed state is $2^{-w}$, we have that the postselection success probability is at least $\Omega(2^{-f(n)-w(n)})$.
Further, in the YES case, the probability of accepting the string $x$ (conditioned on success) is
\begin{align}
 \Pr(o =1) \geq 2^{-w(n)} \times (1-2^{-t(n)}).
\end{align}
In the NO case, we have that no matter what state is in the witness register, the accept probability is 
\begin{align}
 \Pr(o =1) \leq 2^{-u(n)}.
\end{align}
In $\pBQP$ = \PP, we can distinguish between these two cases if $2^{-w} -2^{-t-w} > 2^{-u}$, i.e.\ if $1 - 2^{-t} > 2^{w-u}$, for which it suffices to have $u(n) > w(n)+1$ and $t>1$.
\end{proof}

This result implies that the completeness-soundness gap for $\postQMA$ cannot be boosted beyond a point without incurring a blowup in the size of the witness or by reducing the success probability of postselection.

\subsection{Asymmetric promises on spectral gap and uniqueness}
Motivated by a possible connection to the study of unique witnesses for quantum complexity classes, we consider the complexity class $\GQMA[c,s,g_1,0]$.
Here, there is no promise on the spectral gap for NO instances.
In the YES case, we have $\lambda_1(Q) \geq c$ and $\lambda_2 \leq \lambda_1-g_1 \leq 1-g_1$.
If we choose the spectral gap $g_1$ to be larger than $1-s$, we see that $\lambda_2 \leq s$, ensuring that in the YES case, there is exactly one accepting witness\footnote{In the sense that any witness orthogonal to the accepting witness rejects with probability at least $1-s$.}.
The existence of one accepting witness is exactly the promise that defines the class $\UQMA$:
\begin{defn}[Unique $\QMA$ \cite{Aharonov2008}]
$\UQMA[c,s]$ is the class of promise problems $A=(A_\mathrm{yes},A_\mathrm{no})$ such that for every instance $x$, there exists a polynomial-size verifier circuit $U_x$ acting on $m=\poly(n)$ qubits and an input quantum proof on $w=\poly(n)$ qubits and the associated accept operator $Q$ has properties\\
\begin{tabularx}{\linewidth}{l X}
If $x \in A_\mathrm{yes}$: & $\lambda_1(Q) \geq c$ and $\lambda_2(Q) \leq s$ \\
If $x \in A_\mathrm{no}$: & $\lambda_1(Q) \leq s$.
\end{tabularx}
\end{defn}
\begin{defn}
$\UQMA := \cup_{c-s \geq 1/\poly} \UQMA[c,s]$.
\end{defn}

The earlier statement can be rephrased as ``an instance of $\GQMA[c,s,1-s,0]$ is a $\UQMA[c,s]$ instance''.
In the reverse direction, we can see that a $\UQMA[c,s]$ instance necessarily has a spectral gap $\lambda_1 - \lambda_2 \geq c-s$, and therefore is an instance of $\GQMA[c,s,c-s,0]$.
This hints at, but does not prove, an equivalence between the promise of uniqueness and that of an asymmetric spectral gap of $\Omega(1/\poly)$.
Aharonov et al.\ \cite{Aharonov2008} proved a stronger result by showing that the class $\UQMA$ is equivalent to the class $\PGQMA$ under randomized reductions (defined below), where $\PGQMA$ is the class with spectral gaps for \emph{both} the YES and the NO cases.

In the precise regime, we show the following results for the asymmetric variants of $\pPGQMA$ and $\pEGQMA$.
\begin{thm}\label{thm_asym_symgaps}\raggedright
$\pPGQMA = \cup_{\substack{c-s\geq 1/\exp,\\ g_1 \geq 1/\poly}} \GQMA[c,s,g_1,0]$.\\
$\pEGQMA = \cup_{\substack{c-s\geq 1/\exp,\\ g_1 \geq 1/\exp}} \GQMA[c,s,g_1,0]$.
\end{thm}
The proofs are given in \cref{sec_asymcomplexity} and hinge on the problem of computing ground-state energies when there is a spectral gap only for the YES case, i.e.\ \textsc{LocalHamiltonian[$a,b,g_1,0$]}.
Since the problem with an asymmetric gap can only be more complex than the symmetric case, the nontrivial part of this lemma is to show that this problem has the same $\PP$ upper bound as the symmetric case.
This is not straightforward since the power method we described before does not necessarily work for the NO case, since there is no spectral gap.
We work around this by making use of Ambainis's technique \cite{Ambainis2014} of identifying spectral gaps, which is possible in $\PP$ \cite{Gharibian2019b}.

\subsection{Complexity of $\PGQCMA$, $\EGQCMA$, and $\pEGQCMA$} \label{sec_gqcma}
In this subsection we show that the classes $\PGQCMA$ and $\EGQCMA$ are both equivalent to $\QCMA$ under randomized reductions, which we now define.

We say a problem $A$ is random reducible to problem $X$ if every instance $a$ of $A$ can be mapped to a random set of polynomially instances $x_i$ of $X$, such that
\begin{tabularx}{\linewidth}{l X}
If $a \in A_\mathrm{yes}$: & $\Pr_i(x_i \in X_\mathrm{yes}) \geq 1/\poly$ \\
If $a \in A_\mathrm{no}$: & $\Pr_i(x_i \in X_\mathrm{yes}) = 0$.
\end{tabularx}
A class $\mathsf{Y}$ is random reducible to another class $\mathsf{Z}$ if every problem in $\mathsf{Y}$ is random reducible to some problem in $\mathsf{Z}$ (and vice versa), and is denoted $=_R$.

To show $\PGQCMA=_R \QCMA$ and $\EGQCMA =_R \QCMA$, we make use of the class $\UQCMA$ (Unique $\QCMA$), which has been defined in Ref.~\cite{Aharonov2008}, and was shown to be equal to $\QCMA$ under randomized reductions.

\begin{defn}[{$\UQCMA[c,s]$} \cite{Aharonov2008}]
$\UQMA[c,s]$ is the class of promise problems $A=(A_\mathrm{yes},A_\mathrm{no})$ such that for every instance $x$, there exists a polynomial-size verifier circuit $U_x$ acting on $m=\poly(n)$ qubits and an input classical proof on $w=\poly(n)$ qubits, whose associated accept operator $Q$ has properties\\
\begin{tabularx}{\linewidth}{l X}
If $x \in A_\mathrm{yes}$: & $\lambda_1(Q) \geq c$ and $\lambda_2(Q) \leq s$ \\
If $x \in A_\mathrm{no}$: & $\lambda_1(Q) \leq s$.
\end{tabularx}
\end{defn}
\begin{defn}
$\UQCMA := \cup_{c-s \geq 1/\poly} \UQCMA[c,s]$.
\end{defn}

Aharonov et al.~\cite{Aharonov2008} showed that $\UQCMA =_R \QCMA$ using generalizations of techniques in Ref.~\cite{Valiant1986} to complexity classes with randomness.
In order to show $\PGQCMA=_R \QCMA$ and $\EGQCMA =_R \QCMA$, we show
\begin{lemma}
$\PGQCMA =_R \UQCMA$. \label{lem_uqcma}
\end{lemma}
Since $\PGQCMA \subseteq \EGQCMA \subseteq \QCMA$, the equivalence of $\EGQCMA$ with $\QCMA$ follows.

To show \cref{lem_uqcma}, we observe that the proof of $\PGQMA =_R \UQMA$ in Ref.~\cite{Aharonov2008} works for classical witnesses.
For completeness, we give a self-contained proof here.

\begin{proof}[Proof of \cref{lem_uqcma}]
First, we show the direction $\UQCMA \subseteq \PGQCMA$.
We observe that in a YES instance of $\UQCMA[c,s]$, $\lambda_1 \geq c$ and $\lambda_2 \leq s$. 
Thus, a YES instance already has a spectral gap of $g_1 \geq c-s$ and is a YES instance of $\PGQCMA$.
In the NO case, we modify the verifier's strategy so that it creates a spectral gap.
The verifier expects an additional qubit we call the ``flag qubit'' from the prover, which is measured in the beginning just like the other qubits of any $\QCMA$ proof.
The associated accept operator now has twice as many eigenvalues because it acts on a space with one larger qubit.

The verifier's protocol is as follows.
If the state of the flag qubit is $\ket{0}$, the verifier continues with the original protocol.
This gives the same eigenvalues for the accept operator as the original protocol.
If the state of the flag qubit is $\ket{1}$, the verifier accepts with probability $s + (c-s)/\poly$ if the state of the rest of the witness qubits is $\ket{1}^{\otimes w}$.
If the state of the rest of the witness register is anything else, the verifier rejects.
In the latter case (when the state of the flag qubit is $\ket{1}$), the accept operator has one eigenvalue at $s + (c-s)/\poly$ and $2^w - 1$ eigenvalues with eigenvalue 0, each case corresponding to some state in the witness.
The modified verifier is a $\PGQCMA$ instance with completeness $c$, soundness $s + (c-s)/\poly$ and spectral gaps $g_1 \geq c-s$ and $g_2 \geq (c-s)(1-1/\poly)$.
Therefore $\UQCMA \subseteq \PGQCMA$.

For the other direction, we give a randomized reduction $\PGQCMA \subseteq_R \UQCMA$.
Consider a YES instance of $\PGQCMA[c,s,g_1,g_2]$.
We know $\lambda_1 \geq c$ and $\lambda_2 \leq \lambda_1 - g_1$, but we do not know if $\lambda_2 \leq s$, as is required for the instance to be a $\UQCMA$ instance.
The idea in Ref.\ \cite{Aharonov2008} is to make a query to a $\UQCMA[c_j,s_j]$ oracle with completeness $c_j = c + (j+1)g_1/2$ and soundness $s_j = c + jg_1/2$, for $j$ chosen randomly from $ \{0,1,\ldots \lfloor{\frac{2}{(1-c)g_1}}\rfloor \}$.
In the NO case, all the queries are valid queries to a $\UQCMA$ oracle and return the correct answer (NO).
In the YES case, since the completeness and soundness in each query differ by $g_1/2$, there is at least one $j$ where $\lambda_1 \geq c_j$ and $\lambda_2 \leq s_j$\footnote{In the YES case, there could be some queries that are not valid $\UQCMA$ instances, and the oracle can answer arbitrarily for such ill-formed queries. This does not, however, hamper the proof, since at single valid query is enough to give a nonzero probability of saying YES.}.
Therefore, this is a randomized reduction to $\UQCMA$.
\end{proof}

Therefore, we obtain
\begin{corll} \label{lem_egqcma}
$\PGQCMA =_R \QCMA$. \\
$\EGQCMA =_R \QCMA$.
\end{corll}

Our final result concerns the class $\pEGQCMA$.
Just like we have $\pEGQMA=\pQMA$, we can show that exponentially small spectral gaps are no less complex in the case of classical witnesses.
We show
\begin{lemma} \label{lem_pegqcma}
$\pEGQCMA = \pQCMA$.
\end{lemma}
\begin{proof}
The direction $\pEGQCMA \subseteq \pQCMA$ is trivial.
For the other direction, we take a $\pQCMA[c,s]$ instance and give a $\pEGQCMA[c,s,g_1,g_2]$ instance with an exponentially small spectral gap.
This is done by modifying the verifier so that no two witnesses $y_i$ and $y_j$ are accepted with the same probability.
First, we choose the verifier's gate set so that the accept probability of any witness $y$ is given by $k_{x,y}/2^{l(n)}$, for $k_{x,y} \in [2^{l(n)}]$, where $l(n)$ is the size of the verifier's circuit \cite{Jordan2012}.
The modified verifier rejects the instance straightaway with probability $y_b/2^\poly$, where $y_b$ is a number in $[2^w-1]$ when interpreting the witness $y$ in binary and the polynomial is at least $l(n) + w(n) + \log_2(\frac{1}{c-s})$. 
If the verifier does not reject at this step, they run the original verification protocol.
The overall accept probability when given $y$ is given by $p_y = \frac{k_{x,y}}{2^w} \left(1-\frac{y_b}{2^\poly}\right)$.
Since the polynomial satisfies $\poly \geq l(n)+w(n)+\log_2(\frac{1}{c-s})$, the completeness and soundness are given by $c' \geq c - 2^{-w(n)}(c-s)$ and $s' = s$, which are still separated by $2^{-\poly}$.

We now claim that the resulting accept probabilities are distinct for distinct witnesses, and hence separated by an amount $\Omega(2^{-\poly})$.
This is easily seen for two distinct $y_i$ and $y_j$ such that $k_{x,y_i} = k_{x,y_j}$.
If $k_{x,y_i} \neq k_{x,y_j}$, then for $p_{y_i} = p_{y_j}$, we need
\begin{align}
k_{x,y_i} - k_{x,y_j} = \frac{2^w}{2^{l + w + \poly}}(y_{j_b}-y_{i_b}),
\end{align}
which cannot be satisfied by integers $y_{j_b}$ and $y_{i_b}$ in $[2^l]$.
\end{proof}

The same technique also works to give a more direct proof of $\EGQCMA=\QCMA$.

\appendix

\section{The Schrieffer-Wolff transformation} \label{sec_schwolff}
In this section, we give a brief introduction to the Schrieffer-Wolff transformation \cite{Schrieffer1966}, which is an important tool in some of our subsequent proofs.
We follow the exposition in Ref.~\cite{Bravyi2011}, specialized to our context.

In the context relevant for us, we usually have an ``unperturbed'' Hamiltonian $H_0$ and a ``perturbation'' $H_1$, together forming the full Hamiltonian $H = H_0 + H_1$.
The (possibly degenerate) ground-state subspace of $H_0$, denoted $\mathcal{S}_0$, has energy $\lambda_0$ and is separated from the rest of the spectrum by a gap $\Delta$.
We are interested in cases when the Hamiltonian $H_1$ has small strength relative to the gap $\Delta$, in the sense $\norm{H_1} =: \epsilon < \Delta/2$.
This ensures that all eigenvalues of $H_0$ are shifted by an amount smaller than $\Delta/2$ under the perturbation.
Therefore, the low-energy subspace of $H$, given by
\begin{align}
 \mathcal{S} = \left\{\ket{\psi} : \bra{\psi} H \ket{\psi} \in \left[\lambda_0 - \frac{\Delta}{2}, \lambda_0 + \frac{\Delta}{2}\right]\right\},
\end{align}
has the same dimension as that of $H_0$.
We denote the the projectors on to $\mathcal{S}_0$ and $\mathcal{S}$ by $P_0$ and $P$, respectively.
As long as $\epsilon < \Delta/2$, we have $\norm{P-P_0} < 1$, which captures the fact that the dimension of the two subspaces is the same.

Since the dimension of the two subspaces is the same, there exists a unitary $U$ that maps the subspace $\mathcal{S}_0$ to $\mathcal{S}$:
\begin{align}
UPU^\dag &= P_0, \ \mathrm{with \ } \\
U &= \sqrt{(2P_0 - \mathds{1})(2P -\mathds{1})}.
\end{align}
We are interested in the effective Hamiltonian in the subspace $\mathcal{S}_0$, given by
\begin{align}
 H_\mathrm{eff}= P_0 U(H_0 + H_1)U^\dag P_0.
\end{align}
The Schrieffer-Wolff transformation allows one to express the generator $V = \log(U)$, and consequently, $H_\mathrm{eff}$, as a convergent series in the perturbation $H_1$.
We first write $H_1$ as $H_1^d + H_1^o$, where $H_1^d$ is block-diagonal in the subspace $\mathcal{S}_0$ and $H_1^o$ is block-off-diagonal.
Let the eigenstates of $H_0$ be given by $\{\ket{i}\}$, with corresponding energies $\{E_i\}$.
We denote $\mathcal{I}_0 = \{i: E_i = \lambda_0 \}$, which is the set of indices corresponding to the ground-state space.
The first few terms of the Schrieffer-Wolff expansion are given by
\begin{align}
 &H_\mathrm{eff} = H_0 P_0 + P_0 H_1 P_0\ + \nonumber
 \\&\frac{1}{2} P_0 \sum_{i \in \mathcal{I}_0,j \notin \mathcal{I}_0} \left(\frac{\bra{i} H_1 \ket{j}}{E_i - E_j} \ketbra{i}{j} H_1 + \frac{\bra{j} H_1 \ket{i}}{E_i - E_j} H_1 \ketbra{j}{i} \right)P_0 \nonumber
 \\&+ O(\norm{H_1}^3).
\end{align}

In our work, we use the first-order expansion of the Schrieffer-Wolff series.
The series converges absolutely as long as $\norm{H_1} \leq \Delta/16$ \cite{Bravyi2011}.
We can upper bound the error caused by truncating the formal series to first order:
\begin{align}
& \norm{H_\mathrm{eff} - H_0P_0 - P_0 H_1 P_0} \leq O(1) \times \nonumber
\\& \Biggl\lVert P_0 \sum_{i \in \mathcal{I}_0,j \notin \mathcal{I}_0} \left(\frac{\bra{i} H_1 \ket{j}}{E_i - E_j} \ketbra{i}{j} H_1\ + \right.\nonumber \left. \frac{\bra{j} H_1 \ket{i}}{E_i - E_j} H_1 \ketbra{j}{i} \right)P_0 \Biggr\rVert  \\
&\leq O(1) \Biggl\lVert \sum_{i \in \mathcal{I}_0,j \notin \mathcal{I}_0, k \in \mathcal{I}_0}\frac{1}{E_i - E_j}  \left( \bra{i}H_1\ket{j} \bra{j}H_1\ket{k} \ketbra{i}{k}\ + \right. \nonumber
\\& \left. \bra{j} H_1 \ket{i} \bra{k}H_1\ket{j} \ketbra{k}{i}\right) \Biggr\rVert \\
&\leq O\left(\frac{1}{\Delta}\right) \norm{\sum_{i \in \mathcal{I}_0, k \in \mathcal{I}_0}  \left( \bra{i}H_1^2\ket{k} \ketbra{i}{k} + \bra{k} H_1^2 \ket{i} \ketbra{k}{i}\right)} \\
&= O\left(\frac{1}{\Delta}\right) \norm{2P_0 H_1^2 P_0} \\
&\leq O\left( \frac{\epsilon^2}{\Delta} \right),
\end{align}
where we have used $\abs{E_i-E_j} > \Delta$ for states ${i} \in \mathcal{I}_0, {j} \notin \mathcal{I}_0$.

\section{Modified clock constructions with spectral gaps} \label{sec_gapreds}
In this section, we present the small-penalty clock construction and use it to prove the main hardness results in this work.
We first illustrate the technique by proving the following lemma.
\begin{lemma}
\label{lem_ham_preciseegqma}
\textsc{$(1/\exp,1/\exp)$-LocalHamiltonian} is $\pEGQMA$-hard.
\end{lemma}

\begin{proof}
Consider a $\GQMA[c,s,g_1,g_2]$ instance $x$, where the verifier's circuit $U_x$ acts on $m = \poly(n)$ qubits apart from the proof state.
We assume that the circuit has $T = \poly(n)$ gates.
The idea behind the technique is valid generally, but for concreteness we focus on the clock construction of Kempe et al.\ \cite{Kempe2003}, which proves $\QMA$-hardness of \textsc{$k$-LocalHamiltonian} for $k\geq 3$.
The clock Hamiltonian takes the form
\begin{align}
H = H_\mathrm{input} + H_\mathrm{prop} + H_\mathrm{output} + H_\mathrm{clock}. \label{eq_clockham}
\end{align}
The first term $H_\mathrm{input}$ ensures that the ground state of $H_\mathrm{input}$ coincides with input state to the circuit.
The term on the proof register is identity, allowing for any witness state given by the prover to be input into the verifier's circuit.
It is given by 
\begin{align}
&H_\mathrm{input} = \sum_{i=1}^m \ketbra{1}_i \otimes \mathds{1}_\mathrm{proof} \otimes H_\mathrm{clockinit}\label{eq_clockhaminput}.
\end{align}
In the above, the term $H_\mathrm{clockinit}$ ensures that the clock is properly initialized to the $\ket{1}_\mathrm{clock}$ state.
Next, $H_\mathrm{prop}$ is a Hamiltonian that ensures the ground state is ``propagated'' correctly with each gate applied by the verifier:
\begin{align}
& H_\mathrm{prop} = \sum_{i=0}^T - U_{i+1} \otimes \ketbra{i+1}{i}_\mathrm{clock} -U_{i+1}^\dag \otimes \ketbra{i}{i+1}_\mathrm{clock} \nonumber
\\& + \mathds{1} \otimes (\ketbra{i}_\mathrm{clock} + \ketbra{i+1}_\mathrm{clock}). \label{eq_clockhamprop}
\end{align}
The ground-state subspace of $H_\mathrm{prop}$ contains valid ``partial'' computations until step $i \leq T$, namely $U_i \ldots U_2 U_1 \ket{\psi_0}$ on any initial state $\ket{\psi_0}$ $\forall\ i$.
The term $H_\mathrm{output}$ penalizes states that have any nonzero probability of saying ``NO'' at the output qubit $o$ of the circuit:
\begin{align}
 &H_\mathrm{output} = \epsilon \ketbra{0}_o \otimes \ketbra{T}_\mathrm{clock} \label{eq_clockhamoutput}.
\end{align}
Lastly, $H_\mathrm{clock}$ ensures that states in the clock register that do not encode a valid time step are penalized.
The Hamiltonians $H_\mathrm{clock}$ and $H_\mathrm{clockinit}$ both depend on the details of the particular clock construction.
Our analysis does not depend on these details is largely independent of the way the clock register encodes the time.
We refer the reader to Ref.~\cite{Kempe2003} for an explanation of their construction.

First consider just the Hamiltonian $H_0 = H_\mathrm{input} + H_\mathrm{prop} + H_\mathrm{clock} $, which is the clock Hamiltonian without a penalty term at the output.
The ground-state space of $H_0$ is exactly given by the subspace $\mathcal{S}_0$ of history states:
\begin{align}
&\mathcal{S}_0 = \mathrm{span}\{\ket{\phi_h}: \ket{\phi} \mathrm{arbitrary} \}, \text{where} \nonumber
\\& \ket{\phi_h} := \frac{1}{\sqrt{T+1}} \sum_{i=0}^T U_i \ldots U_0 \ket{0^m}\otimes \ket{\phi} \otimes \ket{i}_\mathrm{clock}.
\end{align}
where $U_0= \mathds{1}$.
Any state having zero support on $\mathcal{S}_0$ has an energy at least $\Omega(1/T^3)$ \cite{Aharonov2007}, implying that the gap above the zero energy subspace is $\Delta  = \Omega(1/T^3)$.

Now, let us add in the term $H_1 = H_\mathrm{output}$, with $\norm{H_\mathrm{output}} = \epsilon$.
We choose $\epsilon < \Delta/16$, unlike the regular clock construction where $\epsilon$ is usually taken to be $\Theta(1)$.
As long as $\epsilon < \Delta/2$, we can restrict our attention to the zero energy space of $H_0$, since $H_1$ can change eigenvalues by at most $\epsilon$.
We use the tool of Schrieffer-Wolff transformation as described in \cref{sec_schwolff} to obtain a description of the Hamiltonian in the low-energy subspace.
The subspace $\mathcal{S}_0$ is the ground-state space of states with energy $0$.
Since $\norm{H_1} = \epsilon$, the associated low-energy subspace of $H = H_0 + H_1$ is 
\begin{align}
 \mathcal{S} = \mathrm{span} \{ \ket{\Phi} : \bra{\Phi}H\ket{\Phi} \in [-\epsilon,\epsilon] \},
\end{align}
the subspace with energies in $[-\epsilon,\epsilon]$.
In our case $H_0 P_0 = 0$ in the ground subspace spanned by history states $\ket{\phi_h}$, and the matrix elements of $P_0 H_1 P_0$ are given by
\begin{align}
& \bra{\phi_h} P_0 H_1 P_0 \ket{\psi_h} = \bra{\phi_h} H_1 \ket{\psi_h} \\
&= \frac{1}{T+1} \left(\sum_{i=0}^T \bra{0}^m \otimes \bra{\phi} \otimes \bra{i}_\mathrm{clock} U_0^\dag  \ldots U_{i}^\dag \right) H_1 \times \nonumber
\\& \left(\sum_{j=0}^{T} U_j \ldots U_0 \ket{0^m}\otimes \ket{\psi} \otimes \ket{j}_\mathrm{clock}\right)
\\ &= \frac{1}{T+1} \left(\sum_{i=0}^T \bra{0}^m \otimes \bra{\phi} \otimes \bra{i}_\mathrm{clock} U_0^\dag  \ldots U_{i}^\dag \right) \times\nonumber
\\& \epsilon \ketbra{0}_o \otimes \ketbra{T}_\mathrm{clock} \left(\sum_{j=0}^{T} U_j \ldots U_0 \ket{0^m}\otimes \ket{\psi} \otimes \ket{j}_\mathrm{clock}\right)
\\ &= \frac{1}{T+1} \bra{0}^m \otimes \bra{\phi} \otimes \bra{T} U^\dag \epsilon \ketbra{0}_o \otimes \nonumber
\\&\ketbra{T}_\mathrm{clock} U \ket{0^m}\otimes \ket{\psi} \ket{T}_\mathrm{clock} \\
&= \frac{\epsilon}{T+1} \bra{0}^m \otimes \bra{\phi} U^\dag \ketbra{0}_o U \ket{0^m}\otimes \ket{\psi} \\
&= \frac{\epsilon}{T+1} \bra{0}^m \otimes \bra{\phi} U^\dag (\mathds{1} - \Pi_\mathrm{out}) U \ket{0^m}\otimes \ket{\psi},
\end{align}
where $\Pi_\mathrm{out}$ is the projector onto the accepting state $\ket{1}_o$.
Continuing, we have
\begin{align}
 \bra{\phi_h} P_0 H_1 P_0 \ket{\psi_h} = \frac{\epsilon}{T+1} (\braket{\phi}{\psi} - \bra{\phi}Q\ket{\psi}),
\end{align}
meaning that the first order correction $P_0 H_1 P_0 $ is simply related to the accept operator $Q$, which was defined as $Q(U) = \bra{0}^{\otimes m} U^\dag\Pi_\mathrm{out}U \ket{0}^{\otimes m}$.
Let the eigenstates of $Q$ be $\ket{\phi_1}, \ket{\phi_2}, \ldots \ket{\phi_{2^w}}$ with eigenvalues $\lambda_1 \geq \lambda_2 \geq \ldots \lambda_{2^w}$. 
We use the associated history states $\ket{\phi_{i_h}}$ as a basis for the subspace $\mathcal{S}_0$.
In this basis, the first order correction $P_0 H_1 P_0$ is diagonal:
\begin{align}
 P_0 H_1 P_0 = \frac{\epsilon}{T+1} \sum_i (1-\lambda_i) \ketbra{\phi_{i_h}}.
\end{align}

We conclude that in the ground space of the original Hamiltonian $H_0$, the full Hamiltonian $H$ has eigenvalues $\epsilon (1-\lambda_i)/(T+1) \pm O(\epsilon^2/\Delta)$, where the quantity $\lambda_i$ is the accept probability of the verifier's circuit given $\ket{\phi_i}$ as witness.
This is the same conclusion we would obtain by applying degenerate perturbation theory, except that the error bound is rigorous.
We now analyze the YES and NO cases to obtain a lower bound on the promise gap.
In each case, we also lower bound the spectral gaps in the resulting Hamiltonian.

In the YES case the ground-state energy is $E_1 \leq \epsilon (1-c)/(T+1)$, as can be seen from the fact that the history state $\ket{\phi_h}$ corresponding to an accepting witness $\ket{\phi}$ would have energy ${\epsilon}(1-\bra{\phi}Q\ket{\phi})/{(T+1)} \leq \epsilon (1-c)/(T+1)$.
Our small-penalty clock construction and the Schrieffer-Wolff transformation comes in handy for the NO case.
We see in the NO case that the ground-state energy is at least $E_1 \geq \epsilon (1-s)/(T+1) - O(\epsilon^2/\Delta)$.
Therefore, the promise gap is at least $\epsilon (c-s)/(T+1) - O(\epsilon^2/\Delta) = \Omega(1/\exp)$ as long as $\epsilon/\Delta = o((c-s)/(T+1))$.

In the above, if we had chosen $\epsilon = \Theta(1)$ instead of $\epsilon < \Delta/16$, the NO case would have given us a bound $E_1 \geq \Omega(1-s)/T^3$.
This would mean that one would have to amplify the completeness and soundness $c,s$ to near unity in order to get a nontrivial promise gap.
However, such an amplification can, in general, shrink the spectral gap of the accept operator.
Independently, a large penalty term $\epsilon = \Theta(1)$ could also reorder some eigenvalues, meaning that the spectral properties of the resulting clock Hamiltonian would not faithfully track those of the original accept operator.

The spectral gap in the YES/NO case is $E_2-E_1 \geq \frac{\epsilon}{T+1}(\lambda_1(Q)-\lambda_2(Q)) - O(\frac{\epsilon^2}{\Delta})$.
We take $\epsilon = o(\Delta (c-s)/(T+1)) = o((c-s)/T^4)$, which is exponentially small if $c-s$ is.
As long as $\epsilon/\Delta = o(\min[g_1,g_2]/(T+1))$, both the YES and NO cases will have an exponentially small spectral gap.
In summary the choice
\begin{align}
 \epsilon = \frac{\min{[g_1,g_2,(c-s)]}}{nT^4} = \Theta(1/\exp)
\end{align}
suffices to have a promise gap and spectral gaps bounded below by $\Omega(1/\exp)$.
This proves \pEGQMA-hardness of \textsc{$(1/\exp,1/\exp)$-LocalHamiltonian} and one half of \cref{thm_pegqma_complete}.
\end{proof}

We generalize the above proof technique to the case of $\GQCMA$-hardness of \textsc{GS-Description-LocalHamiltonian}.
In addition to showing a promise gap and a spectral gap, we should show that the resulting Hamiltonian has a classical description of a circuit to prepare a low-energy state.
We show the following general lemma.

\begin{lemma} \label{lem_pqcma_hard} \raggedright
\textsc{$(\delta,\Delta)$-GS-Description-LocalHamiltonian} is $\GQCMA[c,s,g_1,g_2]$-hard for any $\delta,\Delta$ satisfying both the following conditions. 
\begin{itemize}
 \item[i.] $\delta = O((c-s)^2/\poly(n))$ for some polynomial. \\
 \item[ii.] If $c-s = o(\min[g_1, g_2])$, then any $\Delta$ satisfying $\Delta = O((c-s)\min[g_1,g_2]/\poly(n))$.
 Else, $\Delta = 0$.
\end{itemize}
\end{lemma}
\begin{proof}
To prove $\GQCMA$-hardness, we give a reduction from $\GQCMA[c,s,g_1,g_2]$ to \textsc{GS-Description-LocalHamiltonian[$a,b,g_1',g_2'$]}.
We are promised that the input witnesses are computational basis states (this can be assumed without loss of generality), corresponding to the classical witness sent by the prover.
We would like to show that there exists a circuit $V$ to prepare a state $\delta$-close in energy to the ground state of the clock Hamiltonian both the YES and NO cases.

Consider again the small-penalty clock construction, with the clock Hamiltonian \cref{eq_clockham}.
Let the norm of the penalty term be $\norm{H_\mathrm{output}} = \epsilon$.
When $\epsilon=0$, the ground-state space is given by valid history state computations corresponding to computational basis witness states.
The spectral gap above this subspace is at least $\Omega(1/T^3)$.
The addition of the penalty term changes the energies to $ \frac{\epsilon}{T+1}(1 - \lambda_k) + O(\epsilon^2 T^3)$, where $\lambda_k$ is the accept probability upon input computational basis state $\ket{y_k}$ as witness.
Consider the history state associated with witness $\ket{y_k}$:
\begin{align}
\ket{y_{k_h}} := \frac{1}{\sqrt{T+1}} \sum_{i=0}^T U_i \ldots U_0 \ket{0^m}\otimes \ket{y_k} \otimes \ket{i}_\mathrm{clock}.
\end{align}
This state has energy $\bra{y_{k_h}} H \ket{y_{k_h}} = \frac{\epsilon}{T+1}(1 - \lambda_k)$ and is therefore $O(\epsilon^2T^3)$-close in energy to the true ground state.
Therefore, as long as $\epsilon^2 T^3 < O\left(\frac{(b-a)^3}{f(n)^2}\right)$, a classical description of a circuit that prepares $\ket{y_{k_h}}$ is a valid ground-state description.
The circuit may be described by specifying $y_k$ and a circuit that prepares the history state $\ket{\phi_h}$ upon any quantum input $\ket{\phi}$.
This latter circuit first prepares the state $\frac{1}{\sqrt{T+1}} \sum_{i=0}^T \ket{0^m}\ket{i}_\mathrm{clock}$ and then applies the unitaries $U_j\ldots U_0$ controlled on the clock being in time-step $j$ \cite{Wocjan2003}.

The same promise gap and spectral gap analyses as in the proof of \cref{lem_ham_preciseegqma} hold.
In the YES case, the Hamiltonian has ground-state energy $\leq \frac{\epsilon}{T+1}(1 - \lambda_1) \leq \frac{\epsilon}{T+1}(1 - c)$.
In the NO case, the ground-state energy is at least $\frac{\epsilon}{T+1}(1 - \lambda_1) - O(\epsilon^2 T^3) \geq \frac{\epsilon}{T+1}(1 - s) - O(\epsilon^2 T^3)$.
The promise gap between the ground-state energy for YES and NO cases is $\delta \geq \frac{\epsilon}{T+1}(c-s) - O(\epsilon^2T^3)$.
We make the choice $\epsilon = \Theta(\frac{c-s}{T^4})$ to ensure the promise gap is $\Omega((c-s)^2/T^5)$.
This choice is consistent with the choice $\epsilon^2 T^3 \leq O(\frac{(b-a)^3}{f(n)^2})$ made above.

Let us now analyze the spectral gap of the resulting Hamiltonian.
Using the Schrieffer-Wolff expansion to obtain the eigenvalues of the Hamiltonian for small $\epsilon$, we have $\Delta \geq \frac{\epsilon}{T+1}(\lambda_1 - \lambda_2) - O(\epsilon^2 T^3)$.
The spectral gap is at least $\frac{\epsilon}{T+1} \min[g_1, g_2]$ as long as $\epsilon^2 T^3 = o(\frac{\epsilon}{T+1} \min[g_1, g_2])$.
Using the choice of $\epsilon$ above, this means the spectral gap is $\Omega((c-s)/T^5 \min[g_1,g_2])$ as long as $c-s = o(\min[g_1, g_2])$.
Otherwise, the best bound on the spectral gap is $\Delta \geq 0$.
Observing that $T=\poly(n)$ by assumption, we obtain the lemma.
\end{proof}

The lemma allows us to show the following: 
\begin{corll}[Second half of \cref{thm_qcma_complete,thm_pqcma_complete,thm_pgqcma_complete,thm_pegqcma_complete}] \label{corll_gsdescription_gqcma}
\textsc{$(1/\poly,0)$-GS-Description-LocalHamiltonian} is $\QCMA$-hard. \\
\textsc{$(1/\exp,0)$-GS-Description-LocalHamiltonian} is $\pQCMA$-hard. \\
\textsc{$(1/\poly,1/\poly)$-GS-Description-LocalHamiltonian} is $\PGQCMA$-hard. \\
\textsc{$(1/\exp,1/\exp)$-GS-Description-LocalHamiltonian} is $\pEGQCMA$-hard.
\end{corll}

For the problem with $\delta=1/\exp, \Delta=1/\poly$, we do not give a direct reduction from a $\pPGQCMA$ instance.
Instead, we show $\PP$-hardness through the characterization of $\PP$ in terms of the class $\pBQP$.
From the $\PP$ upper bound to $\pPGQCMA$, we obtain \pPGQCMA-completeness of the problem \textsc{$(1/\exp,1/\poly)$-GS-Description-LocalHamiltonian}.
The argument is similar for $\pPGQMA$-hardness of \textsc{$(1/\exp,1/\poly)$-LocalHamiltonian}.

\begin{lemma}[\cref{thm_gappedham_pphard,thm_pqcma_gap_pphard} restated] \label{thm_gappedham_pphard_apx} {\color{white}.}\\
\textsc{$(1/\exp,1/\poly)$-GS-Description-LocalHamiltonian} is $\PP$-hard. \\
\textsc{$(1/\exp,1/\poly)$-LocalHamiltonian} is $\PP$-hard.
\end{lemma}
\begin{proof}
We give a reduction from any problem in $\pBQP$ to \textsc{$(1/\exp,1/\poly)$-GS-Description-LocalHamiltonian}, which is also an instance of \textsc{$(1/\exp,1/\poly)$-LocalHamiltonian}.
Since $\pBQP$ is the class of problems that can be decided by quantum circuits with a promise gap $c-s = \Omega(1/\exp)$, it can also be thought of as ``$\pQMA$ without an input witness''.
The Hamiltonian is constructed out of the $\pBQP$ computation as $H = H_\mathrm{input} + H_\mathrm{prop} + H_\mathrm{output} + H_\mathrm{clock}$, where the terms are now
\begin{align}
& H_\mathrm{input} = \sum_{i=1}^m \ketbra{0}_i \otimes H_\mathrm{clockinit}, \\
& H_\mathrm{prop} = \sum_{i=0}^T - U_{i+1} \otimes \ketbra{i+1}{i}_\mathrm{clock} -U_{i+1}^\dag \otimes \ketbra{i}{i+1}_\mathrm{clock} \nonumber
\\& + \mathds{1} \otimes (\ketbra{i}_\mathrm{clock} + \ketbra{i+1}_\mathrm{clock}), \mathrm{\ and} \\
& H_\mathrm{output} = \epsilon \ketbra{0}_o \otimes \ketbra{T}_\mathrm{clock}.
\end{align}
The only difference from \cref{eq_clockhaminput,eq_clockhamprop,eq_clockhamoutput} is that $H_\mathrm{input}$ does not have support on an unpenalized proof register, since $\pBQP$ does not rely on a proof state given as input.
This is analogous to the clock construction of Ref.~\cite{Aharonov2007}, which was instrumental in the proof that adiabatic quantum computation is universal for $\BQP$.

We again let the Hamiltonian $H_0$ be $H_\mathrm{input} + H_\mathrm{prop} + H_\mathrm{clock}$  and $H_1 = H_\mathrm{output}$.
The ground state of $H_0$ is now nondegenerate (unique) and given by the history state
\begin{align}
\ket{0_h} := \frac{1}{\sqrt{T+1}} \sum_{i=0}^T U_i \ldots U_0 \ket{0^m} \otimes \ket{i}_\mathrm{clock}.
\end{align}
Let us denote the ground-state space of $H_0$ and the projector onto it by $\Pi_0$.
As for $H_1$, the ground space $\Pi_1$ is spanned by states belonging to subspaces $\mathcal{L}$ and $\mathcal{L}'$, with
\begin{align}
\mathcal{L} &= \ket{1}_o \otimes \ket{T}_\mathrm{clock} \\
\mathcal{L}' &= \mathrm{span}\ \{\ket{\psi}\} \otimes \mathrm{span}\ \{\ket{0}_\mathrm{clock}, \ket{1}_\mathrm{clock}, \ldots \ket{T-1}_\mathrm{clock}\},
\end{align}
with $\ket{\psi}$ arbitrary.

We observe that when $\epsilon = 0$, the Hamiltonian exactly corresponds to Aharonov et al.'s $H_\mathrm{final}$ \cite{Aharonov2007}.
Aharonov et al.\ \cite{Aharonov2007} showed that this Hamiltonian $H_0$ has a spectral gap of $\Delta = \Omega(1/T^3)$ in the full Hilbert space.
Further, the ground state of $H_0$ corresponds to the history state of the $\BQP$ computation ($\pBQP$ in this case), which starts off in a fixed, known state $\ket{0^m}$.

In the YES case, the ground-state energy of $H = H_0 + H_1$ can be bounded above by $\frac{\epsilon}{T+1}(1-c)$.
For the NO case, we again use the expression for the perturbed energies in the ground-state space coming from the Schrieffer-Wolff transformation.
Specifically, in the NO case, we have $E_1 \geq \epsilon(1-s)/(T+1) - O(\epsilon^2/\Delta)$, where $\Delta$ is the spectral gap above the ground state, just as in the proof of \cref{lem_ham_preciseegqma}.
The promise gap is lower-bounded by
\begin{align}
 \epsilon \frac{1-s}{T+1} - \epsilon \frac{1-c}{T+1} - O\left(\frac{\epsilon^2}{\Delta}\right).
\end{align}
Therefore, as long as $\epsilon/\Delta = o((c-s)/(T+1))$ and $\epsilon = \Omega(2^{-\poly})$, the promise gap is at least $\Omega(\epsilon (c-s)/(T+1)) = \Omega(2^{-\poly})$.
The spectral gap for the unperturbed Hamiltonian $H_0$, which is the same as the final Hamiltonian in Ref.~\cite{Aharonov2008}, is at least $\Omega(1/T^3)$.
Therefore, we pick $\epsilon = (c-s)/(nT^4)$, which ensures that the conditions above are satisfied.

Coming to the spectral gap of the full Hamiltonian, we observe that since the original Hamiltonian had a spectral gap of $\Omega(1/T^3)$ and the perturbation $H_1$ is exponentially small, the eigenvalues can change at most by $\norm{H_1} = \epsilon$, preserving the spectral gap.
So far, we have a reduction from any $\pBQP$ instance to an instance of \textsc{$(1/\exp,1/\poly)$-LocalHamiltonian}.

It remains for us to see that there is an efficient circuit that can prepare a state close in energy to the ground state.
By the justification in the proof of \cref{lem_pqcma_hard}, we know that choosing the output penalty term to be exponentially small causes the history state of the computation $\ket{0_h}$ to be exponentially close to the ground state in energy.
We have also seen the existence of a polynomial size circuit that prepares the history state given a description of the input (which here is $\ket{0^m}$ for \pBQP).
Note that when $\epsilon=0$, the ground state is unique and has a $\Omega(1/\poly)$ spectral gap above and therefore taking $\epsilon$ exponentially small does not pose a problem with spectral gaps.
\end{proof}
The difference between the proof of \cref{thm_gappedham_pphard_apx} and the proof of \cref{lem_ham_preciseegqma} is that it is the perturbation $\epsilon$ that creates the spectral gap in the proof of \cref{thm_gappedham_pphard_apx}, while in the proof of \cref{lem_ham_preciseegqma}, the spectral gap already exists in the unperturbed Hamiltonian.
This is why we can afford to take $\epsilon$ exponentially small here, which is needed to obtain an instance with a promise gap.

Thus, we have seen $\PP$-hardness of \textsc{$(1/\exp,1/\poly)$-LocalHamiltonian}.
$\pPGQMA$-hardness of the problem follows from the fact that $\pPGQMA \subseteq \PP$ (\cref{thm_ppgqma_inpp}).
\begin{corll} \label{corll_ppgqma_cook}
\textsc{$(1/\exp,1/\poly)$-LocalHamiltonian} is $\pPGQMA$-hard.
\end{corll}

Similarly, the $\PP$-hardness of \textsc{$(1/\exp,1/\poly)$-GS-Description-LocalHamiltonian} and the result $\pPGQCMA \subseteq \pPGQMA = \PP$ together imply the following result.
\begin{corll} \label{corll_ppgqcma_cook}
\textsc{$(1/\exp,1/\poly)$-GS-Description-LocalHamiltonian} is $\pPGQCMA$-hard.
\end{corll}

Lastly, the remaining case is \textsc{$(1/\poly,1/\exp)$-GS-Description-LocalHamiltonian} with $\delta=1/\poly,\ \Delta=1/\exp$, for which we argue that an instance with spectral gap $\Delta=\Omega(1/\poly)$ is also an instance with $\Delta=\Omega(1/\exp)$.
Therefore, \textsc{$(1/\poly,1/\exp)$-GS-Description-LocalHamiltonian} is $\PGQCMA$-hard, and, since $\PGQCMA=_R \EGQCMA$, $\EGQCMA$-hard under randomized reductions.
For the case of \textsc{$(1/\poly,1/\exp)$-LocalHamiltonian}, we do not currently have a hardness result.
This is because, in performing a reduction from $\EGQMA$, we get an instance of \textsc{$(1/\poly,0)$-LocalHamiltonian} and do not get any promise on the spectral gap that results.

\section{Precise phase estimation of gapped Hamiltonians} \label{sec_ingqma}
In this section, we show that the \textsc{$(1/\exp,\Delta)$-LocalHamiltonian} problems with either $1/\poly$ or $1/\exp$ spectral gaps defined in \cref{sec_completeprobs}  are in the corresponding $\pGQMA$ class.
Together with the results of the previous section, this proves \cref{thm_ppgqma_complete,thm_pegqma_complete}.
\begin{lemma}
\label{lem_gappedham_in_precisepgqma}
\textsc{$(1/\exp,1/\poly)$-LocalHamiltonian} $\in \pPGQMA$.
\end{lemma}
\begin{lemma}
\label{lem_gappedham_in_preciseegqma}
\textsc{$(1/\exp,1/\exp)$-LocalHamiltonian} $\in \pEGQMA$.
\end{lemma}

The proof relies on phase estimation to infer energies of a local Hamiltonian.
The standard phase estimation circuit requires $\exp(n)$ many gates in order to infer the eigenvalues to $1/\exp$ precision.
However, since we want to show containment in a $\mathsf{Precise}$- class, we can use the power of being able to distinguish between two cases with exponentially close accept probabilities.
It turns out that phase estimation with a single ancillary qubit is enough to distinguish between the YES and NO cases, as shown in Ref.~\cite{Fefferman2018}.
Moreover, we show that the circuit preserves spectral gaps of the Hamiltonian: if two eigenstates have energies separated by some amount, then the phase estimation circuit also has a gap in the accept probabilities corresponding to these input states.

\begin{figure}
\mbox
{\Qcircuit @C=1em @R=.7em {
& \qw&\multigate{2}{e^{-iHt}} & \qw &\qw \\
\lstick{\ket{\psi}} & \qw&\ghost{e^{-iHt}} & \qw &\qw & \rstick{\ket{\psi}} \\
& \qw&\ghost{e^{-iHt}} & \qw & \qw \\
\lstick{\ket{0}}& \gate{\mathscr{H}} & \ctrl{-1} & \gate{\mathscr{H}} & \meter 
}
}
\caption{\raggedright One-qubit phase-estimation circuit.
The symbol $\mathscr{H}$ denotes the Hadamard gate.}
\end{figure}
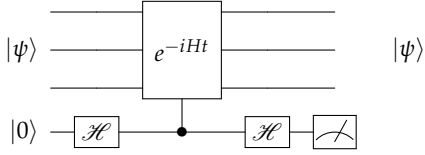

Below we give a unified proof of \cref{lem_sparseHam_in_pEGQMA,lem_gappedham_in_precisepgqma,lem_gappedham_in_preciseegqma}.
Specifically, we show $\pPGQMA$ ($\pEGQMA$) containment of the problem \textsc{$(1/\exp,\Delta)$-GappedSparseHamiltonian} with $\Delta= 1/\poly$ ($\Delta = 1/\exp$).
\begin{lemma} \label{lem_gappedham_in_gappedqma}
\textsc{GappedSparseHamiltonian[$a,b,g_1,g_2$]} has a $\GQMA[c,s,g_1',g_2']$ protocol with spectral gaps $g_1' = \Omega(g_1^2/\poly)$ and $g_2' = \Omega(g_2^2/\poly)$ and promise gap $c-s = (b-a)^2/\poly$.
\end{lemma}
\begin{proof}
The strategy is to ask the prover for the ground state of the sparse Hamiltonian.
The verifier then performs phase estimation on the witness state with a single ancillary qubit and uses the power to decide between two cases with exponentially close accept probabilities.
This power effectively enables computation of the phase of $e^{-iHt}$ to exponential precision, despite having a single ancilla qubit in the phase estimation circuit (see Ref.~\cite{Fefferman2018} for more details).
If $t \leq \frac{\pi}{2\norm{H}}$, all eigenstates of $H$ would correspond to a unique phase and a unique accept probability for the circuit.
We know an upper bound $dk$ on $\norm{H}$ through the Gershgorin circle theorem because we are assured that the magnitude of the entries is $\leq k$ and the sparsity is $d$.
Therefore, it suffices to choose $t \leq \frac{\pi}{2dk}$.

In order to perform phase estimation to exponentially small error, we need to apply a controlled-$e^{-iHt}$ rotation to error $\epsilon = 1/\exp$.
This is possible due to Hamiltonian simulation algorithms for sparse Hamiltonians, whose circuit size scales as $\poly(n)\log(\frac{1}{\epsilon})$ \cite{Berry2014}, which is polynomial in $n$, as desired.
The accept probability of the circuit upon input an eigenstate $\ket{E_i}$ of the Hamiltonian is $\frac{1 + \cos (E_it)}{2} $.
The promise gap can be lower bounded by an inverse exponential, as has been analyzed previously \cite{Fefferman2018}.

We can also show a spectral gap in the accept operator, or equivalently, a gap in the accept probabilities of the circuit for the optimal state and any state orthogonal to it.
Since the phase estimation circuit does not apply the exact controlled-$e^{-iHt}$ unitary but a unitary $U_x$ exponentially close to it, the eigenstates of $Q = \bra{0} \Pi_\mathrm{in} U_x^\dag \Pi_\mathrm{out} U_x \Pi_\mathrm{in} \ket{0}$ are not exactly the eigenstates of $e^{-iHt}$ (or of $H$).
However, since $\norm{e^{-iHt}-U_x} \leq \epsilon$, the eigenvalues of $Q$ are exponentially close to the accept probabilities of the eigenstates $\ket{E_i}$ of $H$.
The difference in accept probabilities can be bounded by $\epsilon$.

The difference in the ideal accept probabilities of the ground state and the first excited state is $\frac{\cos (E_0t) - \cos (E_1t)}{2}$.
Applying Taylor's theorem to $\cos(E_1t)$ around the point $E_0t$, we get
\begin{align}
\cos(E_1t) &= \cos(E_0t) -\sin(E_0t)t(E_1-E_0) -\nonumber
\\&\cos(E_0t) \frac{t^2(E_1-E_0)^2}{2} + \sin(E_0t) \frac{h^3}{6}
\end{align}
for some $h \in [0,(E_1-E_0)t] $.
Therefore,
\begin{align}
\cos(E_0t) &- \cos(E_1t) = \sin(E_0t)t(E_1-E_0) +\nonumber
\\& \cos(E_0t) \frac{t^2(E_1-E_0)^2}{2} - \sin(E_0t) \frac{h^3}{6} \\
& \geq t^2 (E_1-E_0)^2/2 - \frac{t^3(E_1-E_0)^3}{6}. \\
& \geq \Omega(t^2 (E_1-E_0)^2), \label{eq_phase_estimation_gap}
\end{align}
where in the second line we use the fact that $E_0t, E_1t < \pi/2$ and $(E_1-E_0)^3t^3 = O(t (E_1-E_0))$ in the third.
Therefore, the ideal accept probabilities also have a gap of $\Omega((E_1-E_0)^2/\norm{H}^2) = \Omega(\Delta^2/\poly)$ as long as $\epsilon \leq O(t^2 \Delta^2/n) = O(\Delta^2/\poly)$.
Now, when the applied unitary differs from the ideal one by $\epsilon$ in operator norm distance, the gap in the accept probabilities differs from the ideal accept probabilities by $2\epsilon$.
We therefore choose $\epsilon$ sufficiently small, i.e. we choose, say, $\epsilon = \Theta(t^2 (E_1-E_0)^2/2^n)$, which is still $\Omega(2^{-\poly})$, as needed.

To see the existence of a promise gap, notice that $E_0 \leq a$ in the YES case and $E_1 \geq b$ in the NO case, giving $c -s = \Omega(t^2(b-a)^2 - 2\epsilon) = \Omega((b-a)^2/\poly)$.
This proves the lemma.
\end{proof}
As corollaries, we obtain \cref{lem_sparseHam_in_pEGQMA,lem_gappedham_in_precisepgqma,lem_gappedham_in_preciseegqma}, since a local Hamiltonian is also a sparse Hamiltonian.

\section{Phase estimation in the presence of efficient circuit descriptions} \label{sec_pqcmadetails}

In this section, we show the problem \textsc{$(\delta,\Delta)$-GS-Description-LocalHamiltonian} is in $\GQCMA$ with appropriate bounds on the promise and spectral gaps (\cref{thm_qcma_complete,thm_pqcma_complete,thm_pgqcma_complete,thm_pegqcma_complete,thm_ppgqcma_complete}).

We first deal with the case of zero spectral gap:
\begin{lemma}[One half of \cref{thm_qcma_complete,thm_pqcma_complete}] \label{lem_in_qcma_general}
 \textsc{$(1/\poly,0)$-GS-Description-LocalHamiltonian} $\in \QCMA$. \\
 \textsc{$(1/\exp,0)$-GS-Description-LocalHamiltonian} $\in \pQCMA$.
\end{lemma}
\begin{proof}
For the upper bound, we describe a $\QCMA$ or $\pQCMA$ protocol.
We are promised that in both the YES and NO cases, there exists a classical description of a circuit $V$ of polynomial size that will create a state with energy close to the ground-state energy.
Specifically, the energy of this state is $\epsilon$-close to the ground-state energy, for $\epsilon < \frac{(b-a)^3}{f(n)^2}$ for a polynomial $f(n) \geq \norm{H}$.
For $\QCMA$, we have $b-a \geq \Omega(1/\poly)$, while for $\pQCMA$, $b-a\geq \Omega(1/\exp)$.
The verifier asks the prover to give this description (which is promised to exist).
The verifier then creates a state $\ket{\psi}$ with low energy by applying $V$ to $\ket{0^m}$.
The verifier measures the energy of this state via the one-bit phase-estimation protocol outlined in \cref{sec_ingqma}, which involves applying a controlled-$e^{-iHt}$ for time $t \leq \frac{\pi}{2\norm{H}}$.

The proof that this verification protocol works is slightly more involved than the $\QMA[c,s]$ case.
This is because, in the case of $\QMA$ a verifier can assume without loss of generality that the prover sends the optimal eigenstate as a witness.
However, in the case of \textsc{GS-Description-LocalHamiltonian}, we are only promised the existence of an efficient circuit to prepare a state close in energy to the ground state, and not the ground state itself\footnote{The weaker promise is more natural since it is more robust.}.
Despite this complication, we can still show that a state close in energy to the ground state behaves similarly with respect to the accept probabilities of the $\QCMA[c,s]$ verifier.

In the YES case, there is a description $V$ that produces a state $\ket{\psi}$ with energy close to the ground-state energy (i.e.\ with energy $\leq E_1 + \epsilon < a +  \frac{(b-a)^3}{\poly(n)}$).
We show in \cref{lem_phaseestacceptprob} that the accept probability of the verifier upon performing one-bit phase estimation on the state $\ket{\psi}$ is at least $\cos^2\left( \frac{bt}{2} \right) + \Omega((b-a)^2/\poly)$.
In the NO case, the optimal strategy for the prover is to send the description of a circuit that makes a state as close as possible to the ground state, since the accept probabilities are monotonic in energy and there exists no other state with smaller energy, by definition.
Even if the prover sends the verifier a circuit that exactly prepares the ground state $\ket{E_1}$, its energy in the NO case is already $\geq b$.
This means that the verifier will accept with probability at most ${(1+\cos E_1t)}/{2} \leq {(1+\cos bt)}/{2}$.
Therefore there is a separation in the accept probabilities in the YES and NO cases of $c-s = \Omega((b-a)^2/\poly)$, which is $\Omega(1/\poly)$ for $b-a = \Omega(1/\poly)$ and $\Omega(1/\exp)$ for $b-a = \Omega(1/\exp)$.
\end{proof}

\begin{lemma}
If a state $\ket{\psi}$ has energy $\expectationvalue{H}{\psi} = \expval{E} \leq E_1 + \frac{5(b-a)^3}{24f(n)^2}$ for some polynomial $f(n) \geq \norm{H}$, then in the YES case, the accept probability of the state upon phase estimation with one bit of precision is $\expval{p} \geq \cos^2\left( \frac{bt}{2} \right) + \delta$, where $\delta = \Omega\left(\frac{5(b-a)^2}{24 f(n)^2}\right)$. \label{lem_phaseestacceptprob}
\end{lemma}
\begin{proof}
We are given a state $\ket{\psi}$ with energy $\expval{E}$.
Let $p_j = \abs{\braket{E_j}{\psi}}^2$ be the weight of the energy eigenstate $E_j$.
Then we know $p_1 E_1 + p_2E_2 + \ldots p_{2^n} E_{2^n} = \expval{E}$.
The probability of accepting $\ket{\psi}$ in the one-bit phase estimation circuit is given by
$\expval{p} = p_1 \cos^2\left( \frac{E_1t}{2}\right) + p_2 \cos^2\left( \frac{E_2t}{2}\right) + \ldots p_{2^n} \cos^2\left( \frac{\norm{H}t}{2}\right)$, where $\norm{H} = E_{2^n}$.
Given the constraint on the energy $\expval{E}$, we show in \cref{lem_lpbound} that $\expval{p} \geq \cos^2\left( \frac{E_1t}{2}\right) \left( 1-x \right) + \cos^2\left( \frac{\norm{H}t}{2}\right) x$, where $x:=\frac{\expval{E}-E_1}{\norm{H}-E_1}$.
Now in order to have $\expval{p} \geq \cos^2\left( \frac{bt}{2}\right) + \delta$, it suffices to have
\begin{align}
 x &\leq \frac{\cos^2\left( \frac{E_1t}{2}\right) -\cos^2\left(\frac{bt}{2}\right) - \delta}{\cos^2\left( \frac{E_1t}{2}\right)- \cos^2\left(\frac{\norm{H}t}{2}\right)}
 \\ &= \frac{\cos\left({E_1t}\right) -\cos\left({bt}\right) - 2\delta}{\cos\left( {E_1t}\right)- \cos\left({\norm{H}t}\right)}.
\end{align}
It is therefore sufficient if
\begin{align}
x \leq & \frac{(b-a)t}{2}\left( bt - \frac{b^3t^3}{6} \right)- \delta, \label{eq_conditiononx} \mathrm{since}
\\&\frac{(b-a)t}{2}\left( bt - \frac{b^3t^3}{6} \right) - \delta \leq \frac{(b-a)t \sin (bt) - 2\delta}{2}
\\ &\leq \frac{(b-a){t}{}\sin(bt) - 2\delta}{\cos (at) - \cos (\norm{H}t)}
\\&\leq \frac{\cos(at) - \cos(bt) -2\delta}{\cos (at) - \cos (\norm{H}t)} \nonumber
\\&\leq \frac{\cos(at) - \cos(bt) -2\delta}{\cos (E_1t) - \cos (\norm{H}t)},
\end{align}
where we use the inequalities $\sin (bt) \geq bt - \frac{b^3t^3}{6}$, $E_1 \leq a$, $\cos(at) - \cos(bt) \geq (b-a)t \sin (bt)$, and $2 \leq \cos (at) - \cos (\norm{H}t)$.
We now require $\delta \geq (b-a)t \sin(bt)/4 $
, so that the condition \cref{eq_conditiononx} translates to $x \leq \frac{(b-a)t}{4}\left( bt - \frac{b^3t^3}{6} \right)$.

Let us choose $t = \min[1/f(n), 1/b] = 1/f(n)$, since otherwise $b \geq f(n) \geq \norm{H}$ and the instance is trivial.
We thus know $\norm{H}t \leq 1 \leq \pi/2$, $t \geq 1/f(n)$, and $bt< 1$.
We also assume that in the YES case, $\norm{H} - E_1 \geq b-a$.
This is because otherwise a verifier can compute $\frac{\Tr(H)}{2^n}$ efficiently given the Hamiltonian and accept straightaway if $\frac{\Tr(H)}{2^n} \leq b$.
This works since $E_1 \leq \frac{\Tr(H)}{2^n}$, and by the promise, $E_1 \leq b \implies E_1 \leq a$.
Therefore, without loss of generality, one can assume that the nontrivial instances satisfy $b \leq \frac{\Tr(H)}{2^n} \leq \norm{H}$, or $\norm{H}-E_1 \geq b-a$.

Therefore, since $\expval{E} \leq E_1 + \frac{5(b-a)^3}{24f(n)^2}$, we have
\begin{align}
\expval{E}  &\leq E_1 + \frac{5(b-a)^2(\norm{H}-E_1)}{24f(n)^2}
\\ \implies x &\leq \frac{(b-a)^2}{4 f(n)^2} \frac{5}{6} = \frac{(b-a)^2}{4 f(n)^2} \left(1-\frac{1}{6}\right)
\\ &\leq \frac{(b-a)^2}{4 f(n)^2}\left( 1-\frac{b^2t^2}{6} \right)
\\ & \leq \frac{(b-a)}{4 f(n)} \frac{b}{f(n)}\left( 1-\frac{b^2t^2}{6} \right)
\\ &\leq \frac{(b-a)t}{4}\left( bt - \frac{b^3t^3}{6} \right),
\end{align}
as required.
To sum up, we have shown that $\expval{E} \leq E_1 + \frac{5(b-a)^3}{24f(n)^2}$ implies $\delta \geq (b-a)t \sin(bt)/4 \geq \frac{(b-a)^2(1-b^2t^2/6)}{4f(n)^2} \geq \frac{5(b-a)^2}{24f(n)^2}$.
\end{proof}

\begin{lemma}
For probabilities $p_j: j\in [2^n]$ satisfying $\sum_j p_j E_j \leq \expval{E}$ and numbers $E_1 \leq E_2 \leq \ldots E_{2^n}$ satisfying $E_{j}t \in [0,\pi/2]$, the quantity $\sum_j p_j \cos^2\left( \frac{E_jt}{2} \right)$ is bounded below by $\cos^2\left( \frac{E_1t}{2} \right)(1-x) + \cos^2\left( \frac{E_{2^n}t}{2} \right)x$, where $x$ is given by $\frac{\expval{E}-E_1}{E_{2^n}-E_1}$. \label{lem_lpbound}
\end{lemma}
\begin{proof}
Since the function $f(x) = -\cos^2(xt/2)$ is convex for $xt/2 \in [0,\pi/2)$, we have
\begin{align}
 \frac{f(E_1)(E_{2^n} - E_j) + f(E_{2^n})(E_j - E_1)}{E_{2^n}-E_1} &\geq f(E_j).
\end{align}
Therefore,
\begin{align}
 & p_j \frac{f(E_1)(E_{2^n} - E_j) + f(E_{2^n})(E_j - E_1)}{E_{2^n}-E_1} \geq p_jf(E_j)
 \\\implies & \frac{f(E_1)(E_{2^n} - \expval{E}) + f(E_{2^n})(\expval{E} - E_1)}{E_{2^n}-E_1} \geq  \sum_j p_j f(E_j)
\\ \implies & \sum_j p_j \cos^2\left( \frac{E_jt}{2} \right) \geq \cos^2\left( \frac{E_1t}{2} \right) \frac{E_{2^n}-\expval{E}}{E_{2^n}-E_1} + \nonumber \\
& \cos^2\left( \frac{E_{2^n}t}{2} \right) \frac{\expval{E}-E_1}{E_{2^n}-E_1},
\end{align}
which completes the proof.

\end{proof}

We now turn to the cases where in addition to the promise of an efficient circuit to prepare a low-energy state, the Hamiltonian is promised to have a spectral gap $\Delta$.
For this case, we can show the following:
\begin{lemma}
\raggedright
\textsc{GS-Description-LocalHamiltonian[$a,b,g_1,g_2$]} $\in \GQCMA[c,s,g_1',g_2']$ for $c-s = \Omega\left( \frac{(b-a)^2}{f(n)^2} \right)$ and $\min[g_1',g_2'] \geq \frac{5\Delta^2}{36 f(n)}$, where $f(n)$ is a polynomial upper bound to $\norm{H}$, and $\Delta= \min[g_1,g_2] \geq (b-a)^3/f(n)^2$.
\end{lemma}
\begin{proof}
We analyze the same algorithm as the non-gapped case and show that the verification protocol, with slight modifications, preserves the spectral gap.
In particular, in the first step of the original protocol, the verifier straightaway accepts if $\Tr(H)/2^n \leq b$ or if the upper bound to the norm of the Hamiltonian, $f(n)$ satisfies $f(n) \leq b$.
We modify this to requiring the verifier to accept only if, in addition to the previous conditions, measurement of the witness register yields the all zeroes string $0^w$ (where $w$ is the size of the witness register).
This has the effect of creating a spectral gap, since in this case only the all-zeroes state is accepted and all other computational-basis states are rejected.

If the first step does not cause the verifier to accept, the verifier assumes that the witness state is a description of the circuit $V$ to prepare a low-energy state $\ket{\psi}$.
The verifier then proceeds to prepare this state and measure its energy using the one-bit phase estimation protocol.
As shown in the proof of \cref{lem_in_qcma_general}, the protocol has a promise gap $c-s = \Omega\left(\frac{(b-a)^2}{f(n)^2} \right)$.

We now analyze the spectral gap.
Let us denote by $y$ the quantity $\frac{\expval{E}- E_1}{E_2 - E_1}$ and by $x$ the quantity $\frac{\expval{E}-E_1}{E_{2^n}-E_1} \leq y$.
Any state with energy $\expval{E}:= \expval{H}{\psi} \leq E_1 + \frac{(b-a)^3}{f(n)^2} \leq E_1 + {\Delta}$ has a large overlap with the ground state:
\begin{align}
\abs{\braket{\psi}{E_1}}^2 &\geq 1 - \frac{\expval{E}- E_1}{E_2 - E_1} = 1-y.
\end{align}
Therefore, any state $\ket{\phi}$ orthogonal to $\ket{\psi}$ must have an overlap with the ground state that satisfies $\abs{\braket{\phi}{E_1}}^2 \leq y$.
This means that the accept probability for any witness orthogonal to the one corresponding to the ground-state description is
\begin{align}
 \expval{p_{\phi}} &= \sum_j p_j \cos^2\left( \frac{E_jt}{2} \right)
 \\ &\leq y \cos^2\left( \frac{E_1t}{2} \right) + \left( 1- y \right) \cos^2\left( \frac{E_2t}{2}\right).
\end{align}
On the other hand, the accept probability of the optimal witness is at least (\cref{lem_phaseestacceptprob})
\begin{align}
 \expval{p_\psi} &\geq \left( 1-x \right) \cos^2\left( \frac{E_1t}{2} \right) +  x \cos^2\left( \frac{E_2t}{2}\right).
\end{align}
The difference in these two is a lower bound for the spectral gap of the accept operator:
\begin{align}
 g_1,g_2 \geq & \expval{p_\psi} - \expval{p_\phi} \geq \cos^2\left( \frac{E_1t}{2} \right) \left(1 - x- y \right) +\nonumber
\\& \cos^2\left( \frac{E_2t}{2} \right) (x+y-1)
 \\ &= \frac{(1-x-y)}{2} (\cos (E_1t) - \cos(E_2t))
 \\ &\geq \frac{(1-2y)}{2} (E_2-E_1) \sin(E_2t).
\end{align}
Now, we know from the promise that $y = \frac{\expval{E}- E_1}{E_2 - E_1} \leq \frac{(b-a)^3}{f(n)^2 \Delta} \leq \frac{1}{3}$, and $E_2 \geq E_1 + \Delta \geq \Delta$.
Also, we have chosen $t \geq 1/f(n)$ for a polynomial $f(n) \geq \norm{H}$.
Therefore,
\begin{align}
 \min[g_1',g_2'] &\geq \frac{\Delta}{6} \sin(\Delta t)
 \\ &\geq \frac{\Delta}{6} \left( \Delta t - \frac{\Delta^3t^3}{6} \right)
 \\ &\geq \frac{\Delta}{6} \left( \frac{\Delta}{f(n)} - \frac{\Delta^3}{6f(n)^3} \right)
 \\ &= \frac{\Delta^2}{6f(n)} \left( 1 - \frac{\Delta^2}{6f(n)^2} \right)
 \\ &\geq \frac{5\Delta^2}{36f(n)},
\end{align}
since $\Delta \leq \norm{H} \leq f(n)$.
\end{proof}
This proves the following results:
\begin{corll}[One half of \cref{thm_pgqcma_complete,thm_pegqcma_complete,thm_ppgqcma_complete}] 
{\color{white}{.}} \newline
\textsc{$(1/\poly,1/\poly)$-GS-Description-LocalHamiltonian} $\in \PGQCMA$. \\
\textsc{$(1/\exp,1/\exp)$-GS-Description-LocalHamiltonian}  $\in \pEGQCMA$. \\
\textsc{$(1/\exp,1/\poly)$-GS-Description-LocalHamiltonian} $\in \pPGQCMA$.
\end{corll}

\section{Details of $\PP$ algorithm}\label{sec_ppdetails}
In this section we complete the proof of \cref{thm_ppgqma_inpp} by expanding upon the $\PP$ algorithm.
We also prove \cref{lem_gslo_ppp} by giving a $\P^\PP$ algorithm to precisely compute ground-state local observables of $\Omega(1/\poly)$-spectral-gapped Hamiltonians.
\begin{lemma} \label{lem_cooling_othermatrices}
A $\PP$ algorithm can decide whether $\Tr [Q^q A] \leq a'$ or $\geq b'$ when input thresholds $a'$ and $b'$, for matrices $Q$ and $A$ of size $2^{\poly(n)} \times 2^{\poly(n)}$ satisfying the following properties (we use the symbol $R$ to denote both matrices $Q$ and $A$ in the following):
\begin{enumerate}
 \item The norm of the matrix $R$ is upper bounded by a polynomial in $n$.
 \item The matrix $R$ may be written as a polynomial of degree $d =\poly(n)$ in terms of matrices $R_i, i\in [m]$ in the computational basis for $m=\poly(n)$, such that:
 \begin{enumerate}
 \item The matrix elements of each matrix $R_i$ are computable to precision $\delta$ in time polynomial in $n$ and $\log (1/\delta)$.
 \end{enumerate}
\end{enumerate}
\end{lemma}
\begin{proof}
The quantity $\Tr(Q^q A)$ may be expressed as
\begin{align}
 \sum_x \expval{Q^q A}{x} =& \sum_x \sum_{x_1, x_2, \ldots x_q}\bra{x}Q \ketbra{x_1} Q \ket{x_2} \ldots \nonumber
 \\& \bra{x_{q-1}} Q \ketbra{x_{q}} A \ket{x}.
\end{align}

If $Q$ is a polynomial of degree $d$ in terms of matrices $R_1, \ldots R_m$ for $m=\poly(n)$, we can write it as 
\begin{align}
Q = \sum_{\substack{i_1,i_2,\ldots i_m \in [d] \\ i_1 + i_2 + \ldots i_m \leq d}} p_{i_1 i_2 \ldots i_m} R_1^{i_1} R_2^{i_2} \ldots R_m^{i_m}, \label{eq_pppaths}
\end{align}
where each tuple $(i_1,\ldots i_m)$ specifies a monomial.
The number of terms in the polynomial is bounded above by $(d+1)^m = \exp[m \ \log(d+1)] = O(\exp[\poly(n)])$.
We write a term of \cref{eq_pppaths}, $\bra{x_j} Q \ket{x_{j+1}}$, as
\begin{align}
\bra{x_j} Q \ket{x_{j+1}} &= \sum_{\substack{i_1,i_2,\ldots i_m \in [d] \\ i_1 + i_2 + \ldots i_m \leq d}}p_{i_1 i_2 \ldots i_m} \bra{x_j} R_1^{i_1}\ket{z_{j,1}}\times \nonumber
\\& \bra{z_{j,1}} R_2^{i_2}\ketbra{z_{j,2}} \ldots \bra{z_{j,m-1}} R_m^{i_m} \ket{x_{j+1}}. \label{eq_generalham_polynomial_paths}
\end{align}
We can further insert resolutions of the identity in \cref{eq_generalham_polynomial_paths} to get a sum over yet more terms.
Each term in the resulting sum is a product over polynomially many quantities of the form $\bra{w_1}R_s\ket{w_2}$ for some computational basis states $\ket{w_1},\ket{w_2}$ and an index $s \in [m]$.
Each of these can be computed in polynomial time.
The number of terms in the final sum of the form in \cref{eq_pppaths} is still bounded above by $2^\poly$.

From the assumption, the matrix elements of the matrices $R_i$ can be computed to additive error $2^{-g(n)}$ in time scaling as $O(g(n))$.
We therefore choose $g(n)$ to be such that the total additive error resulting from the $2^\poly$ many paths in \cref{eq_pppaths} is negligible compared to $(b'-a')\times 2^{\poly}$, where the second term ($2^\poly$) corresponds to the number of terms in the sum.
This can be ensured by taking $g(n)$ to be a sufficiently large polynomial.

\Cref{eq_pppaths} is a sum over $T = O(2^{\poly})$ many terms $f_i$, each of which may be computed in polynomial time.
Each term of \cref{eq_pppaths} may be interpreted as a path in a Turing machine.
Therefore, a $\PP$ machine can decide whether $\sum_{i=1}^T f_i$ is $\leq a'$ or $\geq b'$ for some thresholds $a', b' \geq a' + \Omega(2^{-\poly})$ input to the $\PP$ machine.
This is seen as follows.
Each term $f_i$ is an efficiently computable real-valued function of the trajectory $x_0^i, x_1^i,\ldots x_K^i$.
Let $a_\mathrm{max}$ be an upper bound to the norm of $A$.
The $\PP$ machine selects a uniformly random trajectory and computes $f_i$.
It accepts with probability $\frac{1}{2} - \frac{f_i}{2^{n+1}a_\mathrm{max}} > 0$ and rejects otherwise.
The overall acceptance probability is $\frac{1}{T}\sum_i (\frac{1}{2} - \frac{f_i}{2^{n+1}a_\mathrm{max}})$.
In the YES case, this is at least $\frac{1}{2} - \frac{a'}{2^{n+1}Ta_\mathrm{max}}$, while in the NO case, it is at most $\frac{1}{2} - \frac{b'}{2^{n+1}Ta_\mathrm{max}}$. 
Since we at least have a separation of $2^{-n-1}/T \times \Omega(b'-a') = \Omega(2^{-\poly(n)})$ between the YES and NO instances, this is a valid $\PP$ algorithm.
\end{proof}
\cref{lem_cooling_othermatrices} applies to the proof of \cref{thm_ppgqma_inpp} because the accept operator $Q$ in that proof is a degree $2T+3$-polynomial in matrices with efficiently computable entries. 

For the proof of \cref{lem_gslo_ppp}, we show in \cref{lem_thermal} that beginning from the maximally mixed initial state, imaginary time evolution for ``time'' $-i\beta$ produces a thermal state with high enough overlap with the ground state for a suitable $\beta$.
Computing local observables in the obtained thermal state then suffices to get exponentially good estimates of ground-state local observables for gapped systems.
We make the choice of a maximally mixed initial state in the above because it is guaranteed to have at least overlap $2^{-n}$ with the ground state.

\begin{lemma}\label{lem_thermal}
For a Hamiltonian $H$ with spectral gap at least $\Delta$, let $\rho_\beta$ be the thermal state at temperature $1/(2\beta)$.
Also let $\ket{E_1}$ be the ground state of $H$ and let $A$ be any local observable satisfying $\norm{A}\leq \poly(n)$.
Then for $\beta = \Omega(n\Delta^{-1})$, the thermal expectation value satisfies $\abs{\Tr [\rho_\beta A] - \bra{E_1} A \ket{E_1}} \leq 2^{-\poly}$.
\end{lemma}
\begin{proof}
Let the eigenstates of the Hamiltonian be given by $\ket{E_i}$, $i \in [2^n]$, with the eigenvalues $E_i$ arranged in nondecreasing order.
Consider the initial state $\rho = \mathds{1}/2^n$ and apply the linear operation $\exp(-\beta H)$, which performs imaginary time evolution for ``time'' $-i\beta$:
\begin{align}
 \rho \rightarrow \rho' = \exp(-\beta H) \rho \ \exp(-\beta H),
\end{align}
up to normalization.
The maximally mixed initial state $\rho = \frac{\mathds{1}}{2^n} = \sum_i \frac{1}{2^n} \ketbra{E_i}$ transforms to the state ${\rho_\beta}$, given by
\begin{align}
 \rho_\beta = \frac{\rho'}{\mathcal{N}} &= \frac{1}{\mathcal{N}} e^{-\beta H} \sum_i \frac{1}{2^n} \ketbra{E_i} e^{-\beta H} \\&
 = \frac{1}{2^n \mathcal{N}} \sum_i e^{-2\beta E_i} \ketbra{E_i}.
\end{align}
This state is the same as the thermal state $e^{-2\beta H}$ at temperature $1/(2\beta)$ up to normalization.
The normalization factor $\mathcal{N} = \Tr \rho'$ is given by $\sum_i {e^{-2\beta E_i}}/{2^n}$.
The overlap of the normalized state with the ground state is thus 
\begin{align}
\Tr[\rho_\beta \ketbra{E_1}] &= \frac{e^{-2\beta E_1}}{2^n \mathcal{N}} \\
&= \frac{e^{-2\beta E_1}}{\sum_i e^{-2\beta E_i}} \\
&= \left( 1 + \sum_{i\neq 1} {e^{-2\beta (E_i - E_1)}} \right)^{-1}.
\end{align}
Since $E_i - E_1 = \Delta = \Omega(1/n^c)$, if $\beta$ is taken to be $\Omega(n^d)$ with $d\geq c+1$, we have that $e^{-2\beta (E_i - E_1)} \leq \exp[-2n^{d-c}]$.
This means that the overlap is at least $1/(1 + \exp[n\ \log\ 2 - 2n^{d-c}]) \geq 1 - \exp[n\ \log\ 2 - 2n^{d-c}]$, which means that the trace distance between the normalized states is $\varepsilon = O(\exp[-n^{d-c}])$.
Therefore, the choice $\beta = \Theta(n\Delta^{-1})$ suffices to ensure that the resulting (normalized) state $\rho_\beta$ is exponentially close to the ground state.
Therefore, the thermal expectation value of any local observable $A$ with polynomially bounded spectral norm is also exponentially close to the ground-state expectation value $\bra{E_1}A\ket{E_1}$.
\end{proof}

We now show the following lemma about computing unnormalized thermal expectation values, which is the core subroutine of our $\P^\PP$ algorithm.
\begin{lemma}
A $\PP$ algorithm can decide whether $\Tr [e^{-2\beta H} A] \leq a'$ or $\geq b'$ when input a sparse Hamiltonian $H$, a number $\beta \geq 0$, a local observable $A$, and thresholds $a'$ and $b'$. \label{lem_pppalg}
\end{lemma}
\begin{proof}
We express the unnormalized thermal expectation value as a sum over several paths as follows:
\begin{align}
\mathcal{A}_\beta &= \Tr[e^{-2\beta H} A] \label{eq_thermalexpectation} \\
=& \sum_{x,y} \bra{x} e^{-2\beta H}\ketbra{y} A\ket{x} \\
\approx & \sum_{x,y} \bra{x}\left( \mathds{1} - 2\beta H + 2(\beta H)^2 + \ldots \frac{(-2\beta H)^K}{K!} \right) \ket{y} \times\nonumber
\\& \bra{y} A\ket{x} =: \mathcal{A'}_\beta \label{eq_taylor}\\
=& \sum_{k=0}^K \frac{1}{k!} \sum_{\substack{x_0, x_1, \ldots x_k}} \bra{x_0} -2\beta H \ketbra{x_1} -2\beta H \ket{x_2} \ldots \nonumber
\\&\bra{x_{k-1}}-2\beta  H \ketbra{x_k} A \ket{x_0}. \label{eq_pppaths} \\
=& \sum_{i=1}^T f_i.
\end{align}
This expression is reminiscent of a Euclidean path integral, although there are some differences.
In a Euclidean path integral, one Trotterizes the map $\exp(-\beta H) \approx \left(\prod_i\exp(-\beta H_i/r)\right)^r$ and use the fact that each term of the Hamiltonian $H_i$ is local in order to compute terms in the series.
In contrast, here we have used the Taylor expansion for $\exp(-\beta H)$ and have inserted resolutions of the identity in order to compute the terms $\bra{x}H^k\ket{y}$.
Using the Taylor series allows us to get exponentially small additive error, which is not guaranteed by Trotterization.

Before we move on, let us analyze the additive error in \cref{eq_taylor}.
It is given by:
\begin{align}
 \epsilon \leq  \frac{(2\beta \norm{H})^{K+1}}{(K+1)!} \norm{A} \times O(1).
\end{align}
By choosing $K > 2\beta e \norm{H} + f(n)$ for some polynomial $f(n) = O(\beta \norm{H}/n)$ and $f(n) = \Omega(n)$, we can ensure that the error is bounded above by $\norm{A}\exp[-f(n)]$:
\begin{align}
& K+ 1 \geq 2\beta e \norm{H} + f(n) \\
\implies & (K+ 1)\log (K+1) \geq (K+1) \log (2\beta e \norm{H}) + \nonumber
\\&(K+1) \log\left(1 + \frac{f(n)}{2\beta e \norm{H}}\right) \\
\geq & (K+1) \log (2\beta e \norm{H}) + (K+1) \frac{f(n)}{2\beta e \norm{H}} -\nonumber
\\& \frac{K+1}{2} \left(\frac{f(n)}{2\beta e \norm{H}}\right)^2 \\
\geq & (K+1) \log (2\beta e \norm{H}) + \Omega(f(n)),
\end{align}
where we have used the fact that $\log (1+x) \geq x - \frac{x^2}{2}$ for small $x$ and that $f(n)= o(\beta \norm{H})$.
Therefore,
\begin{align}
\log \left( \frac{(2\beta \norm{H})^{k+1}}{(K+1)!} \right) & \leq -\Omega(f(n)), \ \mathrm{giving} \\
\epsilon & \leq O\left({\norm{A}} \exp[-f(n)]\right).
\end{align}
\end{proof}

\begin{proof}[Proof of \cref{lem_gslo_ppp}]
From \cref{lem_thermal}, we know that the normalized state is exponentially close to the true ground state.
Therefore, deciding whether the ground state has $\Tr[\ketbra{\Psi} A] \leq a$ or $\geq b$ is equivalent to deciding whether the unnormalized state has expectation value $\Tr[\rho' A] \leq a' = \mathcal{N}_\mathrm{est} (a + \norm{A} \varepsilon)$ or $\geq b' = \mathcal{N}_\mathrm{est} (b - \norm{A} \varepsilon)$, where $\mathcal{N}_\mathrm{est}$ is an estimate of the normalization of the state and $\varepsilon$ the trace distance between the ground state and the thermal state.
To maintain a gap between the YES and NO cases, we need $\varepsilon < 2^{-u(n)}/\norm{A}$ for some polynomial $u$, which can be satisfied by taking $n^{d- c}$ in \cref{lem_thermal} to be $\geq u(n) + \log \norm{A}$.
The norm of $A$ is bounded above by a polynomial in $n$ and therefore is a subleading term.

Since the thresholds $a'$ and $b'$ depend on the normalization, we should compute the normalization
$\mathcal{N}$ beforehand.
Since the normalization is a special case of \cref{eq_thermalexpectation} with $A = \mathds{1}$, we can use the $\PP$ procedure to decide if $\mathcal{N} \leq a_1$ or $\mathcal{N} \geq a_2$ for some $a_1,a_2$ with $a_2-a_1 = \Omega(1/\exp)$.
Performing binary search over the interval $(0,1]$ with polynomially many queries to the $\PP$ oracle, we can estimate the normalization to exponentially small additive error, giving an estimate $\mathcal{N}_\mathrm{est}$.

Therefore, we have shown that a $\P^\PP$ machine can do all the above: compute the normalization and then compute the thermal expectation value for a low-temperature state.
Since we have also shown that setting $\beta = (n/\Delta)$ suffices to get exponentially small error, we have shown that the problem is in $\P^\PP$.
\end{proof}

This technique is also applicable to Hamiltonians or Hermitian operators that are not necessarily local, or even sparse.
For example, it can apply to Hermitian operators of the kind in \cref{lem_cooling_othermatrices}.

\section{Turing machine construction for \PSPACE-hardness}
In this section, we complete the proof of \cref{lem_sparseHamPspaceHard}.
\begin{lemma}[Lower bound on spectral gap for \PSPACE-hard construction] \label{lem_expsmallgapNOcase}
In the NO case, the construction in the proof of \cref{lem_sparseHamPspaceHard} has a spectral gap of $\Theta(\ell_\mathrm{max}^{-2})$, where $\ell_\mathrm{max}$ is the number of vertices in the largest subgraph of $G_x'$.
\end{lemma}
\begin{proof}
\begin{figure*} \centering
\begin{subfigure}{0.49\textwidth}
 \includegraphics[width=\linewidth]{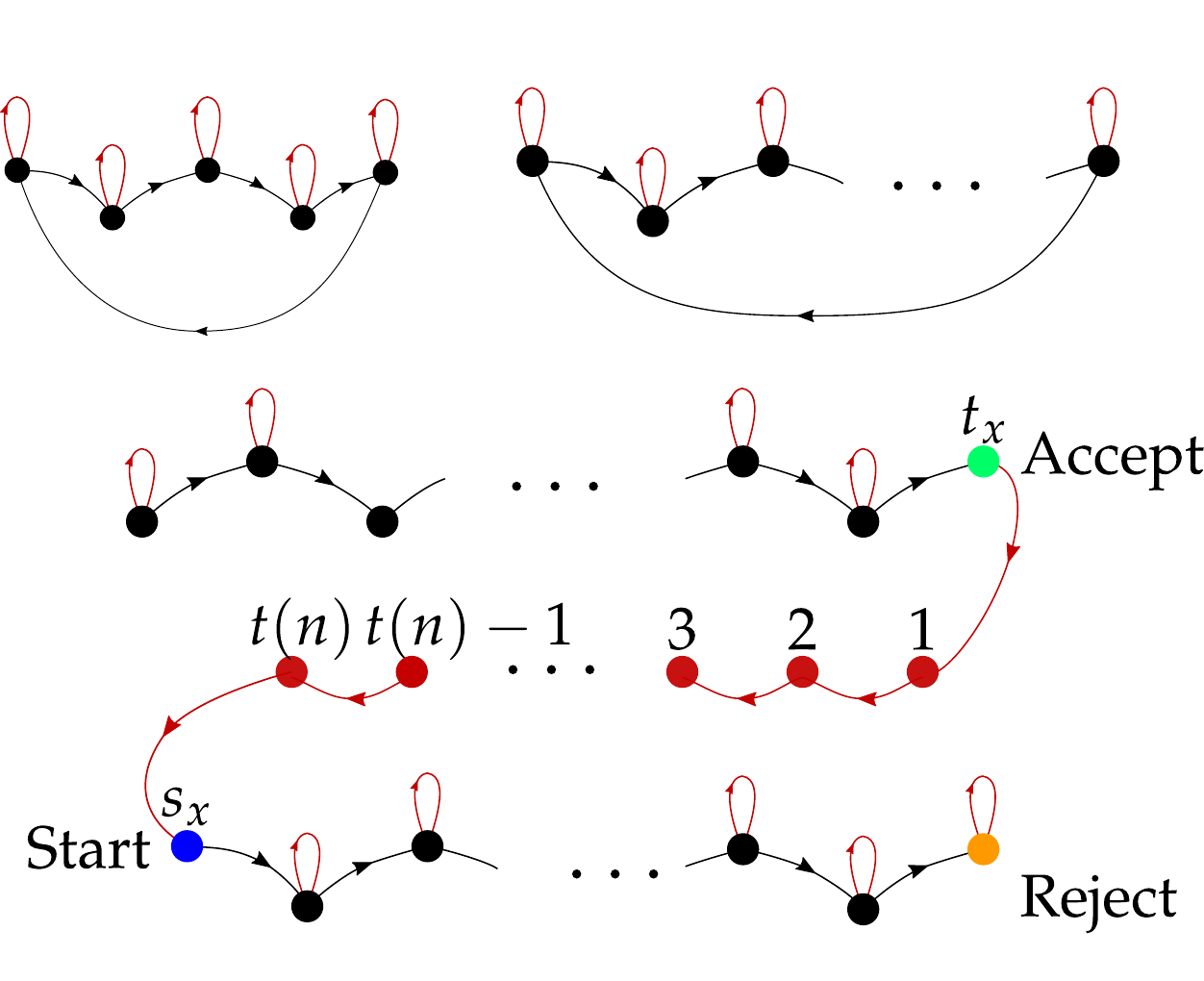}
 \caption{$A_x'$} \label{fig_pspacegraph_mno_apx}
\end{subfigure} \hfill
\begin{subfigure}{0.49\textwidth}
 \includegraphics[width=\linewidth]{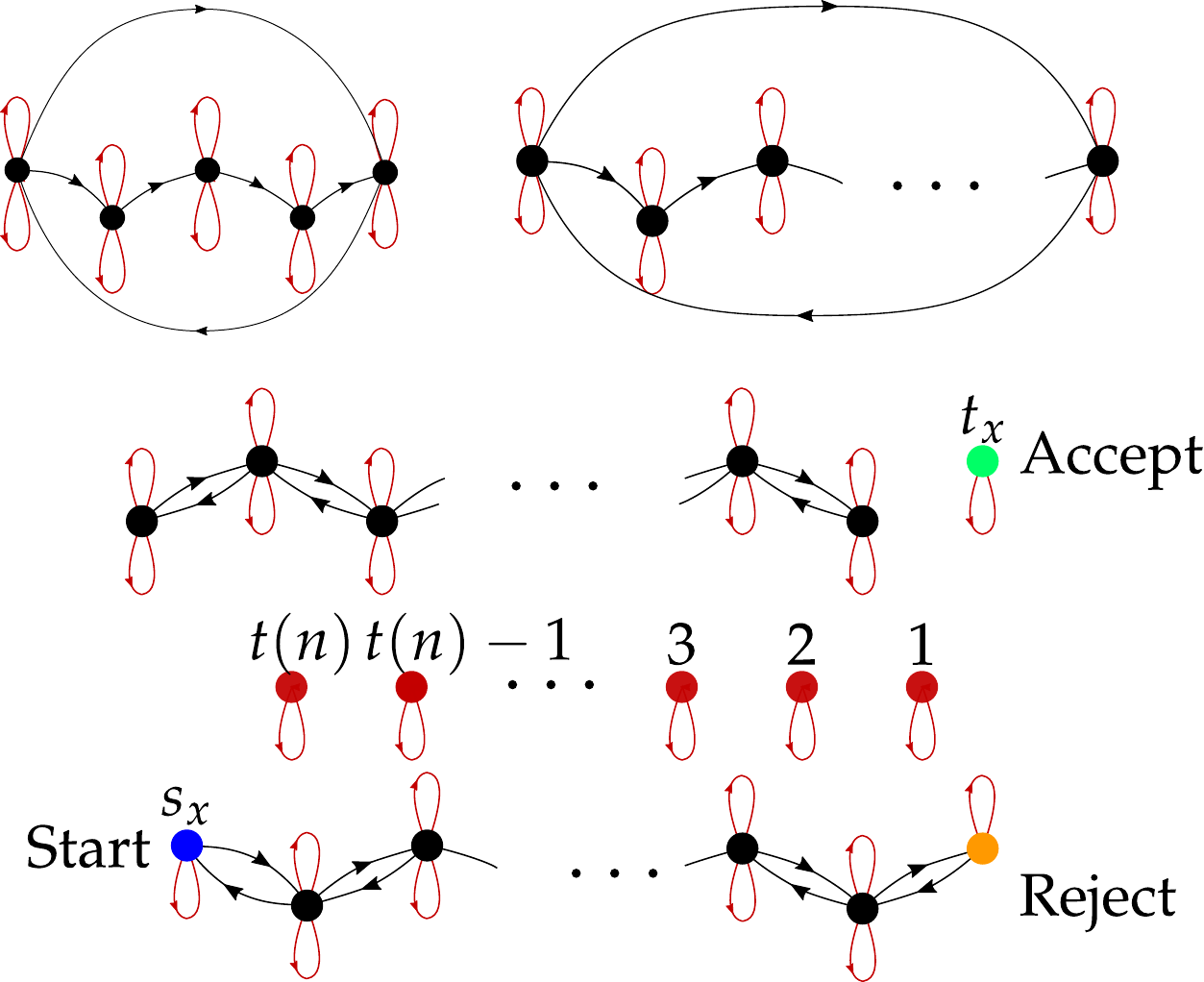} 
 \caption{${A_x^\dag}'A_x$.} \label{fig_pspacegraph_mno_AtA}
\end{subfigure}
\caption{(a) Graph $G_x'$ with adjacency matrix $A_x'$, adapted from \cref{fig_pspacegraph_mno}.
(b) Graph with (weighted, directed) adjacency matrix ${A_x^\dag}'A_x$.
Vertices with two self-loops can be thought of as a single self-loop with weight 2.}
\end{figure*}

Recall the form of the graph $G_x'$ in the NO case, reproduced here in \cref{fig_pspacegraph_mno_apx}.
We first restrict our attention to the subgraph of $G_x'$ containing the start and accept configurations.
The matrix ${A_x^\dag}' A_x'$, when restricted to this subspace, is further composed of three subspaces, each corresponding to a subgraph, as shown in \cref{fig_pspacegraph_mno_AtA}.
We write ${A_x^\dag}' A_x' = \mathcal{G}_1 \oplus \mathcal{G}_2 \oplus \mathcal{G}_3$.
The block $\mathcal{G}_1$ corresponds to the vertices leading to $t_x$ (not including $t_x$).
The block $\mathcal{G}_2$ corresponds to the vertices $\{t_x\} \cup \{1, \ldots t(n)\}$.
Lastly, $\mathcal{G}_3$ is the block with the vertices starting from $s_x$ and leading to the reject state, which are the configurations visited by the Turing machine.
We have
\begin{align}
 \mathcal{G}_1 &= 
\begin{pmatrix}
  2  & 1 & \\
  1 & 2 & 1 \\
    & 1 & 2 & \ddots \\
    &   & \ddots& \ddots & 1 \\
    &   &       & 1 & 2 & \\
\end{pmatrix}_{\ell_1 \times \ell_1},\ 
\mathcal{G}_2 = \mathds{1}_{\ell_2 \times \ell_2},\ \mathrm{and\ }\nonumber
\\ \mathcal{G}_3 &=
\begin{pmatrix}
1 & 1 \\
       1 & 2 & 1\\
         & 1 & 2 & \ddots\\
         &   & \ddots & \ddots & 1\\
         &   &     & 1    & 2 & 1 \\
         &   &     &      & 1 & 1
\end{pmatrix}_{\ell_3 \times \ell_3}.
\end{align}
It may be seen that there is a zero eigenvector $(0,0,\ldots,0,1,-1,1,\ldots (-1)^{\ell_3})^T$, with the zeros corresponding to the subspaces $\mathcal{G}_1$ and $\mathcal{G}_2$.
We now lower-bound the next-smallest eigenvalue.
Let
\begin{align}
r_n(\lambda) &:= \det [\mathcal{G}_3 - \lambda \mathds{1}_{n}] \nonumber
\\&=\det \begin{pmatrix}
1 -\lambda & 1 \\
       1 & 2 -\lambda & 1\\
         & 1 & 2-\lambda & \ddots\\
         &   & \ddots & \ddots & 1\\
         &   &     & 1    & 2-\lambda & 1 \\
         &   &     &      & 1 & 1 -\lambda
\end{pmatrix}_{n \times n} \\
p_n(\lambda) &:= \det
\begin{pmatrix}
2 -\lambda & 1 \\
       1 & 2 -\lambda & 1\\
         & 1 & 2-\lambda & \ddots\\
         &   & \ddots & \ddots & 1\\
         &   &     & 1    & 2-\lambda & 1 \\
         &   &     &      & 1 & 1 -\lambda
\end{pmatrix}_{n\times n} \label{eq_characteristic}.
\end{align}
The polynomial $p_n(\lambda)$ can be computed exactly \cite{Fefferman2016c}, and is given by $p_n(2-2\cos \theta) = \frac{\sin((n+1)\theta)-\sin(n\theta)}{\sin \theta} = \frac{\cos((n+\frac{1}{2})\theta)}{\cos(\frac{\theta}{2})}$.
We can obtain $r_n(\lambda)$ in terms of $p_n(\lambda)$: $r_n(\lambda) = (1-\lambda)p_{n-1}(\lambda)-p_{n-2}(\lambda)$, giving us
\begin{align}
 r_n(\lambda) &= f_n(\theta) \nonumber
\\&= (2\cos \theta -1)\frac{\cos((n-\frac{1}{2})\theta)}{\cos(\frac{\theta}{2})} - \frac{\cos((n-\frac{3}{2})\theta)}{\cos(\frac{\theta}{2})},
\end{align}
where $\theta = \cos^{-1}(1-\frac{\lambda}{2})$, or $\lambda = 2 - 2\cos \theta$.
The eigenvalues of $\mathcal{G}_3$ are related to the roots of the characteristic polynomial $f_n(\theta)=0$.
We can see that $\theta= 0$ is always a root of the polynomial, giving us the zero eigenvalue ($\lambda = 2 - 2\cos \theta = 0$) for the NO case.

Now, it remains to be shown that the next smallest eigenvalue is bounded away from zero.
First consider $\mathcal{G}_1$, whose eigenvalues are the roots of the characteristic equation $\det [\mathcal{G}_1 -\lambda \mathds{1}_{\ell_1}]$.
The eigenvalues of $\mathcal{G}_1$ can be computed in a similar fashion to those of $\mathcal{G}_3$ and are given by $4\sin^2\left( \frac{k\pi}{2(\ell_1+1)} \right), k \in [n]$.
The smallest eigenvalue of $\mathcal{G}_1$ is therefore at least $\Omega\left(1/{\ell_1}^2\right)$.
It is also easily seen that $\mathcal{G}_2 \succ 0$.

We now come to $\mathcal{G}_3$.
As we have seen, $\mathcal{G}_3$ has a zero eigenvalue.
In order to show a spectral gap for $\mathcal{G}_3$, we show that the next root of the polynomial $f_{\ell_3}(\theta)$ must occur at least a distance $\Omega({\ell_3}^{-2})$ away.
The roots of $\mathcal{G}_3$ are given by \cite{Dooley2020}
\begin{align}
\lambda_j = 2 + 2\cos\left( \frac{\pi j}{\ell_3} \right), \quad j \in [\ell_3].
\end{align}
Setting $j=\ell_3$ gives the zero eigenvalue and $j=\ell_3-1$ the first nonzero eigenvalue.
The spectral gap of $\mathcal{G}_3$ is therefore
\begin{align}
 \lambda_{\ell_3-1} &= 2 - 2\cos\left( \frac{\pi}{\ell_3} \right)
 \\ &= 4 \sin^2 \left( \frac{\pi}{2\ell_3} \right)
 \\ & \geq \frac{\pi^2}{\ell_3^2} - O\left( \frac{\pi^4}{\ell_3^4} \right)
 \\ &= \Omega\left({1}/{\ell_3^2} \right).
\end{align}

Finally, we consider other subgraphs that do not contain the start vertex.
Just like the analysis of the YES case, the eigenvalues for these are bounded away from 0 by $\ell^{-2}$, where $\ell$ is the number of vertices in the subgraph.
We have therefore lower bounded the value of the nonzero eigenvalue in each case, showing that the spectral gap is $\Omega(\ell_\mathrm{max}^{-2}) = \Omega(2^{-\poly})$.
\end{proof}
\section{Complexity of $\pPGQMA$ and $\pEGQMA$ with asymmetric spectral gaps} \label{sec_asymcomplexity}
We show here that the promise of asymmetric spectral gaps does not change the complexity class for both $\pPGQMA$ and $\pEGQMA$, proving \cref{thm_asym_symgaps}.

\begin{proof}[Proof of \cref{thm_asym_symgaps}]
It is easy to see that $\GQMA[c,s,g_1,g_2] \subseteq \GQMA[c,s,g_1,0]$ simply by ignoring the promise on the NO instance.
It remains to show that the same upper bounds as the symmetric case hold for the asymmetric case too.
For the case of $c-s=\Omega(1/\exp)$, $g_2 = \Omega(1/\exp)$, we observe that one can also ignore the promise on the YES instance and obtain containment in $\pQMA=\PSPACE$, which equals $\pEGQMA$.

It remains to give an upper bound for the class $\cup_{\substack{c-s \geq \Omega(1/\exp) \\ g_1 \geq \Omega(1/\poly)}} \GQMA[c,s,g_1,0]$.
We give a $\PP$ algorithm for any instance from this class, which implies equivalence of the two classes.

We are given a description of a circuit, with the promise that the YES case has $\Omega(1/\poly)$ spectral gap for the accept operator $Q$.
We want to decide if $\lambda_1(Q)$ is $\geq c$ (YES) or $\leq s$ (NO).
The overall $\PP$ algorithm is as follows.
\begin{enumerate}
\item Use the $\P^{\QMA{[\log]}}$ algorithm of Ambainis \cite{Ambainis2014} to determine whether an instance has spectral gap $\Delta$ $\geq g_1$ (YES) or $\leq g_1/2$ (NO), for $g_1 = \Omega(1/\poly)$.
\item If the spectral gap is $g_1$ or larger, run the algorithm in \cref{lem_cooling_othermatrices} with Hamiltonian $\mathds{1}-Q$ and accept or reject according to the answer returned by the algorithm.
\item Otherwise reject.
\end{enumerate}
We claim that the algorithm of Ambainis works not just for local Hamiltonians, but also for accept operators like $Q$.
This is because the $\QMA$ queries in Ambainis's algorithm pertain to whether the ground-state energy (or the minimum eigenvalue $1-\lambda_1$ in this case) is smaller or larger than a threshold.
A $\QMA$ verifier can compute the eigenvalue of the accept operator given an eigenstate, using phase estimation.
Therefore, all queries to the oracle about $1-\lambda_1$ are still valid $\QMA$ queries.
Also, the final query in Ambainis's algorithm is for the operator $(\mathds{1}-Q)\otimes \mathds{1} +\mathds{1} \otimes (\mathds{1}-Q)$ on two registers, restricted to the antisymmetric subspace.
Since a $\QMA$ verifier can also perform a projection onto the antisymmetric subspace, Ambainis's algorithm (i.e.\ the first step) works to estimate the spectral gap of $Q$ in $\P^{\QMA{[\log]}}$.

Now, since $\P^{\QMA{[\log]}} \subseteq \PP$ \cite{Gharibian2019b}, the overall algorithm is a valid $\PP$ algorithm, since the two queries can be made in parallel.
To see the correctness, we see that if the instance has a YES answer, then it has a spectral gap of at least $g_1$ by virtue of the promise.
In this case the spectral gap algorithm would return YES.
This ensures that the $\PP$ algorithm in \cref{lem_cooling_othermatrices} works correctly and returns the correct answer $E_1 \leq a$ (YES) or $E_1 \geq b$ (NO).
The algorithm outputs YES since the instance has low energy.

In the NO case, there may or may not be a spectral gap.
If the spectral gap $\Delta \leq g_1/2$ is not large enough, the spectral gap algorithm returns NO.
We reject in this case.
If the spectral gap algorithm returns YES, then the spectral gap is at least $\Delta \geq g_1/2$ (this includes the cases when the spectral gap is in the window $[g_1/2,g_1]$, which is outside of the promise in the spectral gap algorithm).
This means that the algorithm in \cref{sec_ppdetails} will work, and return the correct output (NO).
Therefore, we see that $\cup_{\substack{c-s \geq \Omega(1/\exp) \\ g_1 \geq \Omega(1/\poly)}} \GQMA[c,s,g_1,0] = \pPGQMA$.
\end{proof}

We remark that it can be seen that \textsc{LocalHamiltonian}[$a,b,g_1,0$] with $b-a= \Theta(1/\exp)$ is $\pPGQMA$-complete when the spectral gap $g_1$ is $1/\poly$ and $\pEGQMA$-complete when $g_1$ is $1/\exp$.

\begin{acknowledgments}
We thank Sevag Gharibian, Alex Grilo, Hosho Katsura, Cedric Lin, Yupan Liu, and Zachary Remscrim for useful feedback on the manuscript.
We thank an anonymous referee of an earlier version of this manuscript for pointing out the similarity of imaginary-time evolution with the power method and Cedric Lin for showing that the power method can give a $\PP$ algorithm as opposed to a $\P^\PP$ algorithm.
We also thank Zachary Remscrim for suggesting the proof of \cref{lem_pegqcma}, and Peter Love, Yuan Su, Minh Tran, and Seth Whitsitt for helpful discussions.
A.\,D.\ and A.\,V.\,G.\ were supported in part by the DoE ASCR Quantum Testbed Pathfinder program (award No. DE-SC0019040), DoE QSA, NSF QLCI, U.S. Department of Energy Award No. DE-SC0019449, DoE ASCR Accelerated Research in Quantum Computing program (award No. DE-SC0020312), NSF PFCQC program, AFOSR, ARO MURI, AFOSR MURI, and DARPA SAVaNT ADVENT.
A.\,D.\ also acknowledges support from NSF RAISE/TAQS 1839204.
B.\,F.\ acknowledges support from AFOSR (YIP number FA9550-18-1-0148 and FA9550-21-1-0008). This material is based upon work partially supported by the National Science Foundation under Grant CCF- 2044923 (CAREER) and by the U.S. Department of Energy, Office of Science, National Quantum Information Science Research Centers.
The Institute for Quantum Information and Matter is an NSF Physics Frontiers Center PHY-1733907.
\end{acknowledgments}

\end{document}